\def\llncs{0}
\def\fullpage{1}
\def\anonymous{0}
\def\authnote{0}
\def\notxfont{0}
\def\submission{0}
\def\reply{0}
\def\cameraready{0}
\def\noaux{0}
\ifnum\submission=1
\def\anonymous{1}
\def\llncs{1}
\else
\fi

\ifnum\cameraready=1
\def\llncs{1}
\def\anonymous{0}
\def\authnote{0}
\else
\fi

\ifnum\anonymous=1
\def\authnote{0}
\else
\fi

\ifnum\llncs=1
	\documentclass[envcountsect,a4paper,runningheads]{llncs}
\else
	\documentclass[letterpaper,hmargin=1.05in,vmargin=1.05in]{article}
			\ifnum\fullpage=1
		\usepackage{fullpage}
		\fi
\fi

\ifnum\reply=1
\usepackage{ulem}
\renewcommand{\emph}{\textit}
\else
\fi

\usepackage[%
  colorlinks=true,
  citecolor=darkgreen,
  linkcolor=darkblue,
  urlcolor=darkblue,
  pagebackref=true
]{hyperref}

\usepackage{amsmath, amsfonts, amssymb, mathtools,amscd}

\usepackage{amsthm}

\usepackage{lmodern}
\usepackage[T1]{fontenc}
\usepackage[utf8]{inputenc}

\usepackage{etex}

\usepackage{arydshln} 
\usepackage{url}
\usepackage{ifthen}
\usepackage{bm}
\usepackage{multirow}
\usepackage[dvips]{graphicx}
\usepackage[usenames]{color}
\usepackage{xcolor,colortbl} 
\usepackage{threeparttable}
\usepackage{comment}
\usepackage{paralist,verbatim}
\usepackage{cases}
\usepackage{booktabs}
\usepackage{braket}
\usepackage{cancel} 
\usepackage{ascmac} 
\usepackage{framed}
\usepackage{authblk}
\usepackage{pifont}
\usepackage{physics}
\definecolor{darkblue}{rgb}{0,0,0.8}
\definecolor{darkgreen}{rgb}{0,0.5,0}
\definecolor{maroon}{rgb}{0.5,0.1,0.1}
\definecolor{dpurple}{rgb}{0.2,0,0.65}
\definecolor{chocolate}{rgb}{0.8,0.4,0.1}

\usepackage[capitalise,noabbrev]{cleveref}
\usepackage[absolute]{textpos}
\usepackage[final]{microtype}
\usepackage[absolute]{textpos}
\usepackage{everypage}
\usepackage{autonum}

\usepackage{dsfont}
\DeclareMathAlphabet{\mathpzc}{OT1}{pzc}{m}{it}

\ifnum\submission=1
\renewcommand*{\backref}[1]{}
\def\notxfont{1}
\else
\fi

\ifnum\llncs=1
\renewcommand{\subparagraph}{\paragraph}
\else
\ifnum\notxfont=1
\else
\usepackage{mathpazo}
\usepackage{newtxtext}
\usepackage{helvet}
\fi
\fi

\newtheoremstyle{thicktheorem}%
{\topsep}
{\topsep}
{\itshape}{}%
{\bfseries}%
{.}
{ }%
{\thmname{#1}\thmnumber{ #2}%
		\thmnote{ (#3)}%
}

\newtheoremstyle{remark}
{\topsep}
{\topsep}
	{}
	{}
	{}
	{.}
	{ }
	{\textit{\thmname{#1}}\thmnumber{ #2}
			\thmnote{ (#3)}%
	}

\ifnum\llncs=0
	\theoremstyle{thicktheorem}
	\newtheorem{theorem}{Theorem}[section]
	\newtheorem{lemma}[theorem]{Lemma}
	\newtheorem{corollary}[theorem]{Corollary}
	\newtheorem{proposition}[theorem]{Proposition}
	\newtheorem{definition}[theorem]{Definition}

	\theoremstyle{remark}
	
	\newtheorem{remark}[theorem]{Remark}

\fi

	\crefname{theorem}{Theorem}{Theorems}
	\crefname{assumption}{Assumption}{Assumptions}
	\crefname{construction}{Construction}{Constructions}
	\crefname{corollary}{Corollary}{Corollaries}
	\crefname{conjecture}{Conjecture}{Conjectures}
	\crefname{definition}{Definition}{Definitions}
	\crefname{exmaple}{Example}{Examples}
	\crefname{experiment}{Experiment}{Experiments}
	\crefname{counterexample}{Counterexample}{Counterexamples}
	\crefname{lemma}{Lemma}{Lemmata}
	\crefname{observation}{Observation}{Observations}
	\crefname{proposition}{Proposition}{Propositions}
	\crefname{remark}{Remark}{Remarks}
	\crefname{claim}{Claim}{Claims}
	\crefname{fact}{Fact}{Facts}
	\crefname{note}{Note}{Notes}

\ifnum\llncs=1
 \crefname{appendix}{App.}{Appendices}
 \crefname{section}{Sec.}{Sections}
\else
\fi

\ifnum\llncs=1
\renewcommand*{\backref}[1]{}
\else
	\renewcommand*{\backref}[1]{(Cited on page~#1.)}
	\ifnum\notxfont=1
	\else
		\usepackage{newtxtext}
	\fi
\fi

\newcommand*{\keys}[1]{\mathsf{#1}}

\newcommand{\Oracle}[1]{O_{\mathtt{#1}}}

\newcommand*{\algo}[1]{\ensuremath{\mathsf{#1}}}
\newcommand*{\qalgo}[1]{\ensuremath{\mathpzc{#1}}}
\newcommand*{\qstate}[1]{\mathpzc{#1}}

\newcommand*{\qregidx}[2]{{\color{gray}{\mathsf{#1}_{#2}}}}

\newcommand*{\entity}[1]{\mathcal{#1}}

\newcounter{expitem}
\newcommand\expitem[2]{\refstepcounter{expitem}%
   \arabic{expitem}.&#1&\arabic{expitem}.&#2\\%
}
\newcommand\Heading[2]{\multicolumn2{c}{#1}&\multicolumn2{c}{#2}\\}

\usepackage{mathtools}

\newcommand{\chosen}{\leftarrow}

\newcommand{\lrun}{\leftarrow}

\newcommand{\la}{\leftarrow}
\newcommand{\ra}{\rightarrow}
\renewcommand{\gets}{\leftarrow}

\newcommand{\seteq}{\coloneqq}

\newcommand{\tensor}{\otimes}
\newcommand{\concat}{\|}

\newcommand{\setbk}[1]{\{#1\}}

\newcommand{\cC}{\mathcal{C}}
\newcommand{\cD}{\mathcal{D}}
\newcommand{\cE}{\mathcal{E}}
\newcommand{\cF}{\mathcal{F}}

\newcommand{\cI}{\mathcal{I}}

\newcommand{\cM}{\mathcal{M}}
\newcommand{\cN}{\mathcal{N}}
\newcommand{\cO}{\mathcal{O}}
\newcommand{\cP}{\mathcal{P}}

\newcommand{\cR}{\mathcal{R}}

\newcommand{\cU}{\mathcal{U}}

\newcommand{\cX}{\mathcal{X}}
\newcommand{\cY}{\mathcal{Y}}

\newcommand{\qA}{\qalgo{A}}
\newcommand{\qB}{\qalgo{B}}
\newcommand{\qC}{\qalgo{C}}
\newcommand{\qD}{\qalgo{D}}
\newcommand{\qG}{\qalgo{G}}
\newcommand{\qS}{\qalgo{S}}

\def\makeuppercase#1{
\expandafter\newcommand\csname sf#1\endcsname{\mathsf{#1}}
\expandafter\newcommand\csname frak#1\endcsname{\mathfrak{#1}}
\expandafter\newcommand\csname bb#1\endcsname{\mathbb{#1}}
\expandafter\newcommand\csname bf#1\endcsname{\textbf{#1}}
}

\def\makelowercase#1{
\expandafter\newcommand\csname frak#1\endcsname{\mathfrak{#1}}
\expandafter\newcommand\csname bf#1\endcsname{\textbf{#1}}
}

\newcounter{char}
\loop
   \edef\letter{\alph{char}}
   \edef\Letter{\Alph{char}}
   \expandafter\makelowercase\letter
   \expandafter\makeuppercase\Letter
   \stepcounter{char}
   \unless\ifnum\thechar>26
\repeat

\def\makeuppercase#1{
\expandafter\newcommand\csname tl#1\endcsname{\widetilde{#1}}
}

\def\makelowercase#1{
\expandafter\newcommand\csname tl#1\endcsname{\widetilde{#1}}
}

\newcommand{\R}{\mathbb{R}}

\newcommand{\bit}{\{0,1\}}

\newcommand{\Fs}{\mathcal{F}}

\newcommand{\Ys}{\mathcal{Y}}
\newcommand{\Xs}{\mathcal{X}}

\newcommand{\randspace}{\mathcal{R}}

\newcommand{\secp}{\lambda}
\newcommand{\sep}{\lambda}
\newcommand{\coin}{\keys{coin}}

\newcommand{\cert}{\keys{cert}}

\newcommand{\state}{\mathsf{st}}

\newcommand{\A}{\entity{A}}

\newcommand{\realgame}[2]{\mathsf{Real}^{\mathsf{#1}\textrm{-}\mathsf{#2}}}
\newcommand{\simgame}[2]{\mathsf{Sim}^{\mathsf{#1}\textrm{-}\mathsf{#2}}}

\newcommand{\adva}[2]{\mathsf{Adv}_{#1}^{\mathsf{#2}}}
\newcommand{\advb}[3]{\mathsf{Adv}_{#1}^{\mathsf{#2} \mbox{-} \mathsf{#3}}}
\newcommand{\advc}[4]{\mathsf{Adv}_{#1}^{\mathsf{#2} \mbox{-} \mathsf{#3} \mbox{-} \mathsf{#4}}}

\newcommand{\expt}[2]{\mathsf{Expt}_{#1}^{\mathsf{#2}}}
\newcommand{\expa}[2]{\mathsf{Expt}_{#1}^{\mathsf{#2}}}
\newcommand{\expb}[3]{\mathsf{Exp}_{#1}^{ \mathsf{#2} \mbox{-} \mathsf{#3}}}
\newcommand{\expc}[4]{\mathsf{Exp}_{#1}^{ \mathsf{#2} \mbox{-} \mathsf{#3} \mbox{-} \mathsf{#4}}}

\newcommand{\hybi}[1]{\mathsf{Hyb}_{#1}}
\newcommand{\hybij}[2]{\mathsf{Hyb}_{#1}^{#2}}

\newcommand*{\pk}{\keys{pk}}
\newcommand*{\sk}{\keys{sk}}
\newcommand*{\dk}{\keys{dk}}
\newcommand*{\ek}{\keys{ek}}
\newcommand*{\vk}{\keys{vk}}

\newcommand*{\fsk}{\keys{fsk}}

\newcommand*{\msk}{\keys{msk}}
\newcommand*{\mpk}{\keys{mpk}}

\newcommand*{\MPK}{\keys{MPK}}
\newcommand*{\MSK}{\keys{MSK}}

\newcommand*{\ct}{\keys{ct}}

\newcommand*{\msg}{\keys{m}}

\newcommand{\prfkey}{\keys{K}}

\newcommand{\qsk}{\qstate{sk}}

\newcommand{\qdk}{\qstate{dk}}

\newcommand{\qfsk}{\qstate{fsk}}
\newcommand{\qusk}{\qstate{usk}}

\newenvironment{boxfig}[2]{\begin{figure}[#1]\fbox{\begin{minipage}{0.97\linewidth}
                        \vspace{0.2em}
                        \makebox[0.025\linewidth]{}
                        \begin{minipage}{0.95\linewidth}
            {{
                        #2 }}
                        \end{minipage}
                        \vspace{0.2em}
                        \end{minipage}}
                        }
                        {\end{figure}}

\newcommand{\pprotocol}[4]{
\begin{boxfig}{h!}{\footnotesize 
\centering{\textbf{#1}}
    #4
\vspace{0.2em} } \caption{\label{#3} #2}
\end{boxfig}
}

\newcommand{\protocol}[4]{
\pprotocol{#1}{#2}{#3}{#4} }

\newcommand{\prf}{\algo{F}}

\newcommand{\Good}{\mathtt{Good}}

\newcommand{\Setup}{\algo{Setup}}

\newcommand{\Gen}{\algo{Gen}}

\newcommand{\KeyGen}{\algo{KeyGen}}

\newcommand{\KG}{\algo{KG}}
\newcommand{\Enc}{\algo{Enc}}
\newcommand{\Dec}{\algo{Dec}}

\newcommand{\enc}{\algo{Enc}}

\newcommand{\Vrfy}{\algo{Vrfy}}

\newcommand{\certvrfy}{\algo{Vrfy}}

\newcommand{\qKG}{\qalgo{KG}}

\newcommand{\qDec}{\qalgo{Dec}}

\newcommand{\qCert}{\qalgo{Cert}}

\newcommand{\qSEnc}{\qalgo{SEnc}}
\newcommand{\qSKG}{\qalgo{SKG}}
\newcommand{\qVrfy}{\qalgo{Vrfy}}

\newcommand\PKE{\algo{PKE}}

\newcommand\ABE{\algo{ABE}}

\newcommand{\sfabe}{\mathsf{abe}}

\newcommand\FE{\algo{FE}}

\newcommand\PKFE{\algo{PKFE}}

\newcommand\SKFE{\algo{SKFE}}

\newcommand{\SKE}{\algo{SKE}}
\newcommand{\ske}{\algo{ske}}

\newcommand{\E}{\algo{E}}
\newcommand{\D}{\algo{D}}

\newcommand{\GC}{\algo{GC}}
\newcommand{\Garble}{\algo{Grbl}}
\newcommand{\SimGC}{\algo{Sim.GC}}

\newcommand{\GCEval}{\algo{GCEval}}

\newcommand{\lab}{\mathsf{lab}} 
\newcommand{\tildeE}{\widetilde{E}}

\newcommand{\PuncPRF}{\algo{PPRF}}

\newcommand{\PRF}{\algo{PRF}}

\newcommand{\prfgen}{\PRF.\Gen}

\newcommand{\Puncture}{\algo{Puncture}}

\newcommand{\negl}{{\mathsf{negl}}}

\newcommand{\poly}{{\mathrm{poly}}}

\newcommand{\zo}[1]{\{0,1\}^{#1}}
\newcommand{\bin}{\{0,1\}}

\newcommand{\class}[1]{\mathsf{#1}}

\newcommand{\Ppoly}{\class{P/poly}}

\newcommand{\calF}{\mathcal{F}}

\newcommand{\calQ}{\mathcal{Q}}
\newcommand{\calR}{\mathcal{R}}

\newcommand{\calY}{\mathcal{Y}}

\newcommand{\numkey}{q}
\newcommand{\numct}{n}
\newcommand{\numcti}[1]{n_{#1}}
\newcommand{\SKFESKL}{\algo{SKFE}\textrm{-}\algo{SKL}}
\newcommand{\ABESKL}{\algo{ABE}\textrm{-}\algo{SKL}}
\newcommand{\tlx}{\widetilde{x}}
\newcommand{\PKESKL}{\algo{SKL}}
\newcommand{\cKLA}{\textrm{CoIC-KLA}}

\newcommand{\owKLA}{\textrm{OW-KLA}}
\newcommand{\qst}{\qstate{st}}

\newcommand{\SDE}{\algo{SDE}}
\newcommand{\sde}{\algo{sde}}

\newcommand{\cPKE}{\algo{cPKE}}

\newcommand{\qExt}{\qalgo{Ext}}

\newcommand{\qaux}{\qstate{aux}}

\newcommand{\SKL}{\algo{SKL}}
\newcommand{\skl}{\keys{skl}}
\newcommand{\OWSKL}{\algo{OW}}
\newcommand{\INDSKL}{\algo{IND}}

\newcommand{\event}{\mathsf{E}}

\newcommand{\returned}{\top}
\newcommand{\unreturned}{\bot}
\newcommand{\requestchallenge}{\mathsf{RequestChallenge}}

\newcommand{\fe}{\mathsf{fe}}
\newcommand{\skfe}{\algo{skfe}}
\newcommand{\PKFESKL}{\mathsf{PKFE\textrm{-}SKL}}
\newcommand{\List}[1]{L_{\mathtt{#1}}}
\newcommand{\xpadlen}{\ell_{\mathsf{pad}}}
\newcommand{\skectlen}{\ell_{\mathsf{ske}}}
\newcommand{\barx}{\overline{x}}
\newcommand{\Vhyb}{T_{\mathsf{hyb}}}
\newcommand{\Vembed}{T_{\mathsf{emb}}}

\newcommand{\ADAINDKLA}{\textrm{Ada-IND-KLA}}
\newcommand{\SELINDKLA}{\textrm{Sel-IND-KLA}}
\newcommand{\sklpklen}{\ell_{\mathsf{ek}}}
\newcommand{\sklrandlen}{\ell_{\mathsf{rand}}}
\newcommand{\oneABESKL}{\algo{1ABE}}
\newcommand{\qABESKL}{\algo{qABE}}
\newcommand{\oneabeskl}{\algo{1abe}}

\newcommand{\oneABE}{\algo{1ABE}}
\newcommand{\qABE}{\algo{qABE}}
\newcommand{\oneabe}{\algo{1abe}}

\newcommand{\decision}{d}

\newcommand{\authornote}[3]{\textcolor{#3}{[\textsc{#1:} {#2}]}}
\ifnum\authnote=1
\newcommand{\shweta}[1]{\authornote{Shweta}{#1}{magenta}}
\newcommand{\fuyuki}[1]{\authornote{Fuyuki}{#1}{chocolate}}
\newcommand{\ryo}[1]{\authornote{Ryo}{#1}{darkblue}}
\newcommand{\shota}[1]{\authornote{Shota}{#1}{darkgreen}}
\newcommand{\takashi}[1]{\authornote{Takashi}{#1}{dpurple}}
\else
\newcommand{\shweta}[1]{}
\newcommand{\fuyuki}[1]{}
\newcommand{\ryo}[1]{}
\newcommand{\shota}[1]{}
\newcommand{\takashi}[1]{}
\fi

\ifnum\llncs=1
\let\oldvec\vec
\let\vec\oldvec
\makeatletter
\renewcommand*\l@author[2]{}
\renewcommand*\l@title[2]{}
\makeatletter
\fi    

\theoremstyle{remark}

\ifnum\submission=1
\title{
\textbf{Public Key Encryption with Secure Key Leasing}\thanks{{\color{red}{\emph{We attached the full version of this paper as a supplementary material}}}.}
}
\else
\title{
\textbf{Public Key Encryption with Secure Key Leasing}
}
\fi

\begin{document}

\ifnum\anonymous=1 
\ifnum\llncs=1
\author{\empty}\institute{\empty}
\else
\author{}
\fi
\else
%
%
\ifnum\llncs=1
\author{
	Shweta Agrawal\inst{1} \and Fuyuki Kitagawa\inst{2} \and Ryo Nishimaki\inst{2} \and Shota Yamada\inst{3} \and Takashi Yamakawa\inst{2}
}
\institute{
	IIT Madras, Chennai, India \and NTT Social Informatics Laboratories, Tokyo, Japan \and National Institute of Advanced Industrial Science and Technology (AIST), Tokyo, Japan
}
\else
%
%
\author[$\ast$]{Shweta Agrawal}
\author[$\dagger$]{\hskip 1em Fuyuki Kitagawa}
\author[$\dagger$]{\hskip 1em Ryo Nishimaki}
\author[$\ddagger$]{\\\hskip 1em Shota Yamada}
\author[$\dagger$]{\hskip 1em Takashi Yamakawa}
\affil[$\ast$]{{\small IIT Madras, Chennai, India}\authorcr{\small shweta.a@cse.iitm.ac.in}}
\affil[$\dagger$]{{\small NTT Social Informatics Laboratories, Tokyo, Japan}\authorcr{\small \{fuyuki.kitagawa.yh,ryo.nishimaki.zk,takashi.yamakawa.ga\}@hco.ntt.co.jp}}
\affil[$\ddagger$]{{\small National Institute of Advanced Industrial Science and Technology (AIST), Tokyo, Japan}\authorcr{\small yamada-shota@aist.go.jp}}
\renewcommand\Authands{, }
\fi 
\fi

\ifnum\llncs=1
\date{}
\else
\ifnum\anonymous=0
\date{\today}
\else
\date{}
\fi
\fi

\maketitle

\begin{abstract}
We introduce the notion of public key encryption with secure key leasing (PKE-SKL). Our notion supports the leasing of decryption keys so that a leased key achieves the decryption functionality but comes with the guarantee that if the quantum decryption key returned by a user passes a validity test, then the user has lost the ability to decrypt. Our notion is similar in spirit to the notion of secure software leasing (SSL) introduced by Ananth and La Placa (Eurocrypt 2021) but captures significantly more general adversarial strategies\footnote{In more detail, our adversary is not restricted to use an honest evaluation algorithm to run pirated software.}. Our results can be summarized as follows:

\begin{enumerate}
    \item \emph{Definitions: } We introduce the definition of PKE with secure key leasing and formalize a security notion that we call indistinguishability against key leasing attacks (IND-KLA security). We also define a one-wayness notion for PKE-SKL that we call OW-KLA security and show that an OW-KLA secure PKE-SKL scheme can be lifted to an IND-KLA secure one by using the (quantum) Goldreich-Levin lemma.
    
    \item \emph{Constructing IND-KLA PKE with Secure Key Leasing:} We provide a construction of OW-KLA secure PKE-SKL (which implies IND-KLA secure PKE-SKL as discussed above) by leveraging a PKE scheme that satisfies a new security notion that we call \emph{consistent or inconsistent security against key leasing attacks (CoIC-KLA security)}. We then construct a CoIC-KLA secure PKE scheme using 1-key Ciphertext-Policy Functional Encryption (CPFE) that in turn can be based on any IND-CPA secure PKE scheme.
      \item \emph{Identity Based Encryption, Attribute Based Encryption and Functional Encryption with Secure Key Leasing:} We provide definitions of secure key leasing in the context of advanced encryption schemes such as identity based encryption (IBE), attribute-based encryption (ABE) and functional encryption (FE). Then we provide constructions by combining the above PKE-SKL with standard IBE, ABE and FE schemes. 
     
     Notably, our definitions allow the adversary to request \emph{distinguishing} keys in the security game, namely, keys that distinguish the challenge bit by simply decrypting the challenge ciphertext, as long as it returns them (and they pass the validity test) before it sees the challenge ciphertext. All our constructions satisfy this stronger definition, albeit with the restriction that only a bounded number of such keys is allowed to the adversary in the IBE and ABE (but not FE) security games.
\end{enumerate}

Prior to our work, the notion of single decryptor encryption (SDE) has been studied in the context of PKE (Georgiou and Zhandry, Eprint 2020) and FE (Kitigawa and Nishimaki, Asiacrypt 2022) but all their constructions rely on strong assumptions including indistinguishability obfuscation. In contrast, our constructions do not require any additional assumptions, showing that PKE/IBE/ABE/FE can be upgraded to support secure key leasing for free.
\end{abstract}

\ifnum\llncs=1
\else
\newpage
\setcounter{tocdepth}{2}
\tableofcontents

\newpage
\fi

\ifnum\llncs=0

\section{Introduction}\label{sec:intro}
Recent years have seen amazing advances in cryptography by leveraging the power of quantum computation. Several novel primitives such as perfectly secure key agreement \cite{BB84}, quantum money \cite{wiesner1983conjugate}, quantum copy protection \cite{aaronson2009quantum}, one shot signatures \cite{STOC:AGKZ20} and such others, which are not known to exist in the classical world, can be constructed in the quantum setting, significantly advancing cryptographic capabilities. 

In this work, we continue to study harnessing quantum powers to protect against software piracy. The quantum no-cloning principle intuitively suggests applicability to anti-piracy, an approach which was first investigated in the seminal work of Aaronson \cite{aaronson2009quantum}, who introduced the notion of quantum copy protection. At a high level, quantum copy protection prevents users from copying software in the sense that it guarantees that when an adversary is given a copy protected circuit for computing some function $f$, it cannot create two (possibly entangled) quantum states, both of which can compute $f$. While interesting in its own right for preventing software piracy, quantum copy protection (for some class of circuits) also has the amazing application of public-key quantum money \cite{STOC:AarChr12}. Perhaps unsurprisingly, constructions of quantum copy protection schemes from standard cryptographic assumptions have remained largely elusive. This motivates the study of primitives weaker than quantum copy protection, which nevertheless offer meaningful guarantees for anti-piracy.

Secure software leasing (SSL), introduced by Ananth and La Placa \cite{EC:AnaLaP21}, is such a primitive, which while being weaker than quantum copy-protection, is nevertheless still meaningful for software anti-piracy. Intuitively, this notion allows to encode software into a version which may be leased or rented out, for some specific term at some given cost. Once the lease expires, the lessee returns the software and the lessor can run an efficient procedure to verify its validity. If the software passes the test, we have the guarantee that the lessee is no longer able to run the software (using the honest evaluation algorithm).

In this work, we explore the possibility of equipping public key encryption (PKE) with a key leasing capability. The benefits of such a capability are indisputable -- in the real world, decryption keys of users often need to be revoked, for instance, when a user leaves an organization. In the classical setting, nothing prevents the user from maintaining a copy of her decryption key and misusing its power. Revocation mechanisms have been designed to prevent such attacks, but these are often cumbersome in practice. Typically, such a mechanism entails the revoked key being included in a Certificate Revocation List (CRL) or Certificate Revocation Trees (CRT), or some database which is publicly available, so that other users are warned against its usage. However, the challenges of effective certificate revocation are well acknowledged in public key infrastructure -- please see \cite{USENIX:BDTW01} for a detailed discussion. If the decryption keys of a PKE could be encoded as quantum states and allow for verifiable leasing, this would constitute a natural and well-fitting solution to the challenge of key revocation.

\subsection{Prior Work}\label{sec:prior}
In this section, we discuss prior work related to public key encryption (PKE) and public key functional encryption (PKFE), where decryption keys are encoded into quantum states to benefit from uncloneability. For a broader discussion on prior work related to quantum copy protection and secure software leasing, we refer the reader to 
Section \ref{sec:other-rel}.

 Georgiou and Zhandry \cite{EPRINT:GeoZha20} introduced the notion of single decryptor encryption (SDE), where the decryption keys are unclonable quantum objects. They showed how to use one-shot signatures together with extractable witness encryption with quantum auxiliary information to achieve public key SDE. Subsequently, Coladangelo, Liu, Liu, and Zhandry \cite{C:CLLZ21} achieved SDE assuming iO and extractable witness encryption or assuming  subexponential iO, subexponential OWF, LWE and a strong monogamy property (which was subsequently shown to be true \cite{Quantum:CulVid22}). Very recently, Kitagawa and Nishimaki \cite{AC:KitNis22} introduced the notion of single-decryptor functional encryption (SDFE), where each functional decryption key is copy protected and provided collusion-resistant single decryptor PKFE for $\Ppoly$ from the subexponential hardness of iO and LWE. 

It is well-known \cite{C:ALLZZ21,EC:AnaLaP21} that copy protection is a stronger notion than SSL\footnote{The informed reader may observe that this implication may not always be true due to some subtleties, but we ignore these for the purpose of the overview.} -- intuitively, if an adversary can generate two copies of a program, then it can return one of them while keeping the other for later use. Thus, constructions of single decryptor encryption \cite{EPRINT:GeoZha20,C:CLLZ21,AC:KitNis22} imply our notion of PKE with secure key leasing from their respective assumptions, which all include at least the assumption of iO 
\ifnum\llncs=0
(see \cref{sec:SDE_and_PKE-SKL} for the detail).
\else
(see Appendix A of the full version for the detail).
\fi
Additionally, in the context of public key FE, the only prior work by Kitagawa and Nishimaki \cite{AC:KitNis22} considers the restricted single-key setting where an adversary is given a single decryption key that can be used to detect the challenge bit. In contrast, we consider the more powerful multi-key setting, which makes our definition of FE-SKL incomparable to the SDFE considered by \cite{AC:KitNis22}. For the primitives of IBE and ABE, there has been no prior work achieving any notion of key leasing to the best of our knowledge. We also note that Aaronson et al. \cite{C:ALLZZ21} studied the notion of ``copy-detection'', which is a weaker form of copy protection, for any ``watermarkable'' functionalities based on iO and OWF. In particular, by instantiating the construction with the watermarkable PKE of \cite{C:GKMWW19}, they obtain PKE with copy-detection from iO + PKE. 

Overall, all previous works that imply PKE-SKL are designed to achieve the stronger goal of copy protection (or the incomparable goal of copy detection) and rely at least on the strong assumption of iO. In this work, our goal is to achieve the weaker goal of PKE-SKL from standard assumptions.

\subsection{Our Results}
In this work, we initiate the study of public key encryption with secure key leasing. Our results can be summarized as follows:

\begin{enumerate}
    \item \emph{Definitions: } We introduce the definition of PKE with secure key leasing (PKE-SKL) to formalize the arguably natural requirement that decryption keys of a PKE scheme is encoded into a leased version so that the leased key continues to achieve the decryption functionality but now comes with an additional ``returnability'' guarantee. In more detail, the security of PKE-SKL requires that if the quantum decryption key returned by a user passes a validity test, then the user has lost the ability to decrypt. To capture this intuition, we formalize a security notion that we call indistinguishability against key leasing attacks (IND-KLA security). We also define a one-wayness notion for PKE-SKL that we call OW-KLA security and show that an OW-KLA secure PKE-SKL scheme can be lifted to an IND-KLA secure one by using the (quantum) Goldreich-Levin lemma.
    
    \item \emph{Constructing IND-KLA PKE with Secure Key Leasing:} We provide a construction of OW-KLA secure PKE-SKL (which imples IND-KLA PKE-SKL as discussed above) by leveraging a PKE scheme that satisfies a new security notion that we call \emph{consistent or inconsistent security against key leasing attacks (CoIC-KLA security)}. We then construct a CoIC-KLA secure PKE scheme using 1-key Ciphertext-Policy Functional Encryption (CPFE) that in turn can be based on any IND-CPA secure PKE scheme.
    
     \item \emph{Identity Based Encryption, Attribute Based Encryption and Functional Encryption with Secure Key Leasing:} We provide definitions of secure key leasing in the context of advanced encryption schemes such as identity based encryption (IBE), attribute-based encryption (ABE) and functional encryption (FE). Then we provide constructions by combining the above PKE-SKL with standard IBE, ABE and FE schemes. 
     
     Notably, our definitions allow the adversary to request {\it distinguishing} keys in the security game, namely, keys that distinguish the challenge bit by simply decrypting the challenge ciphertext. Recall that this was not permitted in the classical setting to avoid trivializing the security definition. However, in the quantum setting, we consider a stronger definition where the adversary can request such keys so long as it returns them (and they pass the validity test) before it sees the challenge ciphertext. All our constructions satisfy this stronger definition, albeit with the restriction that only a bounded number of such keys be allowed to the adversary in the IBE and ABE (but not FE) security games. We emphasize that this restriction is a result of our techniques and could potentially be removed in future work.
\end{enumerate}
We note that, in general, secure software leasing (SSL) only ensures a notion of security where the adversary is forced to use an honest evaluation algorithm for the software. However, our definition (and hence constructions) of PKE/ABE/FE SKL do not suffer from this limitation. Our constructions do not require any additional assumptions, showing that PKE/IBE/ABE/FE can be upgraded to support secure key leasing for free.


\newcommand{\Basic}{\mathsf{Basic}}
\newcommand{\Pivrfy}{\Pi_{\mathrm{vrfy}}}
\newcommand{\CPFE}{\mathsf{CPFE}}
\newcommand{\qstateq}{\qstate{q}}
\newcommand{\mat}[1]{\boldsymbol{#1}}

\subsection{Technical Overview}\label{sec:technical_overview}
We proceed to give a technical overview of this work.

\paragraph{Definition of PKE with secure key leasing.}
We first introduce the definition of PKE with secure key leasing (PKE-SKL).
A PKE-SKL scheme $\PKESKL$ consists of four algorithms $(\qKG, \Enc, \qDec,\qVrfy)$, where the first three algorithms form a standard PKE scheme except the following differences on $\qKG$.\footnote{In this paper, standard math or sans serif font stands for classical algorithms and classical variables. The calligraphic font stands for quantum algorithms and the calligraphic font and/or the bracket notation for (mixed) quantum states.}
\begin{itemize}
\item $\qKG$ outputs a quantum decryption key $\qdk$ instead of a classical decryption key.
\item $\qKG$ outputs a (secret) verification key $\vk$, together with a public encryption key and quantum decryption key.
\end{itemize}
The verification algorithm $\qVrfy$ takes as input a verification key and a quantum decryption key, and outputs $\top$ or $\bot$.
In addition to decryption correctness, $\PKESKL$ should satisfy verification correctness that states that $\qVrfy(\vk,\qdk) = \top$ holds, where $(\ek,\qdk,\vk)\la\qKG(1^\secp)$.

The security of PKE-SKL requires that once a user holding a quantum decryption key returns the key correctly, the user can no longer use the key and lose the ability to decrypt.
We formalize this as a security notion that we call indistinguishability against key leasing attacks (IND-KLA security).
It is defined by using the following security game.
\begin{enumerate}
\item First, the challenger generates $(\ek,\qdk,\vk)\la\qKG(1^\secp)$ and sends $\ek$ and $\qdk$ to an adversary $\qA$.
\item $\qA$ sends two challenge plaintexts $(\msg_0^*,\msg_1^*)$ and a quantum state $\widetilde{\qdk}$ that is supposed to be a correct decryption key.  The challenger checks if $\qVrfy(\vk,\widetilde{\qdk})=\top$ holds. If not, $\qA$ is regarded as invalid and the game ends here. Otherwise, the game goes to the next step.\footnote{We also consider a slightly stronger definition where the adversary can get access to a verification oracle many times, and the adversary is regarded as valid if the answer to at least one query $\widetilde{\qdk}$ is $\top$. In this overview, we focus on the ``1-query'' security for simplicity.}
\item The challenger generates $\ct^*\la\Enc(\ek,\msg_\coin^*)$ and sends it to $\qA$, where $\coin\la\bit$.
\item $\qA$ outputs $\coin'$.
\end{enumerate}
IND-KLA security guarantees that any QPT $\qA$ cannot guess $\coin$ correctly significantly better than random guessing, conditioned on $\qA$ being valid.  In more detail, for any QPT adversary $\qA$ that passes the verification with a non-negligible probability, we have $\abs{\Pr[\coin'=\coin \mid \qVrfy(\vk,\widetilde{\qdk})=\top]-1/2}=\negl(\secp)$. 

\paragraph{One-wayness to indistinguishability.}
It is natural to define a one-wayness notion for PKE-SKL, which we call OW-KLA security, by modifying the above definition so that the adversary is required to recover entire bits of a randomly chosen message from its ciphertext.
Similarly to standard PKE, we can transform a OW-KLA secure PKE-SKL scheme into an IND-KLA secure one by using (quantum) Goldreich-Levin lemma \cite{AC02,C:CLLZ21}.
Hence, though our goal is to construct an IND-KLA secure scheme, it suffices to construct an OW-KLA secure one.

\paragraph{Basic idea for OW-KLA secure scheme.}
Towards realizing a OW-KLA secure PKE-SKL scheme, we construct an intermediate scheme $\Basic=(\Basic.\qKG,\Basic.\Enc,\allowbreak \Basic.\qDec,\Basic.\qVrfy)$ using two instances of a standard PKE scheme, with parallel repetition.
Let $\PKE=(\PKE.\KG,\PKE.\Enc,\PKE.\Dec)$ be a standard PKE scheme.
$\Basic.\qKG$ generates two key pairs $(\ek_0,\dk_0)$ and $(\ek_1,\dk_1)$ using $\PKE.\KG$ and outputs $\ek:=(\ek_0,\ek_1)$, $\qdk:=1/\sqrt{2} (\ket{0}\ket{\dk_0}+\ket{1}\ket{\dk_1})$, and $\vk:=(\dk_0,\dk_1)$.
Given $\msg$ and $\ek$, $\Basic.\Enc$ generates $\ct_0\la\PKE.\Enc(\ek_0,\msg)$ and $\ct_1\la\PKE.\Enc(\ek_1,\msg)$ and outputs $\ct:=(\ct_0,\ct_1)$. 
$\Basic.\qDec$ can decrypt this ciphertext using the decryption keys $\dk_0$ and $\dk_1$, respectively, in superposition. Since both decryptions result in the same message $\msg$, we can decrypt ciphertexts without collapsing $\qdk$. Finally, $\Basic.\qVrfy$ checks if the input decryption key is an equal-weight superposition of $\dk_0$ and $\dk_1$. Concretely, it applies a binary outcome measurement w.r.t. a projection 
$\Pivrfy\seteq \frac{1}{2}\left(\ket{0}\ket{\dk_{0}}+\ket{1}\ket{\dk_{1}}\right)\left(\bra{0}\bra{\dk_{0}}+\bra{1}\bra{\dk_{1}}\right)$,
and returns $\top$ if and only if the state is projected onto $\Pivrfy$.

Intuitively, if the adversary has returned the correct decryption key, then it no longer has the capability to decrypt since the decryption key cannot be cloned. However, this scheme does not satisfy OW-KLA because an adversary can pass the verification with probability $1/2$ simply by measuring the decryption key and returning the collapsed decryption key.
Such an adversary can keep the decryption capability even after passing verification because the decryption key collapses to a classical string, which can be easily copied.
Nonetheless, it is reasonable to expect that this attack strategy is optimal because there appears to be no obvious way to attack with a better advantage. That said, it is unclear how to turn this intuition into a formal proof assuming only IND-CPA security of the underlying PKE. To address this gap, we introduce a new security notion for PKE, that we call \emph{consistent or inconsistent security against key leasing attacks (CoIC-KLA security)}. Using this, we can prove that the aforementioned adversarial strategy is optimal and $\Basic$ satisfies $1/2$-OW-KLA security.

By being $1/2$-OW-KLA secure, we mean that the probability that an adversary can correctly return a decryption key and recover the challenge plaintext simultaneously is at most $1/2+\negl(\secp)$.
Below, we introduce the definition of CoIC-KLA security and how to prove $1/2$-OW-KLA security of $\Basic$ using CoIC-KLA security.
Then, we explain how to achieve a full OW-KLA secure scheme by applying parallel amplification to $\Basic$.

\paragraph{Definition of CoIC-KLA security.}
CoIC-KLA security is defined by using the following game.
\begin{enumerate}
\item The challenger generates $(\ek_0,\dk_0)$ and $(\ek_1,\dk_1)$ using $\PKE.\KG$, and generates $\qdk:=1/\sqrt{2} (\ket{0}\ket{\dk_0}+\ket{1}\ket{\dk_1})$. The challenger sends $\ek_0$, $\ek_1$, and $\qdk$ to an adversary $\qA$.
In this game, $\qA$ can access the verification oracle only once, where the oracle is given a quantum state and 
returns the outcome of the projective measurement $(\Pivrfy,I-\Pivrfy)$.
\item $\qA$ sends two plaintexts $(\msg^*_0,\msg^*_1)$ to the challenger.
The challenger picks random bits $a,b$ and generates $\ct_0=\Enc(\ek_0,\msg_a)$ and $\ct_1=\Enc(\ek_1,\msg_{a\oplus b})$. Then, the challenger sends $\ct_0$ and $\ct_1$ to $\qA$.
\item $\qA$ outputs a bit $b'$.
\end{enumerate}
Then, CoIC-KLA security requires that any QPT $\qA$ cannot guess $b$ significantly better than random guessing.
In the above game, if $b=0$, $\ct_0$ and $\ct_1$ are ciphertexts of the same plaintext $\msg_a^*$. On the other hand, if $b=1$, $\ct_0$ and $\ct_1$ are ciphertexts of the different plaintexts $\msg_a^*$ and $\msg_{1\oplus a}^*$. Thus, we call this security notion consistent or inconsistent security.

\paragraph{$1/2$-OW-KLA security of $\Basic$.}
We explain how to prove $1/2$-OW-KLA security of $\Basic$ based on CoIC-KLA security of $\PKE$.
The OW-KLA security game for $\Basic$ is as follows.
\begin{enumerate}
\item The challenger generates $(\ek_0,\dk_0)$ and $(\ek_1,\dk_1)$ using $\PKE.\KG$, sets $\ek:=(\ek_0,\ek_1)$ and $\qdk:=1/\sqrt{2} (\ket{0}\ket{\dk_0}+\ket{1}\ket{\dk_1})$, and sends $\ek$ and $\qdk$ to an adversary $\qA$.
\item The adversary returns a quantum state $\widetilde{\qdk}$ that is supposed to be a correct decryption key. The challenger checks if the result of applying $\Pivrfy$ defined above to $\widetilde{\qdk}$ is $1$. If not, $\qA$ is regarded as invalid and the game ends here. Otherwise, the game goes to the next step.
\item The challenger generates random plaintext $\msg^*$ and two ciphertexts $\ct_0\la\PKE.\Enc(\ek_0,\msg^*)$ and $\ct_1\la\PKE.\Enc(\ek_1,\msg^*)$, and sends $\ct:=(\ct_0,\ct_1)$ to $\qA$.
\item $\qA$ outputs $m'$.
\end{enumerate}
In this game, we say that $\qA$ wins if $(a)$ $\widetilde{\qdk}$ passes the verification, that is, the result of applying $\Pivrfy$ to $\widetilde{\qdk}$ is $1$, and $(b)$ $m'=m^*$ holds.
$\qA$ can win this game with probability at least $1/2$ by just measuring $1/\sqrt{2} (\ket{0}\ket{\dk_0}+\ket{1}\ket{\dk_1})$, returns collapsed key, and decrypt the challenge ciphertext with the key. As stated above, we can prove that this is the optimal strategy for $\qA$, that is, we can bound the advantage of $\qA$ by $1/2+\negl(\secp)$.
The proof can be done by using game sequences. We denote the probability that $\qA$ wins in Game $i$ as $\Pr[S_i]$.
\begin{description}
\item[Game $0$:] This is exactly the above game.
\item[Game $1$:] We defer the verification of the returned key $\widetilde{\qdk}$ after $\qA$ outputs $\msg'$.
\end{description}
From the deferred measurement principle, we have $\Pr[S_0]=\Pr[S_1]$.
\begin{description}
\item[Game $2$:] We change $\qA$'s winning condition $(b)$. Concretely, we replace $(b)$ with $(b')$ $\msg' \in \{\msg^*,\tilde{\msg}\}$ holds, where $\tilde{m}$ is a random plaintext.
\end{description}
Since we relaxed $\qA$'s winning condition, we have $\Pr[S_1] \le \Pr[S_2]$.
\begin{description}
\item[Game $3$:] We generate $\ct_1$ as $\ct_1\la\PKE.\Enc(\ek_1,\tilde{\msg})$ instead of $\ct_1\la\PKE.\Enc(\ek_1,\msg^*)$.
\end{description}
The only difference between Game $2$ and $3$ is that $\ct_0$ and $\ct_1$ are ciphertexts of the same plaintext in Game $2$, but they are ciphertexts of different plaintexts in Game $3$.
Thus, we obtain $\abs{\Pr[S_2]-\Pr[S_3]}=\negl(\secp)$ using CoIC security of $\PKE$.

We complete the proof by showing that $\Pr[S_3] \le 1/2 + \negl(\secp)$ holds if $\PKE$ satisfies one-wayness (that is implied by CoIC-KLA security).
To show it, we use the following Fact $1$.
\begin{description}
\item[Fact $1$:] Assume $\PKE$ satisfies one-wayness. Then, given $1/\sqrt{2} (\ket{0}\ket{\dk_0}+\ket{1}\ket{\dk_1})$, $\PKE.\Enc(\ek_0,\msg^*)$, and $\PKE.\Enc(\ek_1,\tilde{\msg})$, no adversary can obtain $(\dk_0,\tilde{\msg})$ or $(\dk_1,\msg^*)$ with non-negligible probability. 
\end{description}
This can be proved by using the fact that even if we measure $1/\sqrt{2} (\ket{0}\ket{\dk_0}+\ket{1}\ket{\dk_1})$ in the computational basis before giving it to the adversary, the adversary still has success probability at least $\epsilon/2$, where $\epsilon$ is the success probability of the original experiment
\ifnum\llncs=0
 \cite[Lemma 2.1]{C:BonZha13} (which is stated as \cref{lem:BZ}).
 \else 
 \cite[Lemma 2.1]{C:BonZha13}.
 \fi

Suppose $\Pr[S_3]=1/2 + 1/\poly(\secp)$ for some polynomial $\poly$.
This means that conditioned that $\msg' \in \{m^*,\tilde{m}\}$, $\widetilde{\qdk}$ returned by $\qA$ passes the verification with probability significantly greater than $1/2$.
Thus, if we measure $\widetilde{\qdk}$ in the computational basis, we obtain $\dk_0$ with some inverse polynomial probability and also $\dk_1$ with some inverse polynomial probability. (If either one is obtained with overwhelming probability, $\widetilde{\qdk}$ cannot pass the verification with probability significantly greater than $1/2$.) This means that using $\qA$, we can obtain either one pair of $(\dk_0,\tilde{\msg})$ or $(\dk_1,\msg^*)$ with inverse polynomial probability, which contradicts Fact $1$.
Thus, we obtain $\Pr[S_3]\le 1/2+\negl(\secp)$.

From the above discussions, we can conclude that if $\PKE$ satisfies CoIC-KLA security, $\Basic$ satisfies $1/2$-OW-KLA security.

\paragraph{Full OW-KLA security by parallel repetition.}
To achieve a fully OW-KLA secure scheme, we apply parallel amplification to $\Basic$ in the following way.
When generating a key tuple, we generate $\secp$ key tuples $(\ek_i,\qdk_i,\vk_i)$ of $\Basic$ and set $\ek':=(\ek_i)_{i\in[\secp}$, $\qdk':=(\qdk_i)_{i\in[\secp]}$, and $\vk':=(\vk_i)_{i\in[\secp]}$.
When encrypting a plaintext $\msg$, we divide it into $\secp$ pieces $\msg_1,\cdots,\msg_\secp$, and encrypt each $\msg_i$ using $\ek_i$.
Then decryption and verification are performed naturally by running the underlying procedures in $\Basic$ for every $i\in[\secp]$. We can prove the full OW-KLA security of this construction using a strategy analogous to that used to achieve $1/2$-OW-KLA security of $\Basic$. We remark that it is unclear whether we can amplify $1/2$-OW-KLA security to full OW-KLA security in a black box way and our security proof relies on the specific structure of our scheme.

\paragraph{Constructing CoIC-KLA secure PKE scheme.}
In the rest of this overview, we mainly explain how to construct CoIC-KLA secure PKE scheme.
We construct it using 1-key Ciphertext-Policy Functional Encryption (CPFE) that in turn can be based on any IND-CPA secure PKE scheme.

We first review the definition of 1-key CPFE scheme.
A 1-key CPFE scheme $\CPFE$ consists of four algorithms $(\FE.\Setup,\FE.\KG,\FE.\Enc,\FE.\Dec)$.
Given a security parameter, $\FE.\Setup$ outputs a master public key $\mpk$ and a master secret key $\msk$.
$\FE.\KG$ takes as input $\msk$ and a string $x$ and outputs a decryption key $\sk_x$ tied to the string $x$.
$\FE.\Enc$ takes as input $\mpk$ and a description of a circuit $C$ and outputs a ciphertext $\ct$.
If we decrypt this ciphertext $\ct$ with $\sk_x$ using $\FE.\Dec$, we can obtain $C(x)$.
The security of it states that ciphertexts of two circuits $C_0$ and $C_1$ are computationally indistinguishable for an adversary who has decryption key $\sk_x$ for $x$ of its choice, as long as $C_0(x)=C_1(x)$ holds.

Letting $\CPFE=(\FE.\Setup,\FE.\KG,\FE.\Enc,\FE.\Dec)$ be a 1-key CPFE scheme, we construct a CoIC secure PKE scheme $\PKE=(\PKE.\KG,\PKE.\Enc,\PKE.\Dec)$ as follows.
$\PKE.\KG$ generates $(\mpk,\msk)\la\CPFE.\Setup(1^\secp)$ and a decryption key $\sk_x\la\CPFE.\KG(\msk,x)$ for random string $x$, and outputs an encryption key $\ek:=\mpk$ and the corresponding decryption key $\dk:=\sk_x$.
Given $\ek=\mpk$ and $\msg$, $\PKE.\Enc$ outputs $\FE.\Enc(\mpk,C[\msg])$, where $C[\msg]$ is the constant circuit that outputs $\msg$ on any input.
Given $\dk=\sk_x$ and $\ct$, $\PKE.\Dec$ simply outputs $\CPFE.\Dec(\sk_x,\ct)$.
We see that $\PKE$ satisfies decryption correctness from that of $\CPFE$.

Before proving CoIC-KLA security of $\PKE$, we explain a nice tracing property of $\PKE$ that plays an important role in the proof.
It says that if there exists a decoder that can distinguish $\PKE.\Enc(\ek,\msg_0^*)$ and $\PKE.\Enc(\ek,\msg_1^*)$ with probability $1/2 + 1/\poly(\secp)$ for some plaintexts $\msg_0^*,\msg_1^*$ and polynomial $\poly$, we can extract the string $x$ tied to the decryption key from the decoder. Concretely, the following fact holds.
\begin{description}
\item[Fact $2$:]Consider the following experiment. The challenger generates $(\ek:=\mpk,\dk:=\sk_x)$ using $\PKE.\KG$ and sends them to an adversary $\qA$. $\qA$ outputs a decoder $D$ together with $\msg_0^*,\msg_1^*$ that can predict random bit $b$ from $\PKE.\Enc(\ek,\msg_b^*)$ with probability $1/2 + 1/\poly(\secp)$ for some polynomial $\poly$.
Then, we can extract $x$ from $D$ with inverse polynomial probability.
\end{description}

In fact, if the decoder $D$ is a classical decoder, we can extract $x$ from $D$ with a probability close to $1$ as follows.
Let $\tilde{C}[b,\msg_0,\msg_1,i]$ be the circuit that is given $x$ as an input and outputs $\msg_{b\oplus x[i]}$, where $x[i]$ is the $i$-th bit of $x$. Then, suppose we generate many random $(b, \FE.\Enc(\mpk,\tilde{C}[b,\msg_0^*,\msg_1^*,i]))$ and estimate the probability that the decoder $D$ outputs $b$ given $\FE.\Enc(\mpk,\tilde{C}[b,\msg_0^*,\msg_1^*,i])$ as an input. By the $\CPFE$'s security, $\FE.\Enc(\mpk,\tilde{C}[b,\msg_0^*,\msg_1^*,i])$ is indistinguishable from a correctly generated ciphertext of $\msg_{b\oplus x_i}^*$, that is, $\PKE.\Enc(\ek,\msg_{b\oplus x_i}^*)=\FE.\Enc(\mpk, C[\msg_{b\oplus x_i}^*])$ from the view of $\qA$ and $D$ who has $\sk_x$, since $\tilde{C}[b,\msg_0^*,\msg_1^*,i](x)=C[\msg_{b\oplus x_i}^*](x)=\msg_{b\oplus x_i}^*$. 
Then, the result of the estimation should be as follows.
\begin{itemize}
\item In the case of $x[i]=0$, each sample used for the estimation looks $(b, \PKE.\Enc(\ek,\msg_b))$ from the view of $D$. Thus, the result of the estimation should be greater than $1/2$ from the fact that $D$ correctly predicts random bit $b$ from $\PKE.\Enc(\ek,\msg_b)$ with probability $1/2 + 1/\poly(\secp)$.
\item In the case of $x[i]=1$, each sample used for the estimation looks $(b, \PKE.\Enc(\ek,\msg_{1\oplus b}))$ from the view of $D$. Thus, the result of the estimation should be smaller than $1/2$ since $D$ outputs $1\oplus b$ given $\PKE.\Enc(\ek,\msg_{1\oplus b})$ with probability $1/2 + 1/\poly(\secp)$.
\end{itemize}
Therefore, by checking if the result of the estimation is greater than $1/2$ or not, we can extract $x[i]$. By doing this for every $i$, we can extract entire bits of $x$.

The above extraction technique is a direct application of that used by Kitagawa and Nishimaki~\cite{EC:KitNis22} to realize watermarking scheme secure against quantum adversaries.
By using their technique, even if the decoder is a quantum decoder $\qD$ that consists of a unitary and an initial quantum state, we can extract $x$ from $\qD$ with inverse polynomial probability, as long as $\qD$ has a high distinguishing advantage.
Roughly speaking, this is done by performing the above estimation using (approximate) projective implementation proposed by Zhandry~\cite{TCC:Zhandry20} that is based on the technique by Marriott and Watrous~\cite{CC:MarWat05}.
By extending the above extraction technique, we can obtain the following fact.
\begin{description}
\item[Fact $3$:]Consider the following experiment. The challenger generates $(\ek_0:=\mpk_0,\dk_0:=\sk_{x_0})$ and $(\ek_1:=\mpk_1,\dk_1:=\sk_{x_1})$ using $\PKE.\KG$, and sends $\ek_0$, $\ek_1$, and $1/\sqrt{2} (\ket{0}\ket{\dk_{0}}+\ket{1}\ket{\dk_{1}})=1/\sqrt{2} (\ket{0}\ket{\sk_{x_0}}+\ket{1}\ket{\sk_{x_1}})$ to an adversary $\qA$.
$\qA$ outputs a quantum decoder $\qD$ together with $(\msg_0^*,\msg_1^*)$ that can predict $b$ from $\PKE.\Enc(\ek_0,\msg_a)$ and $\PKE.\Enc(\ek_1,\msg_{a\oplus b})$ with probability $1/2 + 1/\poly(\secp)$ for some polynomial $\poly$.
Then, we can extract both $x_0$ and $x_1$ from $\qD$ with inverse polynomial probability.
\end{description}

We now explain how we can prove CoIC-KLA security of $\PKE$ using Fact $3$.
To this end, we introduce one more fact.
\begin{description}
\item[Fact $4$:]Given $\mpk_0$, $\mpk_1$, and $1/\sqrt{2} (\ket{0}\ket{\sk_{x_0}}+\ket{1}\ket{\sk_{x_1}})$, where $(\mpk_0,\sk_{x_0})$ and $(\mpk_1,\sk_{x_1})$ are generated as in $\PKE.\KG$, no adversary can compute both $x_0$ and $x_1$ with non-negligible probability.
\end{description}
Similarly to Fact $1$, we can prove this from the fact that even if we measure $1/\sqrt{2} (\ket{0}\ket{\sk_{x_0}}+\ket{1}\ket{\sk_{x_1}})$ in the computational basis before giving it to the adversary, the adversary still has success probability at least $\epsilon/2$, where $\epsilon$ is the success probability of the original experiment \cite[Lemma 2.1]{C:BonZha13}.

Suppose there exists a QPT adversary $\qA$ that breaks CoIC-KLA security of $\PKE$.
We consider the following adversary $\qB$ using $\qA$.
Given $\mpk_0$, $\mpk_1$, and $1/\sqrt{2} (\ket{0}\ket{\sk_{x_0}}+\ket{1}\ket{\sk_{x_1}})$, $\qB$ simulates CoIC-KLA security game for $\qA$ by setting $\ek_0:=\mpk_0$, $\ek_1:=\mpk_1$, and $\qdk:=1/\sqrt{2} (\ket{0}\ket{\sk_{x_0}}+\ket{1}\ket{\sk_{x_1}})$ until $\qA$ outputs two plaintexts $(\msg_0^*,\msg_1^*)$. When $\qA$ makes a verification query, $\qB$ just returns a random bit. Let $\mat{U}$ be the unitary that performs the rest of $\qA$'s actions given the challenge ciphertexts. Also, let $\qstateq$ be the internal state of $\qA$ at this point. Then, from the averaging argument and the fact that $\qB$ correctly answers to $\qA$'s verification query with probability $1/2$, with some inverse polynomial probability, the quantum decoder $\qD=(\mat{U},\qstateq)$ is a decoder that can predict $b$ from $\PKE.\Enc(\ek_0,\msg_a^*)$ and $\PKE.\Enc(\ek_1,\msg_{a\oplus b}^*)$ with probability $1/2 + 1/\poly(\secp)$ for some polynomial $\poly$.
Thus, by using the extractor that is guaranteed to exist by Fact $3$, $\qB$ can obtain both $x_0$ and $x_1$ with some inverse polynomial probability, which contradicts Fact $4$.
This means that $\PKE$ satisfies CoIC-KLA security.

\paragraph{Extension to Advanced Encryption Systems with Secure Key Leasing.}
We also provide constructions of advanced encryption schemes such as ABE and FE with secure key leasing. We do not focus on IBE in this paper since IBE is a special case of ABE and our transformation preserves the underlying function class.\footnote{Although ABE is a special case of FE, we need stronger assumptions for (collusion-resistant) FE to instantiate them. In addition, the security level of FE-SKL that we can achieve is different from that of ABE-SKL. Hence, we consider both ABE and FE.}
We construct these schemes by carefully combining standard ABE (resp. FE) with PKE-SKL in the way that each decryption key of the resulting ABE-SKL (resp. FE-SKL) scheme includes a decryption key of the underlying PKE-SKL scheme and a ciphertext of the ABE-SKL (resp. FE-SKL) scheme cannot be decrypted without the decryption key of the underlying PKE-SKL scheme.
By doing so, our ABE-SKL and FE-SKL take over the secure key leasing security from the underlying PKE-SKL.
Moreover, since PKE-SKL can be based on any PKE, our ABE-SKL and FE-SKL can be based on any standard ABE and FE, respectively.

\paragraph{ABE-SKL.}
Here, we provide an overview of ABE with secure key leasing.
Let us start with the definition of plain ABE (without key leasing). 
An ABE scheme $\ABE$ consists of four algorithms $(\ABE.\Setup,\ABE.\KG,\ABE.\Enc,\allowbreak \ABE.\Dec)$
and is associated with a relation $R$.
Given a security parameter, $\ABE.\Setup$ outputs a master public key $\mpk$ and a master secret key $\msk$.
$\ABE.\KG$ takes as input $\msk$ and a key attribute $y$ and outputs a user secret key $\sk_y$ tied to the attribute $y$.
$\ABE.\Enc$ takes as input $\mpk$, a ciphertext attribute $x$, and a message $\msg$ and outputs a ciphertext $\ct$.
The decryption of the ciphertext is possible only when $R(x,y)=1$.
For this reason, we call a user secret key for attribute $y$ satisfying $R(x,y)=1$ a decrypting key (for a ciphertext associated with $x$).
As for the security, we require that $\ABE.\Enc(x^*,\msg^*_0)$ should be computationally indistinguishable from $\ABE.\Enc(x^*,\msg^*_1)$ as long as an adversary is only given non-decrypting keys for the ciphertext (i.e., user secret keys for $y$ satisfying $R(x^*,y)=0$).

We now define the notion of ABE with secure key leasing (ABE-SKL) by extending the syntax of ABE. 
The difference from the above is that the key generation algorithm is now quantum and it outputs user secret key $\qusk_y$ along with verification key $\vk$. 
We also additionally introduce a verification algorithm that takes $\vk$ and a quantum state $\qusk'$ and outputs $\top$ if it judges that the user secret key corresponding to $\vk$ is correctly returned and $\bot$ otherwise.
As for the security, we require that $\ABE.\Enc(x^*,\msg_0)$ should be computationally indistinguishable from $\ABE.\Enc(x^*,\msg_1)$ if the adversary returns all decrypting keys before it is given the challenge ciphertext.
Here, we say the adversary returns the key if the adversary provides the challenger with a quantum state that makes the verification algorithm output $\top$. 

For the construction, the basic idea is to use ABE for access control and PKE-SKL for obtaining security against key leasing attacks.
To enable this idea, we encrypt a message $\msg$ for an attribute $x$ so that the decryptor recovers PKE-SKL ciphertext $\skl.\ct=\PKESKL.\Enc(\skl.\ek,\msg)$ if it has decrypting key and nothing otherwise, where $\skl.\ek$ is an individual encryption key corresponding to the user.
The user is given the corresponding decryption key $\skl.\dk$ and can recover the message by decrypting $\skl.\ct$.
Roughly speaking, the security follows since (1) a user with a non-decrypting key cannot obtain any information and (2) even a user with a decrypting key cannot recover the message from $\skl.\ct$ once it returns $\skl.\dk$ due to the security of $\PKESKL$.

The generation of user individual $\PKESKL$ ciphertext is somewhat non-trivial since $\ABE$ can only encrypt a single message. 
In order to achieve this, we use an idea similar to \cite{CCS:SahSey10,TCC:GoyKopWat16} that combines encryption with the garbled circuits.
In particular, we garble the encryption circuit of $\PKESKL$ that hardwires a message and
encrypt the labels by ABE. We then provide a secret key of ABE for a user only for the positions corresponding to $\skl.\ek$. 
This allows a user with decrypting key to recover the labels corresponding to $\skl.\ek$ and then run the garbled circuit on input the labels to recover $\skl.\ct$.

Unfortunately, the introduction of the garbled circuits in the construction poses some limitations on the security of the scheme. In particular, once the adversary obtains two decrypting user secret keys, the message can be revealed from the garbled circuit in the ciphertext since the security of garbled circuits is compromised when labels for two different inputs are revealed. 
Therefore, we are only able to prove $1$-bounded distinguishing key security,\footnote{
When we consider the security game for ABE-SKL, a decrypting key can be used for distinguishing the challenge bit by decrypting the challenge ciphertext (if it is not returned). Therefore, we use the term ``decrypting key" and ``distinguishing key" interchangeably.
} where the adversary can make a single decrypting key query and should return the key before the challenge ciphertext is given. We note that the adversary can make an arbitrary number of non-decrypting key queries throughout the game, unlike bounded collusion ABE \cite{C:GorVaiWee12,PKC:ISVWY17}
and only the number of decrypting keys is bounded.

Ideally, we would like to have a scheme without restriction on the number of decrypting keys. However, we do not know how to achieve it without strong assumptions like functional encryption or indistinguishability obfuscation. Instead, we achieve intermediate security notion that we call $q$-bounded distinguishing key security without introducing additional assumption, where the number of decrypting keys is bounded by some pre-determined polynomial.  
To do so, we use the same idea as \cite{PKC:ISVWY17}, which converts single bounded collusion ABE into $q$-bounded collusion ABE. 
The construction is based on the balls and bins idea, where we prepare multiple ``bins", each of which consists of multiple instances of $1$-bounded distinguishing key secure ABE-SKL $\oneABE$. The key generation algorithm chooses a single instance from each bin randomly and generates a user secret key for each of them. The encryption algorithm secret shares the message and encrypts them using the instances of the $\oneABE$ so that the same share is encrypted by the instances in the same bin. 
By careful choices of the parameters and analysis, in the security proof, we can argue that there exists a bin such that $\oneABE$ instances used for generating decrypting keys in that bin are all distinct. This means that for every $\oneABE$ instance in that bin, only a single decrypting key is generated and thus, we can use $1$-bounded distinguishing key security for each of them.
While this overall proof strategy is the same as \cite{PKC:ISVWY17}, our proof is a little bit more complex than theirs because the adversary is allowed to make an unbounded number of (non-decrypting) key queries. \ifnum\llncs=0 We refer to \cref{sec:ibe_abe_skl} for further details.\fi

\paragraph{PKFE-SKL.}
We move to the overview of PKFE-SKL. In this work, we focus on Key-Policy FE (KPFE) with secure key leasing.
We start with the definition of plain FE (without key leasing).
An FE scheme $\FE$ consists of four algorithms $(\FE.\Setup,\FE.\KG,\FE.\Enc,\FE.\Dec)$ and is associated with a function class $\cF$.
Given a security parameter, $\FE.\Setup$ outputs a public key $\pk$ and a master secret key $\msk$.
$\FE.\KG$ takes as input $\msk$ and a function $f\in\cF$ and outputs a functional decryption key $\sk_f$ tied to the function $f$.
$\FE.\Enc$ takes as input $\pk$ and a plaintext $x$ and outputs a ciphertext $\ct$.
The decryption result is $f(x)$. For security, we require that $\FE.\Enc(\pk,x_0)$ should be computationally indistinguishable from $\FE.\Enc(\pk,x_1)$ as long as an adversary is only given functional decryption keys for $\setbk{f_i}_i$ such that $f_i(x_0)=f_i(x_1)$ for all $i$.

We define the notion of FE with secure key leasing (FE-SKL) by extending the syntax of FE like ABE-SKL.
The key generation algorithm is now quantum and it outputs functional decryption key $\qsk_f$ along with verification key $\vk$. 
We also introduce a verification algorithm that takes $\vk$ and a quantum state $\qsk^\prime$ and outputs $\top$ if it judges that the functional decryption key corresponding to $\vk$ is correctly returned and $\bot$ otherwise.

In the security game of PKFE-SKL, the adversary can send a \emph{distinguishing} key query $f$ such that $f(x_0^\ast) \ne f(x_1^\ast)$ where $(x_0^\ast,x_1^\ast)$ are the challenge plaintexts \emph{as long as it returns a valid functional decryption key for $f$}. We consider a security game where the adversary can send unbounded polynomially many distinguishing and non-distinguishing (that is, $f(x_0^\ast)=f(x_1^\ast)$) key queries and tries to distinguish $\FE.\Enc(\pk,x_0)$ from $\FE.\Enc(\pk,x_1)$.

We transform a (classical) PKFE scheme into a PKFE scheme with secure key leasing by using the power of PKE-SKL. The basic idea is as follows. When we generate a functional decryption key for function $f$, we generate a key triple of PKE-SKL and a functional decryption key of the classical PKFE for a function $W$ that computes a PKE-SKL ciphertext of $f(x)$. That is, we wrap $f(x)$ by PKE-SKL encryption. A decryption key of PKE-SKL is appended to $\fe.\sk_W$, which is the functional decryption key for $W$. Hence, we can decrypt the PKE-SKL ciphertext and obtain $f(x)$. The PKE-SKL decryption key for $f$ is useless for another function $g$ since we use different key triples of PKE-SKL for each function.

More specifically, we generate PKE-SKL keys $(\skl.\ek,\skl.\qsk,\skl.\vk)$ and a PKFE functional decryption key $\fe.\sk_W \gets \FE.\KG(\fe.\msk,W[f,\skl.\ek])$, where function $W[f,\skl.\ek]$ takes as input $x$ and outputs a PKE-SKL ciphertext $\SKL.\Enc(\skl.\ek,f(x))$.\footnote{We ignore the issue of encryption randomness here. In our construction, we use (puncturable) PRFs to generate encryption randomness.} A functional decryption key for $f$ consists of $(\fe.\sk_W,\skl.\qsk)$. A ciphertext of $x$ is a (classical) PKFE ciphertext $\FE.\Enc(\fe.\pk,x)$. If we return $\skl.\qsk$ for $f$ (verified by $\skl.\vk$) before we obtain $\FE.\Enc(\fe.\pk,x)$, we cannot obtain $f(x)$ from $\SKL.\Enc(\skl.\ek,f(x))$ by the security of PKE-SKL.

We need to prove security against an adversary that obtains a functional decryption key for $f$ such that $f(x_0^\ast)\ne f(x_1^\ast)$ where $(x_0^\ast,x_1^\ast)$ is a pair of challenge plaintexts if the adversary returns the functional decryption key. To handle this issue, we rely on IND-KLA security and need to embed a challenge ciphertext of PKE-SKL into a PKFE ciphertext. We use the trapdoor method of FE (a.k.a. Trojan method)~\cite{C:ABSV15,JC:BraSeg18} for this purpose.
We embed an SKFE functional decryption key and ciphertext in a PKFE functional decryption key and ciphertext, respectively. We use these SKFE functional decryption key and ciphertext for the trapdoor mode of PKFE. We gradually change SKFE ciphertexts and keys so that we can embed a PKE-SKL challenge ciphertext by using the adaptively single-ciphertext function privacy of SKFE. Once we succeed in embedding a PKE-SKL challenge ciphertext, we can change a ciphertext of $x_0^\ast$ into a ciphertext of $x_1^\ast$ such that $f(x_0^\ast)\ne f(x_1^\ast)$ as long as the functional decryption key $\qsk_f = (\fe.\sk_W,\skl.\qsk)$ for $f$ is returned.
This is because $\skl.\qsk$ is returned and we can use IND-KLA security under $\skl.\ek$. 
\ifnum\llncs=0 See~\cref{sec:PKFE_SKL} for more details. \fi

\subsection{Other Related Work} 
\label{sec:other-rel}
\indent \emph{Quantum Copy Protection.} Aaronson \cite{aaronson2009quantum} introduced the notion of quantum copy protection and constructed a quantum copy protection scheme for arbitrary unlearnable Boolean functions relative to a quantum oracle. He also provided two heuristic copy-protection schemes for point functions in the standard model. Coladangelo et al.~\cite{ARXIV:ColMajPor20} provided a quantum copy-protection scheme for a class of evasive functions in the QROM. Subsequently, Aaronson et al.~\cite{C:ALLZZ21} constructed a quantum copy protection scheme for unlearnable functions relative to classical oracles. By instantiating the oracle with post-quantum candidate obfuscation schemes, they obtained a heuristic construction of copy protection. Coladangelo et al.~\cite{C:CLLZ21} provided a copy-protection scheme for pseudorandom functions in the plain model assuming iO, OWF and extractable witness encryption, or assuming subexponential iO, subexponential OWF, LWE and a strong ``monogamy property'' (which was was proven to be true in a follow-up work \cite{Quantum:CulVid22}). Ananth et al. \cite{TCC:AnaKal21,C:AKLLZ22} also constructed copy protection for point functions, which in turn can be transformed into copy protection for compute-and-compare programs. 
Sattath and Wyborski~\cite{ARXIV:SatWyb22} studied unclonable decryptors, which are an extension of SDE. Their unclonable decryptors scheme is \emph{secret key} encryption and can be instantiated with iO and OWF, or quantum oracles.


\medskip

\noindent\emph{Secure software leasing.} Secure software leasing (SSL) was introduced by Ananth and La Placa \cite{EC:AnaLaP21}, where they also provided the first SSL scheme supporting a subclass of ``evasive'' functions by relying on the existence of public key quantum money and the learning with errors assumption. Evasive functions is a class of functions for which it is hard to find an accepting input given only black-box access to the function. Their construction achieves a strong security notion called {\it infinite term security}. They also demonstrate that there exists an unlearnable function class such that it is impossible to achieve an SSL scheme for that function class, even in the CRS model. Later, Coladangelo et al.~\cite{ARXIV:ColMajPor20} improved the security notion achieved by \cite{EC:AnaLaP21} by relying on the QROM, for the same class of evasive functions. Additionally, Kitagawa, Nishimaki and Yamakawa \cite{TCC:KitNisYam21} provided a finite term secure SSL scheme for pseudorandom functions (PRFs) in the CRS model by assuming the hardness of the LWE problem against polynomial time quantum adversaries. Additionally, this work achieves classical communication. Further, Broadbent et al.~\cite{TCC:BJLPS21} showed that SSL is achievable for the aforementioned evasive circuits without any setup or computational assumptions that were required by previous work, but with finite term security, quantum communication and correctness based on a distribution. The notion of secure leasing for the powerful primitive of functional encryption was studied by Kitagawa and Nishimaki \cite{AC:KitNis22}, who introduced the notion of \emph{secret key} functional encryption (SKFE) with secure key leasing and provided a transformation from standard SKFE into SKFE with secure key leasing without relying on any additional assumptions. 



\medskip
\noindent\emph{Certified deletion.}
Broadbent and Islam~\cite{TCC:BroIsl20} introduced the notion of quantum encryption with certified deletion, where we can generate a (classical) certificate to ensure that \emph{a ciphertext} is deleted. They constructed a one-time SKE scheme with certified deletion without computational assumptions. After that, many works presented various quantum encryption primitives (PKE, ABE, FE and so on) with certified deletion~\cite{AC:HMNY21,ITCS:Poremba23,EPRINT:BarKhu22,EPRINT:HMNY22}.
The root of quantum encryption with certified deletion is revocable quantum time-released encryption by Unruh~\cite{JACM:Unruh15}. It is an extension of time-released encryption where a sender can revoke quantum encrypted data before a pre-determined time. If the revocation succeeds, the receiver cannot obtain the plaintext information.

\medskip
\noindent\emph{Related technique.}
The basic idea of our PKE-SKL is to prepare a superposition of two decryption keys and coherently run the decryption algorithm in each branch. 
Previous works by Zhang~\cite{STOC:Zhang21,FOCS:Zhang22} use a similar idea of running some algorithm (which is an evaluation of ``lookup tables'' in their case) on two branches in superposition though their motivation is to construct efficient blind quantum computation and classical verification of quantum computation, which are completely irrelevant to PKE-SKL.

\subsection{Concurrent Work}
A concurrent and independent work by Ananth, Poremba, and Vaikuntanathan~\cite{APV23} introduces key-revocable PKE, which is similar to PKE-SKL. They construct key-revocable PKE based on the LWE assumption while our construction of PKE-SKL only assumes the existence of IND-CPA secure PKE. In addition, they only prove somewhat weaker security notion called $1$-bit unpredictability. Roughly, it ensures that the probability that the adversary passes the verification for the returned key \emph{and} wins the IND game is at most $1/2+\negl(\secp)$. For example, even if an adversary passes the verification with probability $1/3$ and has a distinguishing advantage $1$ conditioned on the acceptance, it is not considered to break the security while such an adversary breaks IND-KLA security. Thus, we believe that IND-KLA security is more desirable security notion than $1$-bit unpredictability.\footnote{Strictly speaking, IND-KLA security and $1$-bit unpredictability are incomparable because the 
former requires the indistinguishability between ciphertexts of two different messages whereas
the latter requires the indistinguishability between a ciphertext of some message and a uniformly random string. }
On the other hand, the advantages of their work are that their construction of key-revocable PKE is based on dual-Regev encryption, which is likely to be more efficient than our PKE-SKL, and that they also show a fully homomorphic encryption variant.   

\ifnum\llncs=0
\subsection{Organization of the paper}
In Section \ref{sec:prelim} we define the notation and preliminaries that we require in this work.  In Section \ref{sec:pke-skl-defs}, we define the notion of public key encryption with secure key leasing (PKE-SKL) and its various security notions. We also show several general relationships among those security notions. In Section \ref{sec:PKE-CoIC}, we define and construct Public Key Encryption with $\cKLA$ security. In Section \ref{sec:PKE-SKL}, we provide our construction of PKE with secure key leasing. In Section \ref{sec:ibe_abe_skl} and Section \ref{sec:PKFE_SKL} we provide our construction of Attribute Based Encryption with secure key leasing and public key Functional Encryption with secure key leasing respectively. 
\fi


\section{Preliminaries}\label{sec:prelim}
\paragraph{Notations and conventions.}
In this paper, standard math or sans serif font stands for classical algorithms (e.g., $C$ or $\algo{Gen}$) and classical variables (e.g., $x$ or $\keys{pk}$).
Calligraphic font stands for quantum algorithms (e.g., $\qalgo{Gen}$) and calligraphic font and/or the bracket notation for (mixed) quantum states (e.g., $\qstate{q}$ or $\ket{\psi}$).

Let $[\ell]$ denote the set of integers $\{1, \cdots, \ell \}$, $\secp$ denote a security parameter, and $y \seteq z$ denote that $y$ is set, defined, or substituted by $z$.
For a finite set $X$ and a distribution $D$, $x \chosen X$ denotes selecting an element from $X$ uniformly at random, $x \chosen D$ denotes sampling an element $x$ according to $D$. Let $y \gets \algo{A}(x)$ and $y \gets \qalgo{A}(\qstate{x})$ denote assigning to $y$ the output of a probabilistic or deterministic algorithm $\algo{A}$ and a quantum algorithm $\qalgo{A}$ on an input $x$ and $\qstate{x}$, respectively. When we explicitly show that $\algo{A}$ uses randomness $r$, we write $y \gets \algo{A}(x;r)$.
PPT and QPT algorithms stand for probabilistic polynomial-time algorithms and polynomial-time quantum algorithms, respectively.
Let $\negl$ denote a negligible function.
For strings $x,y\in \bit^n$, $x\cdot y$ denotes $\bigoplus_{i\in[n]} x_i y_i$ where $x_i$ and $y_i$ denote the $i$th bit of $x$ and $y$, respectively. 

\subsection{Standard Cryptographic Tools}\label{sec:standard_crypto}

\paragraph{Secret-key encryption.}

\begin{definition}[Secret Key Encryption]\label{def:ske}
An SKE scheme $\SKE$ is a two tuple $(\E, \D)$ of PPT algorithms. 
\begin{itemize}
\item The encryption algorithm $\E$, given a key $K \in \bin^\lambda$ and a plaintext $m \in \cM$, outputs a ciphertext $\ct$,
where $\cM$ is the plaintext space of $\SKE$.

\item The decryption algorithm $\D$, given a key $K$ and a ciphertext $\ct$, outputs a plaintext $\tilde{m} \in \{ \bot \} \cup \cM$.
This algorithm is deterministic.
\end{itemize}
We require $\SKE$ to satisfy correctness.
\begin{description}
\item[Correctness:] We require $\D(K, \E(K, m)) = m$ for every $m \in \cM$ and key $K \in \bin^\lambda$.
\end{description}
\end{definition}

\begin{definition}[Ciphertext Pseudorandomness for SKE]\label{def:ske_pseudorandomct}
Let $\zo{\ell}$ be the ciphertext space of $\SKE$.
We define the following experiment $\expb{\qA,\SKE}{pr}{ct}(1^\secp,\coin)$ between a challenger and an adversary $\qA$.
\begin{enumerate}
\item The challenger generates $K \chosen \bin^\lambda$.
Then, the challenger sends $1^\lambda$ to $\qA$.
\item $\qA$ may make polynomially many encryption queries adaptively.
$\qA$ sends $m \in \cM $ to the challenger.
Then, the challenger returns $\ct \la \E(K, m)$ if $\coin=0$, otherwise $\ct \chosen \zo{\ell}$.
\item $\qA$ outputs $\coin' \in \bin$.
\end{enumerate}
We say that $\SKE$ is pseudorandom-secure if for any QPT adversary $\qA$, we have 
\[
\advb{\SKE, \qA}{pr}{ct}(\lambda) = \abs{\Pr[\expb{\qA,\SKE}{pr}{ct}(1^\secp,0)=1]-\Pr[\expb{\qA,\SKE}{pr}{ct}(1^\secp,1)=1]} \le \negl(\secp).
\]
\end{definition}

\begin{theorem}\label{thm:pseudorandom_ske}
If OWFs exist, there exists a pseudorandom-secure SKE scheme.
\end{theorem}

\paragraph{Public-key encryption.}
\begin{definition}[PKE]
A PKE scheme $\PKE$ is a tuple of three algorithms $(\KG, \Enc, \Dec)$. 
Below, let $\cX$  be the message space of $\PKE$.
\begin{description}
\item[$\KG(1^\secp)\ra(\ek,\dk)$:] The key generation algorithm takes a security parameter $1^\lambda$, and outputs an encryption key $\ek$ and a decryption key $\dk$. 

\item[$\Enc(\ek,\msg)\ra\ct$:] The encryption algorithm takes an encryption key $\ek$ and a message $\msg \in \cX$, and outputs a ciphertext $\ct$.

\item[$\Dec(\dk,\ct)\ra\tilde{\msg}$:] The decryption algorithm is a deterministic algorithm that takes a decryption key $\dk$ and a ciphertext $\ct$, and outputs a value $\tilde{\msg}$.

\item[Correctness:]For every $\msg \in \cX$, we have
\begin{align}
\Pr\left[
\Dec(\dk, \ct) \allowbreak = \msg
\ \middle |
\begin{array}{ll}
(\ek,\dk)\gets\KG(1^\secp)\\
\ct\gets\Enc(\ek,\msg)
\end{array}
\right] 
=1-\negl(\secp).
\end{align}
\end{description}
\end{definition}

\begin{definition}[IND-CPA Security]\label{def:IND-CPA_PKE}
We say that a PKE scheme $\PKE$ with the message space $\Xs$ is IND-CPA secure if it satisfies the following requirement, formalized from the experiment $\expb{\PKE,\qA}{ind}{cpa}(1^\secp,\coin)$ between an adversary $\qA$ and a challenger:
        \begin{enumerate}
            \item  The challenger runs $(\ek,\dk)\gets\KG(1^\secp)$ and sends $\ek$ to $\qA$. 
            \item $\qA$ sends $(\msg_0^*,\msg_1^*)\in \Xs^2$ 
            to the challenger. 
            \item The challenger generates $\ct^*\la\Enc(\ek,\msg_\coin^*)$ and sends $\ct^*$ to $\qA$.
            \item $\qA$ outputs a guess $\coin^\prime$ for $\coin$. The challenger outputs $\coin'$ as the final output of the experiment.
        \end{enumerate}
        For any QPT $\qA$, it holds that
\begin{align}
\advb{\PKE,\qA}{ind}{cpa}(\secp) \seteq \abs{\Pr[\expb{\PKE,\qA}{ind}{cpa} (1^\secp,0) \ra 1] - \Pr[\expb{\PKE,\qA}{ind}{cpa} (1^\secp,1) \ra 1] }\leq \negl(\secp).
\end{align}
\end{definition}

\begin{definition}[OW-CPA Security]\label{def:OW-CPA_PKE}
We say that a PKE scheme $\PKE$ with the message space $\Xs$ is OW-CPA secure if it satisfies the following requirement, formalized from the experiment $\expb{\PKE,\qA}{ow}{cpa}(1^\secp)$ between an adversary $\qA$ and a challenger:
        \begin{enumerate}
            \item  The challenger runs $(\ek,\dk)\gets\KG(1^\secp)$, chooses $\msg^*\gets \Xs$, 
            runs $\ct^*\gets \Enc(\ek,\msg^*)$, 
            and sends $(\ek,\ct^*)$ to $\qA$.  
            \item $\qA$ sends $\msg'\in \Xs$ 
            to the challenger. 
            \item The challenger outputs $1$ if $\msg'=\msg^*$ and otherwise $0$ as the final output of the experiment. 
        \end{enumerate}
        For any QPT $\qA$, it holds that
\begin{align}
\advb{\PKE,\qA}{ow}{cpa}(\secp) \seteq \Pr[\expb{\PKE,\qA}{ow}{cpa} (1^\secp) \ra 1]\leq \negl(\secp).
\end{align}
\end{definition}

It is well-known that IND-CPA security implies OW-CPA security if $|\Xs|$ is super-polynomial. 

\paragraph{Pseudorandom functions.}

\begin{definition}[Puncturable PRF]\label{def:pprf}
A puncturable PRF (PPRF) is a tuple of algorithms $\PuncPRF = (\prfgen, \prf,\Puncture)$ where $\{\prf_{\prfkey}: \bin^{\ell_1} \ra \zo{\ell_2} \mid \prfkey \in \zo{\secp}\}$ is a PRF family and satisfies the following two conditions. Note that $\ell_1$ and $\ell_2$ are polynomials of $\secp$.
   \begin{description}
       \item[Punctured correctness:] For any polynomial-size set $S \subseteq \zo{\ell_1}$ and any $x\in \zo{\ell_1} \setminus S$, it holds that
       \begin{align}
       \Pr[\prf_{\prfkey}(x) = \prf_{\prfkey_{\notin S}}(x)  \mid \prfkey \gets \prfgen(1^{\secp}),
       \prfkey_{\notin S} \gets \Puncture(\prfkey,S)]=1.
       \end{align}
       \item[Pseudorandom at punctured point:] For any polynomial-size set $S \subseteq\zo{\ell_1}$
       and any QPT distinguisher $\qA$, it holds that
       \begin{align}
       \vert
       \Pr[\qA(\prf_{\prfkey_{\notin S}},\{\prf_{\prfkey}(x_i)\}_{x_i\in S}) \ra 1] -
       \Pr[\qA(\prf_{\prfkey_{\notin S}}, (\cU_{\ell_2})^{\abs{S}}) \ra 1]
       \vert \leq \negl(\secp),
       \end{align}
       where $\prfkey\gets \prfgen(1^{\secp})$,
       $\prfkey_{\notin S} \gets \Puncture(\prfkey,S)$ and $\cU_{\ell_2}$ denotes the uniform distribution over $\zo{\ell_2}$.
   \end{description}
   If $S = \setbk{x^\ast}$ (i.e., puncturing a single point), we simply write $\prf_{\ne x^\ast}(\cdot)$ instead of $\prf_{\prfkey_{\notin S}}(\cdot)$ and consider $\prf_{\ne x^\ast}$ as a keyed function.
\end{definition}

It is easy to see that the Goldwasser-Goldreich-Micali tree-based construction of PRFs (GGM PRF)~\cite{JACM:GolGolMic86} from OWF yield puncturable PRFs where the size of the punctured key grows polynomially with the size of the set $S$ being punctured~\cite{AC:BonWat13,PKC:BoyGolIva14,CCS:KPTZ13}. Thus, we have:
\begin{theorem}[\cite{JACM:GolGolMic86,AC:BonWat13,PKC:BoyGolIva14,CCS:KPTZ13}]\label{thm:pprf-owf} If OWFs exist, then for any polynomials $\ell_1(\secp)$ and $\ell_2(\secp)$, there exists a PPRF that maps $\ell_1$-bits to $\ell_2$-bits.
\end{theorem}

\paragraph{Garbling schemes.}
\begin{definition}[Garbling schemes]\label{def:GC}
A garbling scheme $\GC$ is a tuple of PPT algorithms $\GC=(\Garble,\GCEval)$.
\begin{description}
\item[$\Garble(1^\secp, C) \to (\{\lab_{i,b}\}_{i\in [\ell],b\in\bit},\widetilde{C})$:] The garbling algorithm takes a security parameter $1^\secp$ and a circuit $C$ and outputs labels $\{\lab_{i,b}\}_{i\in [\ell],b\in\bit}$ and garbled version of the circuit $\widetilde{C}$,
where $\ell$ is the input length of $C$. 
\item[$\GCEval(\widetilde{C},\{ \lab_i \}_{i\in [\ell]} )\ra z $:] The evaluation algorithm $\GCEval$ takes the garbled circuit $\widetilde{C}$ and labels $\{ \lab_i \}_{i\in [\ell]}$ and outputs an evaluation result $z$.

\item[Correctness:] We require that
\[
\Pr\left[
\GCEval(\widetilde{C},\{ \lab_{i,x_i} \}_{i\in [\ell]} ) = C(x)
 \ \middle |
\begin{array}{rl}
 \Garble(1^\secp, C)) \to (\{\lab_{i,b}\}_{i\in [\ell],b\in\bit},\widetilde{C}) 
\end{array}
\right]=1 -\negl(\secp)
\]
holds for all $\ell\in \mathbb{N}$, $x\in \bit^\ell$ and $C$ with input length $\ell$,
where $x_i$ is the $i$-th bit of $x$.

\item[Security:] We require that there exists a PPT algorithm $\SimGC$ such the following distributions are computationally indistinguishable for all $\ell\in\mathbb{N}$, $x\in \bit^\ell$, and circuit $C$ with input length $\ell$:   
\[
(\{ \lab_{i,x_i}\}_{i\in [\ell]}, \widetilde{C} )
\approx_c 
\SimGC( 1^\secp, \mathsf{info}(C), C(x) )
\]
where $\Garble(1^\secp, C)) \to (\{\lab_{i,b}\}_{i\in [\ell],b\in\bit},\widetilde{C}) $
and $\mathsf{info}(C)$ refers to the size of $C$, input and output lengths of $C$.
\end{description}
We note that we will drop $\mathsf{info}(C)$ from the inputs to $\SimGC$ when it is clear from the context.
\end{definition}

\begin{theorem}\cite{FOCS:Yao86,JC:LinPin09}
If there exists a one-way function, there exists secure garbling scheme. 
\end{theorem}

\paragraph{Attribute-based encryption.}

\begin{definition}[Attribute-Based Encryption]\label{def:ABE}
An ABE scheme $\ABE$ is a tuple of four PPT algorithms $(\Setup, \KG, \Enc, \Dec)$. 
Below, let $\cX=\{\cX_\secp\}_\secp$, $\cY=\{\cY_\secp\}_\secp$, and $R=\{R_\secp:\cX_\secp \times \cY_\secp \ra \bin \}_\secp$ be the ciphertext attribute space, key attribute space, and the relation associated with $\ABE$, respectively.
We note that we will abuse the notation and occasionally drop the subscript for these spaces for notational simplicity.
We also note that the message space is set to be $\bin^\ell$ below. 
\begin{description}
\item[$\Setup(1^\secp)\ra(\pk,\msk)$:] The setup algorithm takes a security parameter $1^\secp$ and outputs a public key $\pk$ and master secret key $\msk$.
\item[$\KG(\msk,y)\ra\sk_y$:] The key generation algorithm $\KG$ takes a master secret key $\msk$ and a key attribute $y \in \cY$, and outputs a decryption key $\sk_y$.

\item[$\Enc(\pk,x,m)\ra\ct$:] The encryption algorithm takes a public key $\pk$, a ciphertext attribute $x \in \cX$, and a message $x$, and outputs a ciphertext $\ct$.

\item[$\Dec(\sk_y,x,\ct)\ra z$:] The decryption algorithm takes a secret key $\sk_f$, a ciphertext attribute $x$, and the corresponding ciphertext $\ct$ and outputs $z \in \{ \bot \} \cup \bin^\ell$.

\item[Correctness:] We require that
\[
\Pr\left[
\Dec(\sk_y, x, \ct) = m
 \ \middle |
\begin{array}{rl}
 &(\pk,\msk) \la \Setup(1^\secp),\\
 & \sk_y \gets \KG(\msk,y), \\
 &\ct \gets \Enc(\pk,x,m)
\end{array}
\right]=1 -\negl(\secp).
\]
holds for all $x\in \cX$ and $y\in \cY$ such that $R(x,y)=1$ and $m\in \bin^\ell$. 
\end{description}
\end{definition}

\begin{definition}[Adaptive Security for ABE]\label{def:ad_ind_ABE}
We say that $\ABE$ is an \emph{adaptively secure} ABE scheme for relation $R:\cX\times \cY \to \bin$, if it satisfies the following requirement, formalized from the experiment $\expb{\qA}{ada}{ind}(1^\secp,\coin)$ between an adversary $\qA$ and a challenger:
        \begin{enumerate}
            \item The challenger runs $(\pk,\msk)\gets\Setup(1^\secp)$ and sends $\pk$ to $\qA$.
            \item $\qA$ sends arbitrary key queries. That is, $\qA$ sends a key attribute $y \in\cY$ to the challenger and the challenger responds with $\sk_{y}\gets \KG(\msk,y)$ for the query.
            \item At some point, $\qA$ sends $(x,m_0,m_1)$ to the challenger. If $R(x,y)=0$ for all queried $y$, the challenger generates a ciphertext $\ct^*\gets\Enc(\pk,x,m_\coin)$. The challenger sends $\ct^*$ to $\qA$.
            \item Again, $\qA$ can send key queries $y$ such that $R(x,y)=0$.
            \item $\qA$ outputs a guess $\coin^\prime$ for $\coin$.
            \item The experiment outputs $\coin^\prime$.
        \end{enumerate}
        We say that $\ABE$ is adaptively secure if, for any QPT $\qA$, it holds that
\begin{align}
\advb{\ABE,\qA}{ada}{ind}(\secp) \seteq \abs{\Pr[\expb{\ABE,\qA}{ada}{ind} (1^\secp,0) \ra 1] - \Pr[\expb{\ABE,\qA}{ada}{ind} (1^\secp,1) \ra 1] }\leq \negl(\secp).
\end{align}
\end{definition}

\begin{definition}[Selective Security for ABE]\label{def:sel_ind_ABE}
We also define selective security for $\ABE$. For doing so, we consider the same security game as that for adaptive security except that the adversary $\qA$ should declare its target $x$ at the beginning of the game (even before it is given $\pk$).
We then define the advantage $\advb{\ABE,\qA}{sel}{ind}(\secp)$ for the selective security similarly. We say $\ABE$ is selectively indistinguishably-secure if for any QPT adversary $\qA$, $\advb{\ABE,\qA}{sel}{ind}(\secp)$ is negligible.    
\end{definition}

By setting $\cX$, $\cY$, and $R$ appropriately, we can recover important classes of ABE.
In particular, if we set $\cX_\secp=\cY_\secp=\bin^*$ and define $R$ so that $R(x,y)=1$ if $x=y$ and $R(x,y)=0$ otherwise, we recover the definition of identity-based encryption (IBE).
If we set $\cX_\secp=\bin^{n(\secp)}$ and $\cY_\secp$ to be the set of all circuits with input space $\bin^{n(\secp)}$ and depth at most $d(\secp)$, where $n$ and $d$ are some polynomials, and define $R$ so that $R(x,y)=y(x)$, we recover the definition of ABE for circuits. 

\paragraph{Functional encryption.}

\begin{definition}[Secret-Key Functional Encryption]\label{def:SKFE}
An SKFE scheme $\SKFE$ is a tuple of four PPT algorithms $(\Setup, \KG, \Enc, \Dec)$. 
Below, let $\cX$, $\cY$, and $\cF$ be the plaintext, output, and function spaces $\SKFE$, respectively.
\begin{description}
\item[$\Setup(1^\secp)\ra\msk$:] The setup algorithm takes a security parameter $1^\lambda$, and outputs a master secret key $\msk$.
\item[$\KG(\msk,f)\ra\sk_f$:] The key generation algorithm takes a master secret key $\msk$ and a function $f \in \calF$, and outputs a functional decryption key $\sk_f$.

\item[$\Enc(\msk,x)\ra\ct$:] The encryption algorithm takes a master secret key $\msk$ and a plaintext $x \in \cX$, and outputs a ciphertext $\ct$.

\item[$\Dec(\sk_f,\ct)\ra y$:] The decryption algorithm takes a functional decryption key $\sk_f$ and a ciphertext $\ct$, and outputs  $y \in \{ \bot \} \cup \cY$.

\item[Correctness:] We require that for every $x \in \cX$, $f \in \calF$, $\numkey\in\bbN$, we have that
\[
\Pr\left[
\Dec(\sk_f, \ct) = f(x)
 \ \middle |
\begin{array}{rl}
 &\msk \la \Setup(1^\secp),\\
 & \sk_f \gets \KG(\msk,f), \\
 &\ct \gets \Enc(\msk,x)
\end{array}
\right]=1 - \negl(\secp).
\]
\end{description}
\end{definition}

\begin{definition}[Function Privacy]\label{def:full_func_privacy_SKFE}
We formalize the experiment $\expb{\qA,\SKFE}{full}{fp}(1^\secp,\coin)$ between an adversary $\qA$ and a challenger for SKFE scheme for $\Xs,\Ys$, and $\Fs$ as follows:
        \begin{enumerate}
            \item At the beginning, the challenger runs $\msk\gets\Setup(1^\secp)$.
            Throughout the experiment, $\qA$ can access the following oracles.
            \begin{description}
            \item[$\Oracle{\Enc}(x_0,x_1)$:] Given $(x_0,x_1)$, it returns $\Enc(\msk,x_\coin)$.
            \item[$\Oracle{\KG}(f_0,f_1)$:] Given $(f_0,f_1)$, it returns $\KG(\msk,f_\coin)$.
            \end{description}
            \item If the following happens during the oracle queries above, the experiment aborts: $f_0(x_0) \ne f_1(x_1)$ or $\abs{x_0} \ne \abs{x_1}$ or $\abs{f_0}\ne \abs{f_1}$.
            \item $\qA$ outputs a guess $\coin^\prime$ for $\coin$. The challenger outputs $\coin'$ as the final output of the experiment.
        \end{enumerate}
        We say that $\SKFE$ is fully function private if, for any QPT $\qA$, it holds that
\begin{align}
\advb{\SKFE,\qA}{full}{fp}(\secp) \seteq \abs{\Pr[\expb{\SKFE,\qA}{full}{fp} (1^\secp,0) \ra 1] - \Pr[\expb{\SKFE,\qA}{full}{fp} (1^\secp,1) \ra 1] }\le \negl(\secp).
\end{align}
If $\qA$ can access $\Oracle{\Enc}$ only once in $\expb{\SKFE,\qA}{full}{fp}$, we say that $\SKFE$ is adaptively single-ciphertext function private.
\end{definition}

\begin{theorem}[\cite{C:GorVaiWee12,JC:BraSeg18,C:ABSV15,TCC:AnaVai19}]\label{thm:1ct_adaptive_function_private_SKFE}
If there exist OWFs, there exists adaptively single-ciphertext function private SKFE for $\Ppoly$.
\end{theorem}

\begin{definition}[Public-Key Functional Encryption]\label{def:PKFE}
A PKFE scheme $\PKFE$ is a tuple of four PPT algorithms $(\Setup, \KG, \Enc, \Dec)$. 
Below, let $\cX$, $\cY$, and $\cF$ be the plaintext, output, and function spaces of $\PKFE$, respectively.
\begin{description}
\item[$\Setup(1^\secp)\ra(\pk,\msk)$:] The setup algorithm takes a security parameter $1^\secp$ and outputs a public key $\pk$ and master secret key $\msk$.
\item[$\KG(\msk,f)\ra\sk_f$:] The key generation algorithm $\KG$ takes a master secret key $\msk$ and a function $f \in \cF$, and outputs a functional decryption key $\sk_f$.

\item[$\Enc(\pk,x)\ra\ct$:] The encryption algorithm takes a public key $\pk$ and a message $x \in \cX$, and outputs a ciphertext $\ct$.

\item[$\Dec(\sk_f,\ct)\ra y$:] The decryption algorithm takes a functional decryption key $\sk_f$ and a ciphertext $\ct$, and outputs $y \in \{ \bot \} \cup \cY$.

\item[Correctness:] We require 
we have that
\[
\Pr\left[
\Dec(\sk_f, \ct) = f(x)
 \ \middle |
\begin{array}{rl}
 &(\pk,\msk) \la \Setup(1^\secp),\\
 & \sk_f \gets \KG(\msk,f), \\
 &\ct \gets \Enc(\pk,x)
\end{array}
\right]=1 -\negl(\secp).
\]
\end{description}
\end{definition}

\begin{definition}[Adaptive Security for PKFE]\label{def:ad_ind_PKFE}
We formalize the experiment $\expb{\qA}{ada}{ind}(1^\secp,\coin)$ between an adversary $\qA$ and a challenger for PKFE scheme for $\Xs,\Ys$, and $\Fs$ as follows:
        \begin{enumerate}
            \item The challenger runs $(\pk,\msk)\gets\Setup(1^\secp)$ and sends $\pk$ to $\qA$.
            \item $\qA$ sends arbitrary key queries. That is, $\qA$ sends function $f_{i}\in\Fs$ to the challenger and the challenger responds with $\sk_{f_i}\gets \KG(\msk,f_i)$ for the $i$-th query $f_{i}$.
            \item At some point, $\qA$ sends $(x_0,x_1)$ to the challenger. If $f_i(x_0)=f_i(x_1)$ for all $i$, the challenger generates a ciphertext $\ct^*\gets\Enc(\pk,x_\coin)$. The challenger sends $\ct^*$ to $\qA$.
            \item Again, $\qA$ can sends function queries $f_i$ such that $f_i(x_0)=f_i(x_1)$.
            \item $\qA$ outputs a guess $\coin^\prime$ for $\coin$.
            \item The experiment outputs $\coin^\prime$.
        \end{enumerate}
        We say that $\PKFE$ is adaptively secure if, for any QPT $\qA$, it holds that
\begin{align}
\advb{\PKFE,\qA}{ada}{ind}(\secp) \seteq \abs{\Pr[\expb{\PKFE,\qA}{ada}{ind} (1^\secp,0) \ra 1] - \Pr[\expb{\PKFE,\qA}{ada}{ind} (1^\secp,1) \ra 1] }\leq \negl(\secp).
\end{align}
If $\qA$ can send only $q$ key queries in $\expb{\PKFE,\qA}{ada}{ind}$ where $q$ is a bounded polynomial, we say that $\PKFE$ is $q$-bounded adaptively secure.
\end{definition}

\begin{theorem}[\cite{C:GorVaiWee12,TCC:AnaVai19}]\label{thm:bounded_PKFE_from_PKE}
If there exists IND-CPA secure PKE, there exists $q$-bounded adaptively secure PKFE for $\Ppoly$.
\end{theorem}
\begin{remark}
We defined FE as key-policy FE (KPFE) here. There is another type of FE called ciphertext-policy FE (CPFE). Since we use CPFE only as a building block of the CoIC-KLA secure PKE scheme in \cref{sec:PKE-CoIC}, we defer its definition to \cref{sec:tools}. 
\end{remark}

\subsection{Useful Lemmata}
The following lemma is taken verbatim from \cite[Lemma 2.1]{C:BonZha13}. 
\begin{lemma}[{\cite[Lemma 2.1]{C:BonZha13}}]\label{lem:BZ}
Let $\qA$ be a quantum algorithm, and let $\Pr[x]$ be the probability that $\qA$ outputs $x$. Let
$\qA'$ be another quantum algorithm obtained from $\qA$ by pausing $\qA$ at an arbitrary stage of execution, performing a partial measurement that obtains one of $k$ outcomes, and then resuming $\qA$. Let $\Pr'[x]$
be the probability $\qA'$ outputs $x$. Then $\Pr'
[x] \ge \Pr[x]/k$. 
\end{lemma}

We will also need the quantum Goldreich-Levin lemma established by \cite{C:CLLZ21} based on \cite{AC02}.
\begin{lemma}[Quantum Goldreich-Levin with Quantum Auxiliary Input~{\cite[Lemma~B.12]{C:CLLZ21}}]\label{lem:QGL}
There exists a QPT algorithm $\qExt$ that satisfies the following.
Let $n \in \mathbb{N}$, $x\in \bit^n$, $\epsilon\in [0,1/2]$,  and $\qA$ be a quantum algorithm with a quantum auxiliary input $\qaux$ such that    
\begin{align}
    \Pr\left[\qA(\qaux,r)\ra x\cdot r ~\middle |~ r\gets \bit^n\right]
    \geq \frac{1}{2}+\epsilon.
\end{align}
Then, we have
\begin{align}
    \Pr\left[\qExt([\qA],\qaux)\rightarrow x \right]\geq 4\epsilon^2.
\end{align}
where $[\qA]$ means the description of $\qA$. 
\end{lemma}

\section{Public Key Encryption with Secure Key Leasing} 
\label{sec:pke-skl-defs}
In this section, we define the notion of public key encryption with secure key leasing (PKE-SKL) and its various security notions. Then we show several general relationships among those security notions.  
\subsection{Definitions}
The syntax of PKE-SKL is defined as follows. 
\begin{definition}[PKE with Secure Key Leasing]
A PKE-SKL scheme $\PKESKL$ is a tuple of four algorithms $(\qKG, \Enc, \qDec,\qVrfy)$. 
Below, let $\cX$  be the message space of $\PKESKL$. 
\begin{description}
\item[$\qKG(1^\secp)\ra(\ek,\qdk,\vk)$:] The key generation algorithm takes a security parameter $1^\lambda$, and outputs an encryption key $\ek$, a decryption key $\qdk$, and a verification key $\vk$.

\item[$\Enc(\ek,\msg)\ra\ct$:] The encryption algorithm takes an encryption key $\ek$ and a message $\msg \in \cX$, and outputs a ciphertext $\ct$.

\item[$\qDec(\qdk,\ct)\ra\tilde{\msg}$:] The decryption algorithm takes a decryption key $\qdk$ and a ciphertext $\ct$, and outputs a value $\tilde{\msg}$.


\item[$\qVrfy(\vk,\widetilde{\qdk})\ra\top/\bot$:] The verification algorithm takes a verification key $\vk$ and a (possibly malformed) decryption key $\widetilde{\qdk}$, and outputs $\top$ or $\bot$.

\item[Decryption correctness:]For every $\msg \in \cX$, we have
\begin{align}
\Pr\left[
\qDec(\qdk, \ct) \allowbreak = \msg
\ \middle |
\begin{array}{ll}
(\ek,\qdk,\vk)\gets\qKG(1^\secp)\\
\ct\gets\Enc(\ek,\msg)
\end{array}
\right] 
=1-\negl(\secp).
\end{align}

\item[Verification correctness:] We have 
\begin{align}
\Pr\left[
\qVrfy(\vk,\qdk)=\top
\ \middle |
\begin{array}{ll}
(\ek,\qdk,\vk)\gets\qKG(1^\secp)
\end{array}
\right] 
=1-\negl(\secp).
\end{align}
\end{description}
\end{definition}
\begin{remark}\label{rem:reusability}
We can assume without loss of generality that a decryption key of a PKE-SKL scheme is reusable, i.e., it can be reused to decrypt (polynomially) many ciphertexts. In particular, we can asusme that 
for honestly generated $\ct$ and $\qdk$, if we decrypt $\ct$ by using $\qdk$, the state of the decryption key after the decryption is negligibly close to that before the decryption in terms of trace distance. 
This is because the output of the decryption is almost deterministic by decryption correctness, and thus such an operation can be done without almost disturbing the input state by the gentle measurement lemma~\cite{Winter99}.    
A similar remark applies to all variants of PKE-SKL  (IBE, ABE, and FE with SKL) defined in this paper.    
\end{remark}
\begin{remark}
Though we are the first to define PKE with secure key leasing, SKFE with secure key leasing was already defined by Kitagawa and Nishimaki~\cite{AC:KitNis22}. 
The above definition is a natural adaptation of their definition with the important difference that we do not require classical certificate of deletion.
\end{remark}


We define several security notions for PKE-SKL.
The first is a natural indistinguishability security definition, which is our primary taget. 
\begin{definition}[IND-KLA Security]\label{def:IND-CPA_PKESKL}
We say that a PKE-SKL scheme $\PKESKL$  with the message space $\Xs$ is IND-KLA secure, if it satisfies the following requirement, formalized from the experiment $\expb{\PKESKL,\qA}{ind}{kla}(1^\secp,\coin)$ between an adversary $\qA$ and a challenger $\qC$:
        \begin{enumerate}
            \item  $\qC$ runs $(\ek,\qdk,\vk)\gets\qKG(1^\secp)$ and sends $\ek$ and $\qdk$ to $\qA$. 
            \item Throughout the experiment, $\qA$ can access the following (stateful) verification oracle $\Oracle{\qVrfy}$ where $V$ is initialized to be $\unreturned$:
            \begin{description}
               \item[ $\Oracle{\qVrfy}(\widetilde{\qdk})$:] It runs $d \gets \Vrfy(\vk,\widetilde{\qdk})$ and returns $d$.  
               If $V=\bot$ and $d=\top$, it updates $V\seteq \top$. 
            \end{description}
            \item $\qA$ sends $(\msg_0^*,\msg_1^*)\in \Xs^2$
            to $\qC$. If $V=\unreturned$, $\qC$ output $0$ as the final output of this experiment. Otherwise, $\qC$ generates $\ct^*\la\Enc(\ek,\msg_\coin^*)$ and sends $\ct^*$ to $\qA$.
            \item $\qA$ outputs a guess $\coin^\prime$ for $\coin$. $\qC$ outputs $\coin'$ as the final output of the experiment.
        \end{enumerate}
        For any QPT $\qA$, it holds that
\begin{align}
\advb{\PKESKL,\qA}{ind}{kla}(\secp) \seteq \abs{\Pr[\expb{\PKESKL,\qA}{ind}{kla} (1^\secp,0) \ra 1] - \Pr[\expb{\PKESKL,\qA}{ind}{kla} (1^\secp,1) \ra 1] }\leq \negl(\secp).
\end{align} 
We say that $\PKESKL$ is 1-query IND-KLA secure if the above holds for any QPT $\qA$ that makes at most one query to $\Oracle{\qVrfy}$. 
\end{definition}
\begin{remark}\label{rem:no-post-challenge-query}
When we consider a 1-query adversary, we can assume that its query is made before receiving the challenge ciphertext $\ct^*$ without loss of generality. This is because otherwise the experiment always outputs $0$. 
\end{remark}
\begin{remark}
By a standard hybrid argument, one can show that IND-KLA security implies multi-challenge IND-KLA security where the adversary is allowed to request arbitrarily many challenge ciphertexts. Thus, if we have an IND-KLA secure PKE-SKL scheme for single-bit messages, we can extend the plaintext length to an arbitrary polynomial by bit-by-bit encryption.  \takashi{I added this remark.}
\end{remark}


We also define the one-way variant of the above security. 
\begin{definition}[OW-KLA Security]\label{def:OW-CPA_PKESKL}
We say that a PKE-SKL scheme $\PKESKL$  with the message space $\Xs$ is OW-KLA secure, if it satisfies the following requirement, formalized from the experiment $\expb{\PKESKL,\qA}{ow}{kla}(1^\secp)$ between an adversary $\qA$ and a challenger $\qC$:
        \begin{enumerate}
            \item  $\qC$ 
            runs $(\ek,\qdk,\vk)\gets\qKG(1^\secp)$ and sends $\ek$ and $\qdk$ to $\qA$. 
              \item Throughout the experiment, $\qA$ can access the following (stateful) verification oracle $\Oracle{\qVrfy}$ where $V$ is initialized to be $\unreturned$:
              \begin{description}
               \item[   $\Oracle{\qVrfy}(\widetilde{\qdk})$:] It runs $d \gets \Vrfy(\vk,\widetilde{\qdk})$ and returns $d$. 
               If $V=\bot$ and $d=\top$, it updates $V\seteq \top$. 
            \end{description}
            \item 
            $\qA$ sends $\requestchallenge$
            to $\qC$. If $V=\unreturned$, $\qC$ outputs $0$ as the final output of this experiment. Otherwise, $\qC$ 
            chooses $\msg^*\gets \Xs$, 
            generates $\ct^*\la\Enc(\ek,\msg^*)$ and sends $\ct^*$ to $\qA$.
            \item $\qA$ outputs $\msg$. $\qC$ outputs $1$ if $\msg=\msg^*$ and otherwise outputs $0$ as the final output of the experiment.
        \end{enumerate}
        For any QPT $\qA$, it holds that
\begin{align}
\advb{\PKESKL,\qA}{ow}{kla}(\secp) \seteq \Pr[\expb{\PKESKL,\qA}{ow}{kla} (1^\secp) \ra 1]\leq \negl(\secp).
\end{align}
We say that $\PKESKL$ is 1-query OW-KLA secure if the above holds for any QPT $\qA$ that makes at most one query to $\Oracle{\qVrfy}$. 
\end{definition}

Similar to normal PKE, IND-KLA security implies OW-KLA security if $|\Xs|$ is super-polynomial in $\secp$.

Finally, we define a security notion which we call one-more unreturnability (OMUR), which requires that an adversary given a single copy of the decryption key cannot pass the verification more than once. Though this does not seem very meaningful by itself, this is a useful intermediate tool for our final goal of constructing IND-KLA secure scheme.  
\begin{definition}[One-More Unreturnability]\label{def:OMUR_PKESKL}
We say that a PKE-SKL scheme $\PKESKL$  with the message space $\Xs$ satisfies One-More UnReturnability (OMUR), if it satisfies the following requirement, formalized from the experiment $\expa{\PKESKL,\qA}{omur}(1^\secp)$ between an adversary $\qA$ and a challenger $\qC$:
        \begin{enumerate}
            \item  $\qC$ 
            runs $(\ek,\qdk,\vk)\gets\qKG(1^\secp)$ and sends $\ek$ and $\qdk$ to $\qA$. 
         \item Throughout the experiment, $\qA$ can access the following (stateful) verification oracle $\Oracle{\qVrfy}$ where $\mathsf{count}$ is initialized to be $0$:
              \begin{description}
               \item[$\Oracle{\qVrfy}(\widetilde{\qdk})$:] 
            It runs $d \gets \Vrfy(\vk,\widetilde{\qdk})$ and
            returns $d$.  
            It updates $\mathsf{count}\seteq \mathsf{count}+1$ if $d=\top$. 
            \end{description}
            \item 
            $\qA$ sends $\mathsf{Finish}$
            to $\qC$. If $\mathsf{count}\ge 2$, $\qC$ outputs $1$ and $0$ otherwise
            as the final output of this experiment. 
        \end{enumerate}
        For any QPT $\qA$, it holds that
\begin{align}
\adva{\PKESKL,\qA}{omur}(\secp) \seteq \Pr[\expa{\PKESKL,\qA}{omur} (1^\secp) \ra 1]\leq \negl(\secp).
\end{align} 
\end{definition}

\subsection{Relationships among Security Notions}
We show several relationships among different security notions for PKE-SKL. In particular, 
we show the following theorem.
\begin{theorem}\label{thm:one-query-ow_to_full}
If there exists a 1-query OW-KLA secure PKE-SKL scheme, there exists an IND-KLA secure PKE-SKL scheme. 
\end{theorem} 
This theorem simplifies our task: For constructing a ($\poly$-query) IND-KLA secure scheme, it suffices to construct a 1-query OW-KLA secure scheme. 
We construct a 1-query OW-KLA secure scheme in \cref{sec:PKE-SKL}. 

We prove \Cref{thm:one-query-ow_to_full} in the following three steps.
\begin{enumerate}
    \item Give a conversion to add OMUR to any 1-query OW-KLA secure scheme (\cref{lem:add_OMUR}).
    \item Convert a 1-query OW-KLA secure scheme that satisfies OMUR to a 1-query IND-KLA secure scheme that satisfies OMUR (\cref{lem:ow-ind}).
    \item Show that any 1-query IND-KLA secure scheme that satisfies OMUR is IND-KLA secure (\cref{lem:one-to-poly_IND}).
\end{enumerate}

It is clear that  \Cref{thm:one-query-ow_to_full} follows from  \cref{lem:add_OMUR,lem:ow-ind,lem:one-to-poly_IND}. 
We prove them in the following.

\begin{lemma}
\label{lem:add_OMUR}
If there exists a 1-query OW-KLA secure PKE-SKL scheme, then there exists a 1-query OW-KLA secure PKE-SKL scheme that satisfies OMUR. 
\end{lemma}
\begin{remark}
This lemma is actually not needed for the purpose of this paper since our construction of a 1-query OW-KLA secure PKE-SKL scheme in \cref{sec:PKE-SKL} already satisfies OMUR as mentioned in \cref{rem:OMUR}. We include this lemma in the paper because this general reduction may be useful in future works.  
\end{remark}
\begin{proof}[Proof of \cref{lem:add_OMUR}]
Let $\OWSKL=(\OWSKL.\qKG, \OWSKL.\Enc, \OWSKL.\qDec,\OWSKL.\qVrfy)$ be a 1-query OW-KLA secure PKE-SKL scheme with the message space $\Xs$. 
We assume that a decryption key of $\OWSKL$ is reusable in the sense of \cref{rem:reusability} 
and $\vk$ contains $\ek$ 
without loss of generality. 
Then we consider a modified PKE-SKL scheme $\OWSKL'=(\OWSKL'.\qKG, \OWSKL'.\Enc, \OWSKL'.\qDec,\OWSKL'.\qVrfy)$ with the same message space $\Xs$ defined as follows. 
The algorithms $\OWSKL'.\qKG$, $\OWSKL'.\Enc$, and $\OWSKL'.\qDec$ are identical to $\OWSKL.\qKG$, $\OWSKL.\Enc$, and $\OWSKL.\qDec$, respectively. The algorithm $\OWSKL'.\qVrfy$ works as follows:

\begin{description}
\item[$\OWSKL'.\qVrfy(\vk,\widetilde{\qdk})$:] On input a verification key $\vk$ and a (possibly malformed) decryption key $\widetilde{\qdk}$, do the following:
\begin{description}
\item[Decryptability verification:] Choose $\msg\gets \Xs$ and run $\ct\gets \Enc(\ek,\msg)$ and $\msg'\gets \Dec(\widetilde{\qdk},\ct)$. 
If $\msg'\ne \msg$, return $\bot$. 
\item[Original verification:]  Otherwise, let $\widetilde{\qdk}'$ be the state of the decryption key after running the decryption algorithm. 
Run $\OWSKL.\qVrfy(\vk,\widetilde{\qdk}')$ and return whatever $\OWSKL.\qVrfy$ returns. 
\end{description} 
\end{description}
\paragraph{Correctness.}
The decryption correctness of $\OWSKL'$ follows from that of $\OWSKL$ because the only difference between these schemes is the verification algorithm, which is irrelevant to the decryption correctness. 
The verification correctness of $\OWSKL'$ follows from that of $\OWSKL$ because we assume that $\OWSKL$ has reusable decryption keys and thus $\widetilde{\qdk}'$  in $\OWSKL'.\qVrfy$ has a negligible trace distance from $\widetilde{\qdk}$, which passes $\OWSKL.\qVrfy$ except for a negligible probability by the verification correctness of $\OWSKL$.  

\paragraph{1-query OW-SKL security.}
The 1-query OW-SKL security of $\OWSKL'$ follows from that of $\OWSKL$ by a straightforward reduction. Specifically, let $\qA$ be an QPT adversary that breaks the 1-query OW-SKL security of $\OWSKL'$. Then, we construct a QPT adversary $\qB$ that breaks the 1-query OW-SKL security of $\OWSKL$ as follows:
\begin{description}
\item[$\qB(\ek,\qdk)$:]
Run $\qA(\ek,\qdk)$ until $\qA$ makes a verification query $\widetilde{\qdk}$. For simulating the verification oracle to $\qA$, 
choose $\msg\gets \Xs$, run $\ct\gets \Enc(\ek,\msg)$ and $\msg'\gets \Dec(\widetilde{\qdk},\ct)$, and let $\widetilde{\qdk}'$ be the state of the decryption key after running the decryption algorithm. 
If $\msg'\ne \msg$, output $0$ and immediately halt. 
Otherwise, query $\widetilde{\qdk}'$ to its own verification oracle, and forward the response to $\qA$.
When $\qA$ sends $\requestchallenge$, forward it to the external challenger to receive $\ct^*$ and forward it to $\qA$. Run $\qA$ until it halts and output whatever $\qA$ outputs. 
\end{description}
We can see that the experiment which $\qB$ plays outputs $1$ if and only if the (simulated) experiment which $\qA$ plays outputs $1$. Therefore, $\qB$ breaks the 1-query OW-SKL security of $\OWSKL$. Thus, the 1-query OW-SKL security of $\OWSKL'$ follows from that of $\OWSKL$.

\paragraph{OMUR.}
In the following, we show that $\OWSKL'$ satisfies OMUR. Let $\qA$ be a QPT adversary against the OMUR of  $\OWSKL'$ that makes $Q=\poly(\secp)$ verification queries. Then we consider the following sequence of hybrids.
\begin{description}
\item[$\hybi{0}$:] This is identical to the experiment $\expa{\OWSKL',\qA}{omur}(1^\secp)$ as defined in \cref{def:OMUR_PKESKL}. 

Note that we have 
\begin{align}
    \Pr[\hybi{0}=1]=\adva{\OWSKL',\qA}{omur}(\secp).
\end{align}

\item[$\hybi{1}$:] This is identical to $\hybi{0}$ except that the challenger uniformly chooses integers $1\le i_1 < i_2\le Q$ at the beginning of the experiment and outputs $1$ if and only if $i_1$-th and $i_2$-th verification queries are the first two queries to which the verification oracle returned $\top$. 

Whenever $\hybi{0}$ returns $1$, there are at least $2$ verification queries accepted by the verification oracle.  
Therefore, when we uniformly choose $1\le i_1 < i_2\le Q$, the probability that $i_1$-th and $i_2$-th queries are the first two queries to be accepted is ${\binom{Q}{2}}^{-1}=\frac{2}{Q(Q-1)}$. Therefore we have 
\begin{align}
    \Pr[\hybi{1}= 1]=\frac{2}{Q(Q-1)}  \Pr[\hybi{0}= 1].
\end{align}

\item[$\hybi{2}$:] This is identical to $\hybi{1}$ except that the verification oracle just returns $\bot$ 
without running the verification algorithm 
to $i$-th query for all $i\in [i_2-1]\setminus \{i_1\}$ and the experiment halts right after running the verification oracle for the $i_2$-th query where it outputs $1$ if and only if the verification oracle returned $\top$ to both $i_1$-th and $i_2$-th queries. 

When $\hybi{1}$ returns $1$, the verification oracle returns $\bot$ to $i$-th query for all $i\in [i_2-1]\setminus \{i_1\}$ since otherwise $i_1$-th and $i_2$-th queries cannot be the first $2$ queries to be accepted. Therefore, these hybrids are identical until $\qA$ makes $i_2$-th query when $\hybi{1}$ returns $1$.\footnote{\label{footnote:measurements} Note that there is a superficial difference that the verification oracle of $\hybi{1}$ runs the verification algorithm to $i$-th query for all $i\in [i_2-1]\setminus \{i_1\}$ in $\hybi{1}$ but it does not in $\hybi{2}$. However since these query registers are not used at all for generating the output of $\hybi{2}$, the difference of if measurements are applied on them cannot affect the probability to output $1$. 
} Moreover, $\hybi{2}$ outputs $1$ whenever $\hybi{1}$ outputs $1$ if we run the rest of $\qA$ to complete $\hybi{1}$.  Therefore, we have 
\begin{align}
    \Pr[\hybi{2}= 1]\ge \Pr[\hybi{1}= 1].
\end{align}

\item[$\hybi{3}$:] This is identical to $\hybi{2}$ except that the experiment outputs $1$ if and only if $i_1$-th query is accepted and $i_2$-th query passes the ``Decryptability verification" part of $\OWSKL'.\qVrfy$, i.e., $\msg=\msg'$ in the notation of the description of $\OWSKL'.\qVrfy$. 

Since the condition to output $1$ is just relaxed, we have 
\begin{align}
    \Pr[\hybi{3}= 1]\ge \Pr[\hybi{2}= 1].
\end{align}
\end{description}
Below, we prove 
\begin{align}
\Pr[\hybi{3}= 1]=\negl(\secp).
\end{align}
To prove this, we consider the following QPT adversary $\qB$ against the 1-query OW-SKL security of $\OWSKL$ that works as follows:
\begin{description}
\item[$\qB(\ek,\qdk)$:]
Uniformly choose integers $1\le i_1 < i_2\le Q$ and run $\qA(\ek,\qdk)$ until it makes $i_2$-th query where the response by the verification oracle to $\qA$'s $i$-th query for $i\in[i_2-1]$ is simulated as follows: If $i\ne i_1$, return $\bot$ as the response from the verification oracle. If $i=i_1$, forward the query to its own verification oracle and forward the response to $\qA$. Let $\qdk_{i_2}$ be $\qA$'s $i_2$-th verification query.  
Send $\requestchallenge$ to the external challenger to receive $\ct^*$. Run $\msg'\gets \OWSKL.\Dec(\ct^*,\qdk)$ and output $\msg'$.
\end{description}
By the definitions of $\hybi{3}$ and $\qB$, we can see that 
\begin{align}
   \advb{\OWSKL,\qB}{ow}{kla}(\secp) =\Pr[\hybi{3}= 1]. 
\end{align}
Thus, we have $\Pr[\hybi{3}= 1]=\negl(\secp)$ by the 1-query OW-SKL security of $\OWSKL$.

Combining the above, we have $\adva{\OWSKL',\qA}{omur}(\secp)=\negl(\secp)$, which means that $\OWSKL$ satisfies OMUR. This completes the proof of \Cref{lem:add_OMUR}. 
\end{proof}

\begin{lemma}
\label{lem:ow-ind}
If there exists a 1-query OW-KLA secure PKE-SKL scheme, then there exists a 1-query IND-KLA secure PKE-SKL scheme. Moreover, if the base scheme satisfies OMUR, then then the resulting scheme satisfies OMUR.
\end{lemma}
\begin{proof}
Let $\OWSKL=(\OWSKL.\qKG, \OWSKL.\Enc, \OWSKL.\qDec,\OWSKL.\qVrfy)$ be a 1-query OW-KLA secure PKE-SKL scheme with the message space $\bit^n$ that satisfies OMUR. Then, we construct an IND-KLA secure PKE-SKL scheme  $\INDSKL=(\INDSKL.\qKG, \INDSKL.\Enc, \INDSKL.\qDec,\INDSKL.\qVrfy)$ with the message space $\bit$ as follows. 
\begin{description}
\item[$\INDSKL.\qKG(1^\secp)\ra(\ek,\qdk,\vk)$:] On input the security parameter $1^\lambda$, 
run $(\ek,\qdk,\vk)\gets \OWSKL.\qKG(1^\secp)$ and output $(\ek,\qdk,\vk)$.

\item[$\INDSKL.\Enc(\ek,\msg)\ra\INDSKL.\ct$:]On input an encryption key $\ek$ and a message $\msg \in \bit$, 
choose $r,x\gets \bit^n$, 
generate $\OWSKL.\ct\gets \OWSKL.\Enc(\ek,x)$, 
set $b \seteq (x\cdot r)\oplus \msg$,
and output a ciphertext $\INDSKL.\ct\seteq (\OWSKL.\ct,r,b)$.

\item[$\INDSKL.\qDec(\qdk,\INDSKL.\ct)\ra\tilde{\msg}$:] On input a decryption key $\qdk$ and a ciphertext $\INDSKL.\ct= (\OWSKL.\ct,r,b)$, 
compute $\tilde{x}\gets \OWSKL.\qDec(\qdk,\OWSKL.\ct)$
and output $\tilde{\msg}\seteq (\tilde{x}\cdot r)\oplus b$.

\item[$\INDSKL.\qVrfy(\vk,\widetilde{\qdk})\ra\top/\bot$:] On input a verification key $\vk$ and a (possibly malformed) decryption key $\widetilde{\qdk}$, run $\OWSKL.\qVrfy(\vk,\widetilde{\qdk})$ and
output whatever $\OWSKL.\qVrfy$ outputs.  
\end{description}
The decryption correctness and verification correctness of $\INDSKL$ immediately follow from those of $\OWSKL$. 
The OMUR of $\INDSKL$ immediately follows from that of $\OWSKL$ since their key generation and verification algorithms are identical and the definition of OMUR only depends on these algorithms.   
In the following, we prove that $\INDSKL$ is IND-KLA secure assuming that $\OWSKL$ is OW-KLA secure. 
Toward contradiction, suppose that $\INDSKL$ is not IND-KLA secure. Then, there is a QPT adversary $\qA$ such that $\advb{\INDSKL,\qA}{ind}{kla}(\secp)$ is non-negligible. Without loss of generality, we assume that 
\begin{align} \label{eq:IND_A_wins}
\Pr_{\coin\gets \bit}[\expb{\INDSKL,\qA}{ind}{kla} (1^\secp,\coin) \ra \coin]\ge 1/2+\epsilon(\secp)
\end{align}
for a non-negligible $\epsilon(\secp)$.  
Since $\INDSKL$ is a bit encryption, we assume that the challenge message pair $(\msg_0,\msg_1)$ is $(0,1)$ without loss of generality.  
We divide $\qA$ into the following two stages $\qA_0$ and $\qA_1$: 
\begin{description}
\item[$\qA_0^{\Oracle{\qVrfy}}(\ek,\qdk)\rightarrow \qst_\qA$:] Upon receiving $(\ek,\qdk)$ from $\qC$, makes a single query to $\Oracle{\qVrfy}$ and outputs a quantum state $\qst_\qA$. 
\item[$\qA_1(\qst_\qA,\INDSKL.\ct)\rightarrow \coin'$:] Upon receiving the state $\qst_\A$ from $\qA_0$ and $\INDSKL.\ct=(\OWSKL.\ct,r,b)$ from $\qC$, output $\coin'$.  
\end{description} 
We remark that we can assume that $\qA_1$ does not make any query to $\Oracle{\qVrfy}$ without loss of generality by \cref{rem:no-post-challenge-query}. 

We have 
\begin{align}
&\Pr_{\coin\gets \bit}[\expb{\INDSKL,\qA}{ind}{kla} (1^\secp,\coin) \ra \coin]\\
&=\Pr_{\coin\gets \bit}[\expb{\INDSKL,\qA}{ind}{kla} (1^\secp,\coin) \ra \coin \wedge V=\returned]\\
&+\Pr_{\coin\gets \bit}[\expb{\INDSKL,\qA}{ind}{kla} (1^\secp,\coin) \ra \coin \wedge V=\unreturned]\\
&=
\Pr[V=\returned]
\cdot \Pr_{\coin\gets \bit}[\expb{\INDSKL,\qA}{ind}{kla} (1^\secp,\coin) \ra \coin \mid V=\returned]\\
&+\frac{1}{2}(1-\Pr[V=\returned]). \label{eq:advantage_A_decompose}
\end{align}
By \cref{eq:IND_A_wins,eq:advantage_A_decompose}, we have\footnote{We can assume  $\Pr[V=\returned]\ne 0$ since otherwise  \cref{eq:IND_A_wins} cannot be satisfied.}  
\begin{align}\label{eq:conditioned_on_returned_prob}
 \Pr_{\coin\gets \bit}[\expb{\INDSKL,\qA}{ind}{kla} (1^\secp,\coin) \ra \coin \mid V=\returned]\ge \frac{1}{2}+\frac{\epsilon(\secp)}{\Pr[V=\returned]}.
\end{align}

Then, we construct an adversary $\qB=(\qB_0,\qB_1)$ against OW-KLA security of $\OWSKL$ that works as follows. 
\begin{description}
\item[$\qB_0^{\Oracle{\qVrfy}}(\ek,\qdk)\rightarrow \qst_\qA$:]  
This is identical to $\qA_0$. Specifically, run $\qst_\qA\gets \qA_0^{\Oracle{\qVrfy}}(\ek,\qdk)$ and output  $\qst_\qA$.
\item[$\qB_1(\qst_\qA,\OWSKL.\ct)\rightarrow x$:] Upon receiving 
$\qst_\qA$ from $\qB_0$, send $\requestchallenge$ to $\qC$ and receive 
$\OWSKL.\ct$ from $\qC$. 
Then set $\qaux\seteq (\qst_\qA,\OWSKL.\ct)$ and define an algorithm $\qA'$ as follows.
\begin{description}
\item[$\qA'(\qaux,r)$:]
On input $\qaux=(\qst_\qA,\OWSKL.\ct)$ and $r\in \bit^n$,  choose $b\gets \bit$, 
set $\INDSKL.\ct=(\OWSKL.\ct,r,b)$,  
run $\coin' \gets \qA_1(\qst_\qA,\INDSKL.\ct)$, and output $\coin'\oplus b$.  
\end{description}
Run $x\gets \qExt([\qA'],\qaux)$, and output $x$ where 
$\qExt$ is the algorithm as in \Cref{lem:QGL} and
$[\qA']$ is the description of $\qA'$.
\end{description}

In the following, we show that $\qB$ breaks OW-KLA security of $\OWSKL$. 
Let $\qG$ be an algorithm that works as follows.
\begin{description}
\item[$\qG(1^\secp)$:] 
Generate $(\ek,\qdk,\vk)\gets\OWSKL.\qKG(1^\secp)$, 
$\qst_\qA\gets \qA_0^{\Oracle{\qVrfy}}(\ek,\qdk)$,   
$x\gets \bit^n$, and 
$\OWSKL.\ct\gets \OWSKL.\Enc(\ek,x)$.
Let $V\seteq \returned$ if the response to $\qA_0$'s query (which is assumed to be made once) is $\top$ and $V\seteq \unreturned$ otherwise.    
Output $(V,\qst_\qA,\OWSKL.\ct,x)$. 
\end{description}

By \cref{eq:conditioned_on_returned_prob} and a standard averaging argument,  
for at least $\frac{\epsilon(\secp)}{2\Pr[V=\returned]}$-fraction of $(V,\qst_\qA,\OWSKL.\ct,x)$ generated by $\qG(1^\secp)$ conditioned on $V=\returned$, we have 
\begin{align} 
    \Pr\left[
    \qA_1(\qst_\qA,\INDSKL.\ct)\rightarrow \coin
    \right]\ge \frac{1}{2}+\frac{\epsilon(\secp)}{2\Pr[V=\returned]}\ge \frac{1}{2}+\frac{\epsilon(\secp)}{2}
\end{align}
where $\coin\gets \bit$, 
    $r\gets \bit^n$,  
    $b\seteq (x\cdot r)\oplus \coin$, and  
    $\INDSKL.\ct= (\OWSKL.\ct,r,b).$
    
Therefore, for at least $\frac{\epsilon(\secp)}{2}$-fraction of $(V,\qst_\qA,\OWSKL.\ct,x)$ generated by $\qG(1^\secp)$, we have 
\begin{align}\label{eq:IND_A_wins_rewrite_again} 
    \Pr\left[V=\returned \wedge
    \qA_1(\qst_\qA,\INDSKL.\ct)\rightarrow \coin
    \right]\ge \frac{1}{2}+\frac{\epsilon(\secp)}{2}
\end{align}
where $\coin\gets \bit$, 
    $r\gets \bit^n$,  
    $b\seteq (x\cdot r)\oplus \coin$, and  
    $\INDSKL.\ct= (\OWSKL.\ct,r,b).$

For such $(V,\qst_\qA,\OWSKL.\ct,x)$, if we let $\qaux=(\qst_\qA,\OWSKL.\ct)$, \cref{eq:IND_A_wins_rewrite_again} directly implies  
\begin{align}
\Pr_{r\gets \bit^n}\left[\qA'(\qaux,r)\ra x\cdot r\right]\geq \frac{1}{2}+\frac{\epsilon(\secp)}{2}.
\end{align}
Therefore, by \Cref{lem:QGL}, we have 
\begin{align} \label{eq:ext_succeed}
    \Pr\left[\qExt([\qA'],\qaux)\rightarrow x \right]\geq \epsilon(\secp)^2.
\end{align}
Since \cref{eq:ext_succeed} \emph{and} $V=\returned$ hold at the same time for at least $\frac{\epsilon(\secp)}{2}$-fraction of $(V,\qst_\qA,\OWSKL.\ct,x)$, we have  
\begin{align}
    \Pr_{ (\mathsf{v},\qst_\qA,\OWSKL.\ct,x)\gets \qG(1^\secp) }\left[
    V=\returned
    \wedge
    \qB_1(\qst_\qA,\OWSKL.\ct)\rightarrow x 
    \right]\ge \frac{\epsilon(\secp)^3}{2}.
\end{align} 
By the definitions of $\qB=(\qB_0,\qB_1)$ and $\qG$ and the assumption that $\epsilon(\secp)$ is non-negligible, this implies that $\qB$ breaks OW-KLA security of $\OWSKL$. 
\end{proof}
\begin{remark}[On Multiple-Query Case]
In the above reduction, it is important that $\qA_1$ can be assumed to not make any verification query because otherwise we cannot apply the quantum Goldreich-Levin theorem~(\cref{lem:QGL}). In the 1-query setting, this can be assumed without loss of generality by \cref{rem:no-post-challenge-query}. In the multiple-query setting, we cannot assume it in general. If we assume that the base scheme satisfies OMUR, we can assume it without loss of generality because post-challenge verification queries are useless for such schemes. However, we do not know how to resolve the issue in the multiple-query setting without relying on OMUR.       
\end{remark}


\begin{lemma}\label{lem:one-to-poly_IND}
If a PKE-SKL scheme is 1-query IND-KLA secure and satisfies OMUR,  then it is IND-KLA  secure.
\end{lemma}
\begin{proof} 
Let $\PKESKL=(\qKG, \Enc, \qDec,\qVrfy)$ be an IND-KLA secure PKE-SKL scheme that satisfies OMUR. 
For a QPT adversary $\qA$ against IND-KLA security of $\PKESKL$ that makes $Q=\poly(\secp)$ verification queries and $\coin\in \bit$, we consider the following sequence of hybrids.

\begin{description}
\item[$\hybi{0}^{\coin}$:] This is identical to $\expb{\PKESKL,\qA}{ind}{kla}(1^\secp,\coin)$. 

Note that our goal is to prove 
\begin{align}
    \left|\Pr[\hybi{0}^{0}=1]-\Pr[\hybi{0}^{1}=1]\right|=\negl(\secp).
\end{align}

\item[$\hybi{1}^{\coin}$:] This is identical to $\hybi{0}^{\coin}$ except that the verification oracle returns $\bot$ to all queries made after it returns $\top$ once. 

By the OMUR of $\PKESKL$, we have 
\begin{align}
    \left|\Pr[\hybi{1}^{\coin}\rightarrow 1]-\Pr[\hybi{0}^{\coin}\rightarrow 1]\right|=\negl(\secp)
\end{align}
for $\coin\in \bit$. 

\item[$\hybi{2}^{\coin}$:] This is identical to $\hybi{1}^{\coin}$ except that the challenger chooses $i^*\gets [Q]$ at the beginning of the game, the verification oracle just returns $\bot$ without running the verification algorithm to $i$-th query for $i\ne i^*$, and the experiment returns $0$ if the verification oracle returns $\bot$ to $i^*$-th query. 

Note that there is exactly one verification query to be accepted in $\hybi{1}^{\coin}$ whenever it returns $1$. If $i^*$ is the correct guess for such query, which occurs with probability $\frac{1}{Q}$, then $\hybi{2}^{\coin}$ is identical to $\hybi{1}^{\coin}$.\footnote{A similar remark to \cref{footnote:measurements} applies here.}  
Moreover, $\hybi{2}^{\coin}$ outputs $0$ when the guess is incorrect. 
 Therefore, we have 
\begin{align}
    \Pr[\hybi{2}^{\coin}\rightarrow 1]= \frac{1}{Q}\Pr[\hybi{1}^{\coin}\rightarrow 1].
\end{align}
\end{description}

Below, we prove 
\begin{align}
   \left|\Pr[\hybi{2}^{0}= 1]-\Pr[\hybi{2}^{1}= 1]\right|=\negl(\secp).  
\end{align}
To prove this, we consider a QPT adversary $\qB$ against $1$-query IND-KLA security of $\PKESKL$ that works as follows.
\begin{description}
\item[$\qB(\ek,\qdk)$:] 
Choose $i^*\gets [Q]$ and 
run $\qA(\ek,\qdk)$ where the $i^*$-th query is forwarded to its own verification oracle and responded according to the response from the oracle  while all the other queries are responded by $\bot$. 
When $\qA$ sends $(\msg_0^*,\msg_1^*)$, forward it to the external challenger, receive $\ct^*$ from the challenger, and forward it to $\qA$.
Finally, output whatever $\qA$ outputs.   
\end{description} 

By the definitions of $\qB$ and $\hybi{2}^{\coin}$, one can see that  
\begin{align}
    \Pr[\expb{\PKESKL,\qB}{ind}{kla} (1^\secp,\coin) \ra 1]=\Pr[\hybi{2}^{\coin}= 1]
\end{align}
for $\coin\in \bit$.
Therefore, we have 
\begin{align}
 &\abs{\Pr[\hybi{2}^{0}= 1] - \Pr[\hybi{2}^{1}= 1]}\\
 &=\abs{\Pr[\expb{\PKESKL,\qB}{ind}{kla} (1^\secp,0) \ra 1] - \Pr[\expb{\PKESKL,\qB}{ind}{kla} (1^\secp,1) \ra 1]}\\
 &=\negl(\secp)
\end{align} 
by the $1$-query IND-SKL security of $\PKESKL$. 

Combining the above, we have 
\begin{align}
    \left|\Pr[\hybi{0}^{0}=1]-\Pr[\hybi{0}^{1}=1]\right|=\negl(\secp).
\end{align}
This completes the proof of \cref{lem:one-to-poly_IND}. 
\end{proof}


\newcommand{\Drev}{D^{\mathtt{rev}}}
\newcommand{\cPrev}{\cP^{\mathtt{rev}}}

\newcommand{\qExtract}{\qalgo{Extract}}
\newcommand{\API}{\qalgo{API}}
\newcommand{\TI}{\qalgo{TI}}
\newcommand{\ATI}{\qalgo{ATI}}
\newcommand{\recover}{\algo{Recover}}
\newcommand{\shiftdis}[1]{\Delta_{\mathsf{Shift}}^{#1}}
\newcommand{\qevaluator}{\ket{\cE}}
\newcommand{\projimp}{\algo{ProjImp}}
\newcommand{\cproj}{\algo{CProj}}
\newcommand{\Live}{\mathsf{Live}}
\newcommand{\BadExt}{\mathsf{BadExt}}
\newcommand{\GoodExt}{\mathsf{GoodExt}}
\newcommand{\sflive}{\mathsf{live}}
\newcommand{\sfbad}{\mathsf{bad}}
\newcommand{\sfgood}{\mathsf{good}}
\newcommand{\ketisu}{\ket{\mathds{1}_\randspace}}
\newcommand{\braisu}{\bra{\mathds{1}_\randspace}}
\newcommand{\kbisu}{\ketisu\braisu}
\newcommand{\IsU}{\algo{IsU_{\randspace}}}
\newcommand{\MixM}{\algo{MixM}}

\section{Public Key Encryption with CoIC-KLA Security}
\label{sec:PKE-CoIC}
In this section, we introduce a new security notion called CoIC-KLA security for PKE, and construct a PKE scheme that satisfies it based on any IND-CPA secure PKE scheme. Looking ahead, it is used as a building block of our construction of PKE-SKL in~\cref{sec:PKE-SKL}.

\subsection{Tools}\label{sec:tools}

We first introduce some tools used in this section.

\paragraph{Measurement Implementation.}
We review some notions related to measurement implementations used in the definition and the security proof of CoIC-KLA security.

\begin{definition}[Projective Implementation]\label{def:projective_implementation}
Let:
\begin{itemize}
 \item $\cD$ be a finite set of distributions over an index set $\cI$.
 \item $\cP=\setbk{\mat{P}_i}_{i\in \cI}$ be a positive operator valued measure (POVM).
 \item $\cE = \setbk{\mat{E}_D}_{D\in\cD}$ be a projective measurement with index set $\cD$.
 \end{itemize}
 We consider the following measurement procedure.
 \begin{enumerate}
 \item Measure under the projective measurement $\cE$ and obtain a distribution $D$.
 \item Output a random sample from the distribution $D$.
 \end{enumerate}
 We say $\cE$ is the projective implementation of $\cP$, denoted by $\projimp(\cP)$, if the measurement process above is equivalent to $\cP$.
\end{definition}

\begin{theorem}[{\cite[Lemma 1]{TCC:Zhandry20}}]\label{lem:commutative_projective_implementation}
Any binary outcome POVM $\cP=(\mat{P},\mat{I}-\mat{P})$ has a unique projective implementation $\projimp(\cP)$.
\end{theorem}


\begin{definition}[Shift Distance]\label{def:shift_distance}
For two distributions $D_0,D_1$, the shift distance with parameter $\epsilon$, denoted by $\shiftdis{\epsilon}(D_0,D_1)$, is the smallest quantity $\delta$ such that for all $x \in \R$:
\begin{align}
\Pr[D_0\le x] & \le \Pr[D_1\le x + \epsilon] + \delta,&& \Pr[D_0\ge x]  \le \Pr[D_1\ge x - \epsilon] + \delta,\\
\Pr[D_1\le x] & \le \Pr[D_0\le x + \epsilon] + \delta,&& \Pr[D_1\ge x]  \le \Pr[D_0\ge x - \epsilon] + \delta.
\end{align}
For two real-valued measurements $\cM$ and $\cN$ over the same quantum system, the shift distance between $\cM$ and $\cN$ with parameter $\epsilon$ is
\[
\shiftdis{\epsilon}(\cM,\cN)\seteq \sup_{\ket{\psi}}\shiftdis{\epsilon}(\cM(\ket{\psi}),\cN(\ket{\psi})).
\]
\end{definition}

\begin{definition}[Mixture of Projetive Measurement~\cite{TCC:Zhandry20}]\label{def:mixture_projective_measurement}
Let $D: \cR \ra \cI$ where $\cR$ and $\cI$ are some sets.
Let $\setbk{(\mat{P}_i,\mat{Q}_i)}_{\in \cI}$ be a collection of binary projective measurement.
The mixture of projective measurements associated to $\cR$, $\cI$, $D$, and $\setbk{(\mat{P}_i,\mat{Q}_i)}_{\in \cI}$ is the binary POVM $\cP_D =(\mat{P}_D,\mat{Q}_D)$ defined as follows.
\begin{align}
& \mat{P}_D = \sum_{i\in\cI}\Pr[i \chosen D(R)]\mat{P}_i && \mat{Q}_D = \sum_{i\in\cI}\Pr[i \chosen D(R)]\mat{Q}_i,
\end{align}
where $R$ is uniformly distributed in $\cR$.
\end{definition}

\begin{theorem}[\cite{TCC:Zhandry20,EC:KitNis22}]\label{thm:api_property}
Let $D$ be any probability distribution and $\cP=\setbk{(\Pi_i,\mat{I} -\Pi_i)}_i$ be a collection of binary outcome projective measurements. For any $0<\epsilon,\delta<1$, there exists an algorithm of measurement $\API_{\cP,\cD}^{\epsilon,\delta}$ that satisfies the following.
\begin{itemize}
\item $\shiftdis{\epsilon}(\API_{\cP,D}^{\epsilon,\delta},\projimp(\cP_D))\le \delta$.  
\item $\API_{\cP,D}^{\epsilon,\delta}$ is $(\epsilon,\delta)$-almost projective in the following sense. For any quantum state $\ket{\psi}$, we apply $ \API_{\cP,D}^{\epsilon ,\delta}$ twice in a row to $\ket{\psi}$ and obtain measurement outcomes $x$ and $y$, respectively. Then, $\Pr[\abs{x-y}\le \epsilon]\ge 1-\delta$.
\item $\API_{\cP,D}^{\epsilon,\delta}$ is $(\epsilon,\delta)$-reverse almost projective in the following sense. For any quantum state $\ket{\psi}$, we apply $ \API_{\cP,D}^{\epsilon ,\delta}$ and $\API_{\cPrev,D}^{\epsilon,\delta}$ in a row to $\ket{\psi}$ and obtain measurement outcomes $x$ and $y$, respectively, where $\cPrev=\setbk{(\mat{I} -\Pi_i,\Pi_i)}_i$. Then, $\Pr[\abs{(1-x)-y}\le \epsilon]\ge 1-\delta$.
\item The expected running time of $\API_{\cP,D}^{\epsilon,\delta}$ is $T_{\cP,D}\cdot \poly(1/\epsilon,\log(1/\delta))$ where $T_{\cP,D}$ is the combined running time of $D$, the procedure mapping $i \ra (\mat{P}_i,\mat{I}- \mat{P}_i)$, and the running time of measurement $(\mat{P}_i,\mat{I}-\mat{P}_i)$.
\end{itemize}
\end{theorem}


\begin{theorem}[{\cite[Corollary 1]{TCC:Zhandry20}}]\label{cor:cind_sample_api}
Let $\qstateq$ be an efficiently constructible, potentially mixed state, and $D_0,D_1$ efficiently sampleable distributions.
If $D_0$ and $D_1$ are computationally indistinguishable, for any inverse polynomial $\epsilon$ and any function $\delta$, we have $\shiftdis{3\epsilon}(\API_{\cP,D_0}^{\epsilon,\delta}(\qstateq),\API_{\cP,D_1}^{\epsilon,\delta}(\qstateq)) \le 2\delta + \negl(\secp)$.
\end{theorem}

\begin{definition}[Quantum Program with Classical Inputs and Outputs~\cite{C:ALLZZ21}]\label{def:Q_program_C_IO}
A quantum program with classical inputs is a pair of quantum state $\qstateq$ and unitaries $\setbk{\mat{U}_x}_{x\in[N]}$ where $[N]$ is the domain, such that the state of the program evaluated on input $x$ is equal to $\mat{U}_x \qstateq \mat{U}_x^\dagger$. We measure the first register of $\mat{U}_x \qstateq \mat{U}_x^\dagger$ to obtain an output. We say that $\setbk{\mat{U}_x}_{x\in[N]}$ has a compact classical description $\mat{U}$ when applying $\mat{U}_x$ can be efficiently computed given $\mat{U}$ and $x$.
\end{definition}

\paragraph{Ciphertext-Policy Functional Encryption.} We review the definition of ciphertext-policy functional encryption (CPFE) that we use as the building block of our CoIC-KLA secure PKE scheme.
\begin{definition}[Ciphertext-Policy Functional Encryption]\label{def:cpfe}
A CPFE scheme for the circuit space $\cC$ and the input space $\cX$ is a tuple of algorithms $(\Setup, \KG, \Enc, \Dec)$.
\begin{itemize}
\item The setup algorithm $\Setup$ takes as input a security parameter $1^\lambda$, and outputs a master public key $\MPK$ and master secret key $\MSK$.
\item The key generation algorithm $\KG$ takes as input the master secret key $\MSK$ and $x\in\cX$, and outputs a decryption key $\sk_x$.

\item The encryption algorithm $\Enc$ takes as input the master public key $\MPK$ and $C\in\cC$, and outputs a ciphertext $\ct$.

\item The decryption algorithm $\Dec$ takes as input a functional decryption key $\sk_x$ and a ciphertext $\ct$, and outputs $y$.
\end{itemize}

\begin{description}
\item[Decryption Correctness:] We require $\Dec(\KG(\MSK, x), \Enc(\MPK, C)) \allowbreak = C(x)$ for every $C\in\cC$, $x\in\cX$, and $\left(\MPK,\MSK\right) \la \Setup(1^\lambda)$.
\end{description}

\end{definition}

Next, we introduce $1$-bounded security for CPFE schemes.

\begin{definition}[$1$-Bounded Security]\label{def:CPFE_security}
Let $\CPFE$ be a CPFE scheme.  
We define the game $\expt{\qA,\CPFE}{1\textrm{-}bounded}(\secp,\coin)$ as follows.

\begin{enumerate}
\item The challenger generates $(\MPK,\MSK) \la \Setup(1^\lambda)$ and sends $\MPK$ to $\qA$.
$\qA$ sends $x\in\cX$ to the challenger. The challenger generates $\sk_x\la\KG(\MSK,x)$ and sends $\sk_x$ to $\qA$.

\item $\qA$ outputs $(C_0,C_1)$ such that $C_0(x)=C_1(x)$ and $C_0$ and $C_1$ have the same size. The challenger picks $\coin\la\bit$, generates $\ct\la\Enc(\MPK,C_\coin)$, and sends $\ct$ to $\qA$.

\item $\qA$ outputs $\coin' \in \bin$.
\end{enumerate}
We say that $\CPFE$ is $1$-bounded secure if for every QPT $\qA$, we have
\begin{align}
\adva{\qA,\CPFE}{1\textrm{-}bounded}(\secp)=
2\abs{\Pr[
\expt{\qA,\CPFE}{1\textrm{-}bounded}(\secp)=1
] -\frac{1}{2}} =\negl(\secp).
\end{align}

\end{definition}

\begin{theorem}[\cite{C:GorVaiWee12}]\label{thm:CPFE_from_PKE}
If there exists IND-CPA secure PKE, there exists $1$-bounded secure CPFE for $\Ppoly$.\footnote{Though \cite{C:GorVaiWee12} present their construction as KPFE instead of CPFE, it is easy to see that they implicitly give CPFE.}
\end{theorem}

%
%
%
%
%

\subsection{Definitions of CoIC-KLA Security}

We introduce definitions of CoIC-KLA security.
In addition to normal CoIC-KLA security needed to realize our PKE-SKL, we also define what we call strong CoIC-KLA security.
We can prove that strong CoIC-KLA security implies CoIC-KLA security.
The reason we introduce strong CoIC-KLA is that it is more compatible to our construction strategy in \cref{sec:scKLA_PKE} that uses watermarking technique by Kitagawa and Nishimaki~\cite{EC:KitNis22}.

\begin{definition}[CoIC-KLA Security]\label{def:cKLA_PKE}
We say that a PKE scheme $\PKE$ with the message space $\Xs$ is CoIC-KLA secure, if it satisfies the following requirement, formalized from the experiment $\expb{\PKE,\qA}{coic}{kla}(1^\secp)$ between an adversary $\qA$ and a challenger $\qC$:
        \begin{enumerate}
            \item  $\qC$ runs $(\ek_0,\dk_0)\gets\KG(1^\secp)$ and $(\ek_1,\dk_1)\gets\KG(1^\secp)$, and generates $\qdk:=\frac{1}{\sqrt{2}}(\ket{0}\ket{\dk_0}+\ket{1}\ket{\dk_1})$. $\qC$ sends $\ek_0$, $\ek_1$, and $\qdk$ to $\qA$. $\qA$ can get access to the following oracle only once.
            \begin{description}
            \item[$\cO(\widetilde{\qdk})$:]On input a possibly malformed decryption key $\widetilde{\qdk}$, it applies a binary-outcome measurement $(\mat{I}-\Pi_{\mathrm{vrfy}},\Pi_{\mathrm{vrfy}})$, where $\Pi_{\mathrm{vrfy}}$ is the projection to the right decryption key, i.e., 
            \begin{align}
 \Pi_{\mathrm{vrfy}}\seteq 
 \left(\frac{1}{\sqrt{2}}\left(\ket{0}\ket{\dk_{0}}+\ket{1}\ket{\dk_{1}}\right)\right)  
 \left(\frac{1}{\sqrt{2}}\left(\bra{0}\bra{\dk_{0}}+\bra{1}\bra{\dk_{1}}\right)\right).  
\end{align}
It returns the measurement outcome (indicating whether the state was projected onto $\Pi_{\mathrm{vrfy}}$ or not).  
            \end{description}

            \item $\qA$ sends $(\msg_0^*,\msg_1^*)\in \Xs^2$ to $\qC$. $\qC$ generates $a,b\la\bit$ and generates $\ct_0^*\la\Enc(\ek_0,\msg_a^*)$ and $\ct_1^*\la\Enc(\ek_1,\msg_{a\oplus b}^*)$. $\qC$ sends $\ct_0^*$ and $\ct_1^*$ to $\qA$.

            \item $\qA$ outputs a guess $b^\prime$ for $b$. $\qC$ outputs $1$ if $b=b^\prime$ and $0$ otherwise as the final output of the experiment.
        \end{enumerate}
        For any QPT $\qA$, it holds that
\begin{align}
\advb{\PKE,\qA}{coic}{kla}(\secp) \seteq 2\cdot\abs{\Pr[\expb{\PKE,\qA}{coic}{kla} (1^\secp) \ra 1] -  \frac{1}{2} }\leq \negl(\secp).
\end{align}
\end{definition}

\begin{definition}[Strong CoIC-KLA Security]\label{def:scKLA_PKE}\label{def_scKLA_PKE}
We say that a PKE scheme $\PKE$ with the message space $\Xs$ is $\epsilon$-strong CoIC-KLA secure, if it satisfies the following requirement, formalized from the experiment $\expc{\PKE,\qA}{s}{coic}{kla}(1^\secp,\epsilon)$ between an adversary $\qA$ and a challenger $\qC$:
        \begin{enumerate}
            \item  $\qC$ runs $(\ek_0,\dk_0)\gets\KG(1^\secp)$ and $(\ek_1,\dk_1)\gets\KG(1^\secp)$, and generates $\qdk:=\frac{1}{\sqrt{2}}(\ket{0}\ket{\dk_0}+\ket{1}\ket{\dk_1})$. $\qC$ sends $\ek_0$, $\ek_1$, and $\qdk$ to $\qA$. 

            \item $\qA$ sends $(\msg_0^*,\msg_1^*)\in \Xs^2$ and a quantum circuit $\qD =(\qstateq,\mat{U})$, where $\qD$ is a quantum program with classical inputs and one-bit outputs and $\mat{U}$ is a compact classical description of $\{\mat{U}_{\ct_0,\ct_1}\}_{\ct_0,\ct_1}$ to $\qC$. 

\item
Let $D$ be the following distribution.
\begin{description}
\item[$D$:]Generate $a,b\la\bit$ and $\ct_0\la\Enc(\ek_0,\msg_a)$ and $\ct_1\la\Enc(\ek_1,\msg_{a\oplus b})$. Output $(b,\ct_0,\ct_1)$.
\end{description}
We also let $\cP=(\mat{P}_{b,\ct_0,\ct_1},\mat{Q}_{b,\ct_0,\ct_1})_{b,\ct_0,\ct_1}$ be a collection of binary outcome projective measurements, where
\begin{align}
\mat{P}_{b,\ct_0,\ct_1}=\mat{U}_{\ct_0,\ct_1}^\dagger(\ket{b}\bra{b}\tensor \mat{I}) \mat{U}_{\ct_0,\ct_1}\textrm{~~~~and~~~~}\mat{Q}_{b,\ct_0,\ct_1}=\mat{I}-\mat{P}_{b,\ct_0,\ct_1}.
\end{align} 
Moreover, we let $\cM_D=(\mat{P}_D,\mat{Q}_D)$ be binary outcome POVMs, where
\begin{align}
\mat{P}_D=\sum_{r\in\calR}\frac{1}{\abs{\calR}}\mat{P}_{D(r)}\textrm{~~~~and~~~~}\mat{Q}_D =\mat{I}-\mat{P}_D.
\end{align}
Note that $\calR$ is the random coin space of $D$ and $\mat{P}_{D(r)}=\mat{P}_{b,\ct_0,\ct_1}$, where $(b,\ct_0,\ct_1)\la D(r)$.\footnote{The random coin $r$ for $D$ consists of random bits $a,b$ and encryption coins of two ciphertexts.} 
$\qC$ applies the measurement $\projimp(\cM_D)$ to $\qstateq$, and obtain a value $p$.
$\qC$ outputs $1$ if $p\geq\frac{1}{2}+\epsilon$ and $0$ otherwise.

        \end{enumerate}
        For any QPT $\qA$, it holds that
\begin{align}
\advc{\PKE,\qA}{s}{coic}{kla}(\secp) \seteq \Pr[\expc{\PKESKL,\qA}{s}{coic}{kla} (1^\secp,\epsilon) \ra 1] \leq \negl(\secp).
\end{align}
\end{definition}

\begin{theorem}
If $\PKE$ is $\epsilon$-strong CoIC-KLA secure for any inverse polynomial $\epsilon$, then $\PKE$ is CoIC-KLA secure.
\end{theorem}
\begin{proof}
Assume there exists $\qA$ that breaks CoIC-KLA security of $\PKE$. Without loss of generality, we assume that $\qA$ correctly guesses the bit $b$ with probability $\frac{1}{2}+\gamma$ for some inverse polynomial $\gamma$.
Then, consider the following experiment using $\qA$.
\begin{enumerate}
\item Execute $\expb{\PKE,\qA}{coic}{kla}(1^\secp)$ until the point $\qA$ outputs $(\msg_0^*,\msg_1^*)$.
\item Construct a quantum program with classical inputs and outputs $\qD=(\qstateq,\mat{U})$, where $\qstateq$ is the inner quantum state of $\qA$ and $\mat{U}$ is a compact description of $\{\mat{U}_{\ct_0,\ct_1}\}_{\ct_0,\ct_1}$ and $\mat{U}_{\ct_0,\ct_1}$ is a unitary that performs the rest of $\qA$'s computations on input $(\ct_0,\ct_1)$.
\item Obtain $p$ by applying $\projimp(\cM_D)$ to $\qstateq$, where the measurement $\cM_D$ and the distribution $D$ are defined in \cref{def:scKLA_PKE}.
\end{enumerate}
Then, from the definition of $\projimp$ and the fact that $\qA$'s advantage is $\frac{1}{2}+\gamma$, we have $E[p]=\frac{1}{2}+\gamma$.
By the averaging argument, we obtain $\Pr[p\ge\frac{1}{2}+\frac{\gamma}{2}]\ge \frac{\gamma}{2}$.
Consider the following adversary $\qB$ that attacks $\frac{\gamma}{2}$-strong CoIC-KLA security of $\PKE$.
\begin{enumerate}
\item Given, $\ek_0$, $\ek_1$, and $\qdk$, $\qB$ executes $\expb{\PKE,\qA}{coic}{kla}(1^\secp)$ until the point $\qA$ outputs $(\msg_0^*,\msg_1^*)$. When $\qA$ makes a query to $\cO$, $\qB$ returns a random bit.
\item $\qB$ constructs a quantum program with classical inputs and outputs $\qD=(\qstateq,\mat{U})$, where $\qstateq$ is the inner quantum state of $\qA$, $\mat{U}$ is a compact description of $\{\mat{U}_{\ct_0,\ct_1}\}_{\ct_0,\ct_1}$, and $\mat{U}_{\ct_0,\ct_1}$ is a unitary that performs the rest of $\qA$'s computations on input $(\ct_0,\ct_1)$.
$\qB$ outputs $(\msg_0^*,\msg_1^*)$ and $\qD$.
\end{enumerate}

$\qB$ correctly answers to $\qA$'s query to $\cO$ and correctly simulates $\expb{\PKE,\qA}{coic}{kla}$ for $\qA$ with probability $\frac{1}{2}$.\footnote{$\qB$ does not apply the verification procedure to the queried state differently from $\expb{\PKE,\qA}{coic}{kla}$. This is not a problem since from the view of $\qA$, the experiment simulated by $\qB$ is the same as the experiment where the verification process is applied to the queried state, but the result is ignored and a random bit is returned.}
Moreover, from the above discussion, under the  condition that $\qB$ correctly answers to $\qA$'s query to $\cO$, $\qB$ wins with probability $\frac{\gamma}{2}$.
Overall, $\advc{\PKE,\qB}{s}{coic}{kla}(\secp)\ge\frac{\gamma}{4}$, which contradicts $\frac{\gamma}{2}$-strong CoIC-KLA security of $\PKE$.
This completes the proof.
\end{proof}


\newcommand{\msglen}{\ell}
\newcommand{\tlp}{\tilde{p}}
\newcommand{\BadDec}{\mathsf{BadDec}}

\subsection{Strong $\cKLA$ Secure PKE from CPFE}\label{sec:scKLA_PKE}

We construct a strong $\cKLA$ secure PKE $\PKE=(\Gen,\Enc,\Dec)$ using a CPFE scheme $\CPFE=(\CPFE.\Setup,\allowbreak\CPFE.\KG,\CPFE.\Enc,\CPFE.\Dec)$ as a building block.

\begin{description}

 \item[$\Gen(1^\secp)$:] $ $
 \begin{itemize}
 \item Generate $(\MPK,\MSK)\lrun \CPFE.\Setup(1^\secp)$.
 \item Generate $x\la\bit^\secp$ and $\sk_{x}\la\CPFE.\KG(\MSK,x)$.
 \item Output $\ek\seteq\MPK$ and $\dk:=\sk_x$.
 \end{itemize}
 \item[$\Enc(\ek,\msg)$:] $ $
 \begin{itemize}
\item Parse $\ek= \MPK$.
\item Let $C[\msg]$ be a constant circuit that outputs $\msg$ on any input. $C$ is padded so that it has the same size as the circuit $C^*$ appeared in the security proof.
\item Output $\ct\la\CPFE.\Enc(\MPK,C[\msg])$.
 \end{itemize}
  \item[$\Dec(\dk,\ct)$:] $ $
 \begin{itemize}
\item Parse $\dk= \sk_x$.
\item Output $\msg^\prime\la\CPFE.\Dec(\sk_x,\ct)$.
 \end{itemize}
\end{description}

The decryption correctness of $\PKE$ follows from that of $\CPFE$.
We also have the following theorems.

\begin{theorem}\label{thm:scKLA_PKE}
If $\CPFE$ is a $1$-bounded secure CPFE scheme, then $\PKE$ is a $\epsilon$-strong $\cKLA$ secure PKE scheme for any inverse polynomial $\epsilon$.
\end{theorem}

\begin{proof}
We show that if there exists a QPT adversary $\qA$ that breaks $\epsilon$-strong $\cKLA$ security for some inverse polynomial $\epsilon$, then we can construct a QPT adversary $\qB$ that contradicts the following lemma.

\begin{lemma}\label{lem:BZ_CPFE}
Consider the following experiment $\expa{\CPFE,\qB}{BZ}(1^\secp)$ between an adversary $\qB$ and a challenger $\qC$.
\begin{enumerate}
\item $\qC$ generates $(\MPK_0,\MSK_0)\la\CPFE.\Setup(1^\secp)$, $(\MPK_1,\MSK_1)\la\CPFE.\Setup(1^\secp)$, $x_0,x_1\la\bit^\secp$, $\sk_{x_0}\la\CPFE.\KG(\MSK_0,x_0)$, and $\sk_{x_1}\la\CPFE.\KG(\MSK_1,x_1)$. $\qC$ gives $\MPK_0$, $\MPK_1$, and $\frac{1}{\sqrt{2}}(\ket{0}\ket{\sk_{x_0}}+\ket{1}\ket{\sk_{x_1}})$ to $\qB$.
\item $\qB$ outputs $x_0^\prime$ and $x_1^\prime$. $\qC$ outputs $1$ if $x_0^\prime=x_0$ and $x_1^\prime=x_1$ and $0$ otherwise.
\end{enumerate}
Then, for any QPT adversary $\qB$, we have $\adva{\CPFE,\qB}{BZ}(1^\secp)=\Pr[\expa{\CPFE,\qB}{BZ}(1^\secp)=1]=\negl(\secp)$.
\end{lemma}
\begin{proof}
This lemma directly follows from \cref{lem:BZ}.
\end{proof}

Let $\epsilon$ be some inverse polynomial. Assume there exits a QPT $\qA$ such that $\advc{\PKE,\qA}{s}{coic}{kla}(\secp,\epsilon)=\gamma$ for some inverse polynomial $\gamma$.
We construct the following adversary $\qB$.
\begin{enumerate}
\item Given $\MPK_0$, $\MPK_1$, and $\qdk$, $\qB$ sets $\ek_0:=\MPK_0$ and $\ek_1:=\MPK_1$. $\qB$ sends $\ek_0$, $\ek_1$, and $\qdk$ to $\qA$.
\item When $\qA$ outputs $(\msg_0^*,\msg_1^*)$ and $\qD=(\qstateq,\mat{U})$, $\qB$ outputs $(x_0^\prime,x_1^\prime)\la\qExtract(\MPK_0,\MPK_1,\msg_0^*,\msg_1^*,\qD,\epsilon)$, where $\qExtract$ is described below.
\end{enumerate}
\begin{description}
\item[$\qExtract(\MPK_0,\MPK_1,\msg_0^*,\msg_1^*,\qD,\epsilon)$:]$ $
\begin{itemize}
\item Let $\epsilon'=\epsilon/8\secp$ and $\delta'=2^{-\lambda}$.
\item Parse $(\qstateq,\mat{U})\gets\qD$.
\item Let $\cP$ be defined in the same way as that in \cref{def_scKLA_PKE} and $D_{0,i}$ and $D_{1,i}$ be the following distributions for every $i\in[\secp]$.
\begin{description}
\item[$D_{0,i}$:] Generate $a,b\la\bit$. Generate $\ct_0\la\CPFE.\Enc(\MPK_0,C^*[a,b,\msg_0,\msg_1,i])$, where $C^*[a,b,\msg_0,\msg_1,i]$ is a circuit that takes $x$ as input and outputs $\msg_{a\oplus b\oplus x[i]}$. Generate $\ct_1\la\CPFE.\Enc(\MPK_1,C[\msg_a])$. Output $(b,\ct_0,\ct_1)$.
\item[$D_{1,i}$:] Generate $a,b\la\bit$. Generate $\ct_1\la\CPFE.\Enc(\MPK_0,C^*[a,b,\msg_0,\msg_1,i])$, where $C^*[a,b,\msg_0,\msg_1,i]$ is a circuit that takes $x$ as input and outputs $\msg_{a\oplus b\oplus x[i]}$. Generate $\ct_0\la\CPFE.\Enc(\MPK_1,C[\msg_a])$. Output $(b,\ct_0,\ct_1)$.
\end{description}
\item Let $\D$ be the distribution defined in the same way as that in \cref{def:scKLA_PKE}. Compute $\tlp_{0} \gets \API_{\cP,D}^{\epsilon' ,\delta'}(\qstateq)$. If $\tlp_0<\frac{1}{2}+\epsilon-4\epsilon'$, return $\bot$. Otherwise, let $\qstateq_{0,0}$ be the post-measurement state, go to the next step.
\item For all $i \in [\secp]$, do the following.
\begin{enumerate}
\item Compute $\tlp_{0,i} \gets \API_{\cP,D_{0,i}}^{\epsilon' ,\delta'}(\qstateq_{0,i-1})$. Let $\qstateq_{0,i}$ be the post-measurement state.
\item If $\tlp_{0,i}>\frac{1}{2}+\epsilon-4(i+1)\epsilon'$, set $x_0^\prime[i]=0$. If $\tlp_{0,i}<\frac{1}{2}-\epsilon+4(i+1)\epsilon'$, set $x_0^\prime[i]=1$. Otherwise, exit the loop and output $\bot$.
\end{enumerate}
\item Let $\qstateq_{1,0}$ be $\qstateq_{0,\secp}$. For all $i \in [\secp]$, do the following.
\begin{enumerate}
\item Compute $\tlp_{1,i} \gets \API_{\cP,D_{1,i}}^{\epsilon' ,\delta'}(\qstateq_{1,i-1})$. Let $\qstateq_{1,i}$ be the post-measurement state.
\item If $\tlp_{1,i}>\frac{1}{2}+\epsilon-4(\secp+i+1)\epsilon'$, set $x_1^\prime[i]=0$. If $\tlp_{1,i}<\frac{1}{2}-\epsilon+4(\secp+i+1)\epsilon'$, set $x_1^\prime[i]=1$. Otherwise, exit the loop and output $\bot$.
\end{enumerate}
\item Output $x_0'=x_0^\prime[1] \concat \cdots \concat x_0^\prime[\secp]$ and $x_1'=x_1^\prime[1] \concat \cdots \concat x_1^\prime[\secp]$.
\end{itemize}
\end{description}

We will estimate $\adva{\CPFE,\qB}{BZ}(1^\secp)$.
We define the events $\BadDec$, and $\BadExt_{0,i}$ and $\BadExt_{1,i}$ for every $i\in[\secp]$.
\begin{description}
\item[$\BadDec$:]When $\qB$ runs $\qExtract(\MPK_0,\MPK_1,\msg_0^*,\msg_1^*,\qD,\epsilon)$, $\tlp_0<\frac{1}{2}+\epsilon-4\epsilon'$ holds.
\item[$\BadExt_{0,i}$:]When $\qB$ runs $\qExtract(\MPK_0,\MPK_1,\msg_0^*,\msg_1^*,\qD,\epsilon)$, the following conditions hold.
\begin{itemize}
\item $\tlp_0\geq\frac{1}{2}+\epsilon-4\epsilon'$ holds.
\item $x_0^\prime[j]=x_0[j]$ holds for every $j\in[i-1]$.
\item $x_0^\prime[i]\neq x_0[i]$ holds.
\end{itemize}
\item[$\BadExt_{1,i}$:]When $\qB$ runs $\qExtract(\MPK_0,\MPK_1,\msg_0^*,\msg_1^*,\qD,\epsilon)$, the following conditions hold.
\begin{itemize}
\item $\tlp_0\geq\frac{1}{2}+\epsilon-4\epsilon'$ holds.
\item $x_0^\prime[j]=x_0[j]$ holds for every $j\in[\secp]$.
\item $x_1^\prime[j]=x_1[j]$ holds for every $j\in[i-1]$.
\item $x_1^\prime[i]\neq x_1[i]$ holds.
\end{itemize}
\end{description}

From the assumption that $\advc{\PKE,\qA}{s}{coic}{kla}(\secp)=\gamma$, for $\tlp_0$ computed in $\qExtract$, $\tlp_0\ge\frac{1}{2}+\epsilon-\epsilon^\prime$ holds with probability $\gamma-\negl(\secp)$ due to the first item of \cref{thm:api_property}.
This means that $\Pr[\BadDec]\le1-\gamma+\negl(\secp)$. 
Then, we have
\begin{align}
\adva{\CPFE,\qB}{BZ}(1^\secp)&\ge1-\left(\Pr[\BadDec]+\sum_{i\in[\secp]}\Pr[\BadExt_{0,i}]+\sum_{i\in[\secp]}\Pr[\BadExt_{1,i}]\right)\\
&\ge\gamma - \negl(\secp) -\left(\sum_{i\in[\secp]}\Pr[\BadExt_{0,i}]+\sum_{i\in[\secp]}\Pr[\BadExt_{1,i}]\right).
\end{align}

\paragraph{Estimation of $\Pr[\BadExt_{0,i}]$ for every $i\in[\secp]$.}
We first estimate $\Pr[\BadExt_{0,1}]$.
We first consider the case of $x_0[1]=0$.
From the first item of the event, we have $\tlp_{0}>\frac{1}{2}+\epsilon-4\epsilon'$.
Let $\tlp'_{0} \gets \API_{\cP,D}^{\epsilon' ,\delta'}(\qstateq_{0,0})$.
From the almost-projective property of $\API$, we have
\begin{align}
\Pr[\tlp'_{0}>\frac{1}{2}+\epsilon-4\epsilon'-\epsilon']\geq1-\delta'.
\end{align}
\begin{lemma}\label{lem:ind_cpfe_0}
When $x_0[1]=0$, $D_{0,1}$ is computationally indistinguishable from $D$.
\end{lemma}
\begin{proof}
The difference between $D_{0,1}$ and $D$ is that $\ct_0$ is generated as  $\ct_0\la\CPFE.\Enc(\MPK_0,C^*[a,b,\msg_0,\msg_1,1])$ in $D_{0,1}$ and it is generated as  $\ct_0\la\CPFE.\Enc(\MPK_0,C[\msg_{a\oplus b}])$ in $D$.
From the condition that $x_0[1]=0$, we have $C^*[a,b,\msg_0,\msg_1,1](x_0)=C[\msg_{a\oplus b}](x_0)=\msg_{a\oplus b}$.
Thus, from the $1$-bounded security of $\CPFE$, $D_{0,1}$ and $D$ are computationally indistinguishable when $x_0[1]=0$.
\end{proof}
Thus, from \cref{cor:cind_sample_api} and \cref{lem:ind_cpfe_0}, we have
\begin{align}
1-\delta'\leq
\Pr[\tlp'_{0}>\frac{1}{2}+\epsilon-5\epsilon']
\leq\Pr[\tlp_{0,1}>\frac{1}{2}+\epsilon-8\epsilon']+\negl(\secp).
\end{align}
This means that $\Pr[\BadExt_{0,1}]=\negl(\secp)$ when $x_0[1]=0$.
We next consider the case of $x_0[1]=1$.
We define the following distribution $\Drev$.
\begin{description}
\item[$\Drev$:] Generate $(b,\ct_0,\ct_1)\gets D$. Output $(1\oplus b,\ct_0,\ct_1)$.
\end{description}
That is, the first bit of the output is flipped from $D$.
Then, for any random coin $r$, we have $(\mat{P}_{\Drev(r)},\mat{Q}_{\Drev(r)})=(\mat{Q}_{D(r)},\mat{P}_{D(r)})$.
This is because we have $\mat{Q}_{b,\ct_0,\ct_1}=\mat{I}-\mat{P}_{b,\ct_0,\ct_1}=\mat{P}_{1\oplus b,\ct_0,\ct_1}$ for any tuple $(b,\ct_0,\ct_1)$.
Therefore, $ \API_{\cP,\Drev}^{\epsilon' ,\delta'}$ is exactly the same process as  $\API_{\cPrev,D}^{\epsilon' ,\delta'}$, where $\cPrev=(\mat{Q}_{b,\ct_0,\ct_1},\mat{P}_{b,\ct_0,\ct_1})_{b,\ct_0,\ct_1}$.
Let $\tlp'_{0} \gets \API_{\cP,\Drev}^{\epsilon' ,\delta'}(\qstateq_{0,0})$.
From, the reverse-almost-projective property of $\API$, we have
\begin{align}
\Pr[\tlp'_{0}<\frac{1}{2}-\epsilon+4\epsilon'+\epsilon']\geq1-\delta'.
\end{align}
\begin{lemma}\label{lem:ind_cpfe_1}
When $x_0[1]=1$, $D_{0,1}$ is computationally indistinguishable from $\Drev$.
\end{lemma}
\begin{proof}
We see that $\Drev$ is identical to the following distribution.
\begin{itemize}
\item Generate $a,b\la\bit$ and $\ct_0\la\Enc(\ek_0,\msg_a)$ and $\ct_1\la \Enc(\ek,\msg_{a\oplus 1\oplus b})$. Output $(b,\ct_0,\ct_1)$.
\end{itemize}
Then, the difference between $D_{0,1}$ and $\Drev$ is that $\ct_0$ is generated as  $\ct_0\la\CPFE.\Enc(\MPK_0,C^*[a,b,\msg_0,\msg_1,1])$ in $D_{0,1}$ and it is generated as  $\ct_0\la\CPFE.\Enc(\MPK_0,C[\msg_{a\oplus 1\oplus b}])$ in $\Drev$.
From the condition that $x_0[1]=1$, we have $C^*[a,b,\msg_0,\msg_1,1](x_0)=C[\msg_{a\oplus 1\oplus b}](x_0)=\msg_{a\oplus 1\oplus b}$.
Thus, from the $1$-bounded security of $\CPFE$, $D_{0,1}$ and $\Drev$ are computationally indistinguishable when $x_0[1]=1$.
\end{proof}
Thus, from \cref{cor:cind_sample_api} and \cref{lem:ind_cpfe_1}, we have
\begin{align}
1-\delta'\leq
\Pr[\tlp'_{0}<\frac{1}{2}-\epsilon+5\epsilon']
\leq\Pr[\tlp_{0,1}<\frac{1}{2}-\epsilon+8\epsilon']+\negl(\secp).
\end{align}
This means that $\Pr[\BadExt_{0,1}]=\negl(\secp)$ when $x_0[1]=1$.

Overall, $\Pr[\BadExt_{0,1}]=\negl(\secp)$ regardless of the value of $x_0$.
We can similarly show that $\Pr[\BadExt_{0,i}]=\negl(\secp)$ for $i\in\{2,\cdots,\secp\}$ using the fact that $D_{0,i}$ is computationally indistinguishable from $D$ if $x_0[i]=0$ and it is computationally indistinguishable from $\Drev$ if $x_0[i]=1$.
We omit the details.

\paragraph{Estimation of $\Pr[\BadExt_{1,i}]$ for every $i\in[\secp]$.}
We estimate $\Pr[\BadExt_{1,1}]$.
We first consider the case of $x_0[\secp]=0$ and $x_1[1]=0$.
From the second item of the event, we have $\tlp_{0,\secp}>\frac{1}{2}+\epsilon-4(\secp+1)\epsilon'$.
Let $\tlp'_{0,\secp} \gets \API_{\cP,D_{0,\secp}}^{\epsilon' ,\delta'}(\qstateq_{0,\secp})$.
From, the almost-projective property of $\API$, we have
\begin{align}
\Pr[\tlp'_{0,\secp}>\frac{1}{2}+\epsilon-4(\secp+1)\epsilon'-\epsilon']\geq1-\delta'.
\end{align}
\begin{lemma}\label{lem:ind_cpfe_00}
When $x_0[\secp]=x_1[1]=0$, $D_{0,\secp}$ and $D_{1,1}$ are computationally indistinguishable.
\end{lemma}
\begin{proof}
We can show that $D_{0,\secp}$ is computationally indistinguishable from $D$ when $x_0[\secp]=0$ similarly to \cref{lem:ind_cpfe_0}.
We see that $D$ is identical to the following distribution.
\begin{itemize}
\item Generate $a,b\la\bit$ and $\ct_0\la\Enc(\ek_0,\msg_{a\oplus b})$ and $\ct_1\la\Enc(\ek,\msg_{a})$. Output $(b,\ct_0,\ct_1)$.
\end{itemize}
Then, the difference between $D_{1,1}$ and $D$ is that $\ct_1$ is generated as  $\ct_1\la\CPFE.\Enc(\MPK_1,C^*[a,b,\msg_0,\msg_1,1])$ in $D_{1,1}$ and it is generated as  $\ct_1\la\CPFE.\Enc(\MPK_1,C[\msg_{a\oplus b}])$ in $D$.
From the condition that $x_1[1]=0$, we have $C^*[a,b,\msg_0,\msg_1,1](x_1)=C[\msg_{a\oplus b}](x_1)=\msg_{a\oplus b}$.
Thus, from the $1$-bounded security of $\CPFE$, $D_{1,1}$ and $D$ are computationally indistinguishable when $x_0[1]=0$.
This means that $D_{0,\secp}$ and $D_{1,1}$ are computationally indistinguishable when $x_0[\secp]=x_1[1]=0$.
\end{proof}
Thus, from \cref{cor:cind_sample_api} and \cref{lem:ind_cpfe_00}, we have
\begin{align}
1-\delta'\leq
\Pr[\tlp'_{0,\secp}>\frac{1}{2}+\epsilon-(4\secp+5)\epsilon']
\leq\Pr[\tlp_{1,1}>\frac{1}{2}+\epsilon-4(\secp+2)\epsilon']+\negl(\secp).
\end{align}
This means that $\Pr[\BadExt_{1,1}]=\negl(\secp)$ when $x_0[\secp]=0$ and $x_1[1]=0$.
We next consider the case of $x_0[\secp]=0$ and $x_1[1]=1$.
We define the following distribution $\Drev_{0,\secp}$.
\begin{description}
\item[$\Drev_{0,\secp}$:] Generate $(b,\ct_0,\ct_1)\gets D_{0,\secp}$. Output $(1\oplus b,\ct_0,\ct_1)$.
\end{description}
That is, the first bit of the output is flipped from $D_{0,\secp}$.
Then, for any random coin $r$, we have $(\mat{P}_{\Drev_{0,\secp}(r)},\mat{Q}_{\Drev_{0,\secp}(r)})=(\mat{Q}_{D_{0,\secp}(r)},\mat{P}_{D_{0,\secp}(r)})$.
(Again, this is because we have $\mat{Q}_{b,\ct_0,\ct_1}=\mat{I}-\mat{P}_{b,\ct_0,\ct_1}=\mat{P}_{1\oplus b,\ct_0,\ct_1}$ for any tuple $(b,\ct_0,\ct_1)$.)
Therefore, $ \API_{\cP,\Drev_{0,\secp}}^{\epsilon' ,\delta'}$ is exactly the same process as  $\API_{\cPrev,D_{0,\secp}}^{\epsilon' ,\delta'}$, where $\cPrev=(\mat{Q}_{b,\ct_0,\ct_1},\mat{P}_{b,\ct_0,\ct_1})_{b,\ct_0,\ct_1}$.
Let $\tlp'_{0,\secp} \gets \API_{\cP,\Drev_{0,\secp}}^{\epsilon' ,\delta'}(\qstateq_{0,\secp})$.
From, the reverse-almost-projective property of $\API$, we have
\begin{align}
\Pr[\tlp'_{0,\secp}<\frac{1}{2}-\epsilon+4(\secp+1)\epsilon'+\epsilon']\geq1-\delta'.
\end{align}
\begin{lemma}\label{lem:ind_cpfe_01}
When $x_0[\secp]=0$ and $x_1[1]=1$, $\Drev_{0,\secp}$ and $D_{1,1}$ are computationally indistinguishable.
\end{lemma}
\begin{proof}
We can show that both $\Drev_{0,\secp}$ and $D_{1,1}$ are computationally indistinguishable from $\Drev$ when $x_0[\secp]=0$ and $x_1[1]=1$.
The proof is similarly to those for \cref{lem:ind_cpfe_0,lem:ind_cpfe_1,lem:ind_cpfe_00}, thus we omit the details.
\end{proof}

Thus, from \cref{cor:cind_sample_api} and \cref{lem:ind_cpfe_01}, we have
\begin{align}
1-\delta'\leq
\Pr[\tlp'_{0,\secp}<\frac{1}{2}-\epsilon+(4\secp+5)\epsilon']
\leq\Pr[\tlp_{1,1}<\frac{1}{2}-\epsilon+4(\secp+2)\epsilon']+\negl(\secp).
\end{align}
This means that $\Pr[\BadExt_{1,1}]=\negl(\secp)$ when $x_0[\secp]=0$ and $x_1[1]=1$.

Similarly, we can show that $\Pr[\BadExt_{1,1}]=\negl(\secp)$ holds when $(x_0[\secp],x_1[1])=(1,0)$ and $(x_0[\secp],x_1[1])=(1,1)$.
Moreover, we can show that $\Pr[\BadExt_{1,i}]=\negl(\secp)$ holds for $i\in\{2,\cdots,\secp\}$.

\medskip

From the above discussion, we have $\adva{\CPFE,\qB}{BZ}(1^\secp)\ge\gamma-\negl(\secp)$ for some inverse polynomial $\gamma$, which contradicts \cref{lem:BZ_CPFE}.
This completes the proof of \cref{thm:scKLA_PKE}.
\end{proof}


\section{Construction of PKE with Secure Key Leasing}\label{sec:PKE-SKL}
In this section, we prove the following theorem: 
\begin{theorem}\label{thm:PKE-SKL_from_PKE}
If there is an IND-CPA secure PKE scheme, then there is an IND-KLA secure PKE-SKL scheme.
\end{theorem}

By~\Cref{thm:one-query-ow_to_full}, it suffices to construct 1-query OW-KLA secure PKE-SKL scheme. In the rest of this section, we construct such a scheme.  To build our scheme, we rely on a PKE scheme satisfying $\cKLA$ security, which is constructed from any IND-CPA secure PKE scheme in Section \ref{sec:PKE-CoIC}.


Let $\cPKE=(\cPKE.\KG,\cPKE.\Enc,\cPKE.\Dec)$ be a PKE scheme satisfying $\cKLA$ security with message space $\bit^\ell$ where $\ell=\omega(\log \secp)$. We note that $\cKLA$ security implies OW-CPA security when $\ell=\omega(\log \secp)$. (See \cref{sec:OW_from_CoIC} for the proof.) 
Then, we construct a PKE-SKL scheme $(\PKESKL.\qKG, \PKESKL.\Enc, \PKESKL.\qDec,\PKESKL.\qVrfy)$ with message space $\bit^{\secp\ell}$ as follows.

\begin{description}
\item[$\PKESKL.\qKG(1^\secp)$:] $ $
\begin{itemize}
    \item Generate $(\cPKE.\ek_{i,b},\cPKE.\dk_{i,b})\gets \cPKE.\KG(1^\secp)$ for $i\in [\secp]$ and $b\in \bit$. 
    \item Output an encryption key 
    \[\ek\seteq \{\cPKE.\ek_{i,b}\}_{i\in[\secp],b\in\bit},\] 
    a decryption key 
    \[\qdk\seteq \bigotimes_{i\in[\secp]}\frac{1}{\sqrt{2}}\left(\ket{0}\ket{\cPKE.\dk_{i,0}}+\ket{1}\ket{\cPKE.\dk_{i,1}}\right),\]
    and a verification key 
    \[\vk\seteq \{\cPKE.\dk_{i,b}\}_{i\in[\secp],b\in\bit}.\]
    For convenience, we write $\qregidx{DK}{i}$ to mean the registers of $\qdk$ that contains $\frac{1}{\sqrt{2}}\left(\ket{0}\ket{\cPKE.\dk_{i,0}}+\ket{1}\ket{\cPKE.\dk_{i,1}}\right)$ for $i\in[\secp]$.
\end{itemize}

\item[$\PKESKL.\Enc(\ek,\msg)$:] $ $
\begin{itemize}
    \item Parse $\ek=\{\cPKE.\ek_{i,b}\}_{i\in[\secp],b\in\bit}$ and  $\msg = \msg_1\concat \ldots \concat\msg_\secp$   where $\msg_i\in \bit^\ell$ for each $i\in [\secp]$. 
    \item Generate
$\cPKE.\ct_{i,b}\gets \cPKE.\Enc(\cPKE.\ek_{i,b},\msg_{i})$ for $i\in [\secp]$ and $b\in \bit$.
\item Output $\ct\seteq \{\cPKE.\ct_{i,b}\}_{i\in [\secp],b\in \bit}$. 
\end{itemize}
 
\item[$\PKESKL.\qDec(\qdk,\ct)$:] $ $
\begin{itemize}
   \item Parse  $\qdk= \bigotimes_{i\in[\secp]}\qdk_i$ and $\ct= \{\cPKE.\ct_{i,b}\}_{i\in [\secp],b\in \bit}$.
    \item Let $U_{\mathrm{dec}}$ be a unitary such that for all $\cPKE.\dk'$, $\cPKE.\ct'_0$, and $\cPKE.\ct'_1$:
\begin{align}
  &\ket{b}\ket{\cPKE.\dk'}\ket{\cPKE.\ct'_0,\cPKE.\ct'_1}\ket{0}\\  \xrightarrow{U_{\mathrm{dec}}}
   &\ket{b}\ket{\cPKE.\dk'}\ket{\cPKE.\ct'_0,\cPKE.\ct'_1}\ket{\cPKE.\Dec(\cPKE.\dk',\cPKE.\ct'_b)} 
\end{align}

Note that such a unitary can be computed in quantum polynomial-time since we assume that $\cPKE.\Dec$ is a  deterministic classical polynomial-time algorithm. 
\item For all $i\in[\secp]$, 
generate
\[
U_{\mathrm{dec}}\left(\qdk_i\otimes \ket{\cPKE.\ct_{i,0},\cPKE.\ct_{i,1}}\bra{\cPKE.\ct_{i,0},\cPKE.\ct_{i,1}}\otimes \ket{0}\bra{0}   \right)U_{\mathrm{dec}}^\dagger,
\]
measure the rightmost register, and let $\msg'_i$ be the measurement outcome. 
\item Output $\msg'\seteq \msg'_1\concat\ldots\concat \msg'_\secp$. 
\end{itemize}

\item[$\PKESKL.\qVrfy(\vk,\widetilde{\qdk})$:] $ $
\begin{itemize}
    \item Parse $\vk=\{\cPKE.\dk_{i,b}\}_{i\in[\secp],b\in\bit}$.
    \item Apply a binary-outcome measurement $( \mat{I}-\Pi_{\mathrm{vrfy}}^{\vk},\Pi_{\mathrm{vrfy}}^{\vk})$ on $\widetilde{\qdk}$   
where $\Pi_{\mathrm{vrfy}}^{\vk}$ is the projection onto the right decryption key, i.e., 
\begin{align}
 \Pi_{\mathrm{vrfy}}^{\vk}\seteq 
 \bigotimes_{i\in[\secp]}\left(\frac{1}{\sqrt{2}}\left(\ket{0}\ket{\cPKE.\dk_{i,0}}+\ket{1}\ket{\cPKE.\dk_{i,1}}\right)\right)  
 \left(\frac{1}{\sqrt{2}}\left(\bra{0}\bra{\cPKE.\dk_{i,0}}+\bra{1}\bra{\cPKE.\dk_{i,1}}\right)\right).  
\end{align}
If the measurement outcome is $1$ (indicating that the state was projected onto $\Pi_{\mathrm{vrfy}}^{\vk}$), output $\top$ and otherwise output $\bot$.  
\end{itemize}
\end{description}

The correctness of $\SKL$ easily follows from that of 
$\cPKE$.
Below, we show that $\SKL$ is 1-query  OW-KLA secure. 

\begin{theorem}\label{thm:PKE_OW-KLA}
If $\cPKE$ is $\cKLA$ secure, then $\SKL$ is 1-query OW-KLA secure. 
\end{theorem}
\begin{proof}  
Let $\qA$ be a QPT adversary against 1-query OW-KLA security of $\SKL$. 
By \cref{rem:no-post-challenge-query}, 
we assume that $\qA$ makes the verification query before receiving the challenge ciphertext without loss of generality. 
We consider the following sequence of hybrids. 
\begin{description}
\item[$\hybi{0}$:] This is the same as $\expb{\PKESKL,\qA}{ow}{kla}(1^\secp)$. More specifically, it works as follows.
\begin{enumerate}
            \item  The challenger generates $(\cPKE.\ek_{i,b},\cPKE.\dk_{i,b})\gets \cPKE.\KG(1^\secp)$ for $i\in [\secp]$ and $b\in \bit$,
            sets 
         $\ek\seteq \{\cPKE.\ek_{i,b}\}_{i\in[\secp],b\in\bit}$ and
         $\qdk\seteq \bigotimes_{i\in[\secp]}\frac{1}{\sqrt{2}}\left(\ket{0}\ket{\cPKE.\dk_{i,0}}+\ket{1}\ket{\cPKE.\dk_{i,1}}\right)$,   
         and sends $\ek$ and $\qdk$ to $\qA$.  
            \item \label{step:send_qdk}
            $\qA$ queries  $\widetilde{\qdk}$
            to the verification oracle. 
            The challenger applies a binary-outcome measurement $(\mat{I}-\Pi_{\mathrm{vrfy}}^{\vk},\Pi_{\mathrm{vrfy}}^{\vk})$ on $\widetilde{\qdk}$where $\Pi_{\mathrm{vrfy}}^{\vk}$ is the projection defined in the description of $\PKESKL.\qVrfy$. If the measurement outcome is $0$ (indicating that the state was projected onto $\mat{I}-\Pi_{\mathrm{vrfy}}^{\vk}$), the challenger outputs $0$ as the final outcome of this experiment.\footnote{In the description of the OW-KLA experiment in \cref{def:OW-CPA_PKESKL}, the oracle returns $\bot$ even if the decryption key does not pass the verification. However, in the 1-query setting, if the first (and only) query is rejected, the experiment finally outputs $0$. Thus, we terminate the experiment at this point when the query is rejected.} 
            Otherwise, the challenger returns $\top$ to $\qA$ as the response from the oracle.
          \item  The challenger chooses $\msg^*_i\gets \bit^{\ell}$ for $i\in [\secp]$,  
           generates
$\cPKE.\ct^*_{i,b}\gets \cPKE.\Enc(\cPKE.\ek_{i,b},\msg^*_{i})$ for $i\in [\secp]$ and $b\in \bit$, and sends $\ct^*\seteq \{\cPKE.\ct^*_{i,b}\}_{i\in [\secp],b\in \bit}$  to $\qA$.\footnote{Since $\qA$ makes only one verification query, we can assume that $\qA$ requests the challenge ciphertext immediately after finishing the first verification query without loss of generality.}
            \item 
            \label{step_finalize}
            $\qA$ outputs $\msg'=\msg'_1\concat \ldots\concat \msg'_\secp$. The challenger outputs $1$ if $\msg'_i=\msg^*_i$ for all $i\in [\secp]$ and otherwise $0$ as the final outcome of this experiment.  
        \end{enumerate}
Note that we have $\Pr[\hybi{0}=1]=\advb{\PKESKL,\qA}{ow}{kla}(1^\secp)$. Our goal is to prove $\Pr[\hybi{0}=1]=\negl(\secp)$.
\item[$\hybi{1}$:] This is identical to  $\hybi{0}$ except for the following modifications:
\begin{itemize}
    \item The challenger chooses $\msg^*_{i,b}\gets \bit^\ell$ for $i\in[\secp]$ and $b\in \bit$ (instead of choosing $\msg^*_{i}\gets \bit^\ell$ for $i\in [\secp]$) and $a_i\gets \bit$ for $i\in[\secp]$. 
    \item 
    $\cPKE.\ct^*_{i,b}$ is generated as 
$\cPKE.\ct^*_{i,b}\gets \cPKE.\Enc(\cPKE.\ek_{i,b},\msg^*_{i,a_i})$ for $i\in [\secp]$ and $b\in \bit$. We emphasize that $\msg^*_{i,\color{red}\underline{a_i}}$ is encrypted for both cases of $b=0$ and $b=1$ and $\msg^*_{i,a_i\oplus 1}$ is not used in this step. 
\item In Step~\ref{step_finalize}, the challenger outputs $1$ if $\msg'_i\in \{\msg^*_{i,0},\msg^*_{i,1}\}$ for all $i\in [\secp]$.
\end{itemize}
By considering $\msg^*_{i,a_i}$ in $\hybi{1}$ as $\msg^*_i$ in $\hybi{0}$, these hybrids are identical from the view of $\qA$ except that the winning condition (i.e., the condition that the challenger returns $1$) is just relaxed in $\hybi{1}$. Therefore, we trivially have $\Pr[\hybi{0}=1]\le \Pr[\hybi{1}=1]$. 
\item[$\hybi{2}$:] This is identical to  $\hybi{1}$ except that $\cPKE.\ct^*_{i,b}$ is generated as 
$\cPKE.\ct^*_{i,b}\gets \cPKE.\Enc(\cPKE.\ek_{i,b},\msg^*_{i,a_i \oplus b})$ for $i\in [\secp]$ and $b\in \bit$. We remark that the way of generating $\cPKE.\ct^*_{i,1}$ is changed but that of $\cPKE.\ct^*_{i,0}$ is unchanged (because $a_i\oplus 0 = a_i$).  

By the CoIC-KLA security of $\cPKE$ and a standard hybrid argument, we have 
$\abs{\Pr[\hybi{1}=1] - \Pr[\hybi{2}=1]} = \negl(\secp)$. See \cref{lem:PKE_SKL_1-to-2} for the detail.  
\item[$\hybi{3}$:] This is identical to  $\hybi{2}$ except that 
the challenger quits choosing $a_i\gets \bit$ for $i\in [\secp]$ and 
$\cPKE.\ct^*_{i,b}$ is generated as 
$\cPKE.\ct^*_{i,b}\gets \cPKE.\Enc(\cPKE.\ek_{i,b},\msg^*_{i,b})$ for $i\in [\secp]$ and $b\in \bit$.  

This modification is just conceptual and we have 
$\Pr[\hybi{2}=1] = \Pr[\hybi{3}=1]$. 
\item[$\hybi{4}$:] This is identical to $\hybi{3}$ except for a conceptual modification that the measurement of the returned key $\widetilde{\qdk}$ is deferred until the end of the experiment. For clarity, we give the full description of this experiment.
\begin{enumerate}
            \item  The challenger generates $(\cPKE.\ek_{i,b},\cPKE.\dk_{i,b})\gets \cPKE.\KG(1^\secp)$ for $i\in [\secp]$ and $b\in \bit$,
            sets 
         $\ek\seteq \{\cPKE.\ek_{i,b}\}_{i\in[\secp],b\in\bit}$ and
         $\qdk\seteq \bigotimes_{i\in[\secp]}\frac{1}{\sqrt{2}}\left(\ket{0}\ket{\cPKE.\dk_{i,0}}+\ket{1}\ket{\cPKE.\dk_{i,1}}\right)$,   
         and sends $\ek$ and $\qdk$ to $\qA$. 
            \item \label{step:return_key}
            $\qA$ queries  $\widetilde{\qdk}$
            to the verification oracle. The challenger returns $\top$ to $\qA$ as the response from the oracle.
            \item The challenger chooses $\msg^*_{i,b}\gets \bit^{\ell}$ for $i\in [\secp]$ and $b\in \bit$ 
           generates
$\cPKE.\ct^*_{i,b}\gets \cPKE.\Enc(\cPKE.\ek_{i,b},\msg^*_{i,b})$ for $i\in [\secp]$ and $b\in \bit$, and sends $\ct^*\seteq \{\cPKE.\ct^*_{i,b}\}_{i\in [\secp],b\in \bit}$  to $\qA$.
            \item \label{step:check_msg_prime}
            $\qA$ outputs $\msg'=\msg'_1\concat \ldots\concat \msg'_\secp$. 
          The challenger outputs $0$ as the final outcome of this experiment if $\msg'_i\notin \{\msg^*_{i,0},\msg^*_{i,1}\}$ for some $i\in [\secp]$. 
           \item \label{step:finalize_two}
           Otherwise, 
          the challenger applies a binary-outcome measurement $(\mat{I}-\Pi_{\mathrm{vrfy}}^{\vk},\Pi_{\mathrm{vrfy}}^{\vk})$ on $\widetilde{\qdk}$where $\Pi_{\mathrm{vrfy}}^{\vk}$ is the projection defined in the description of $\PKESKL.\qVrfy$. 
          the challenger outputs the outcome of the measurement as the final outcome of this experiment. 
        \end{enumerate} 
        
By the deferred measurement principle, we have $\Pr[\hybi{3}=1] = \Pr[\hybi{4}=1]$. 
\item[$\hybi{5}$:] 
This is identical to $\hybi{4}$ except that
the challenger measures the returned key $\widetilde{\qdk}$ in the \emph{computational basis} instead of applying the projective measurement $(\mat{I}-\Pi_{\mathrm{vrfy}}^{\vk},\Pi_{\mathrm{vrfy}}^{\vk})$ in Step \ref{step:finalize_two}, and the condition to output $1$ is modified as follows: 
\begin{itemize}
            \item 
            Let $\{\widetilde{b}_i,\cPKE.\widetilde{\dk}_i\}_{i\in [\secp]}$  be the outcome of the measurement of $\widetilde{\qdk}$ in the computational basis.  
            If there is $i\in [\secp]$ such that $\cPKE.\widetilde{\dk}_i\ne \cPKE.\dk_{i,\widetilde{b}_i}$, the challenger outputs $0$ as the final outcome of this experiment. 
            Otherwise, define $\mathbf{b}=b_1\concat\ldots\concat b_\secp \in \bit^\secp$ in  such a way that $\msg'_i=\msg^*_{i,b_i}$ for $i\in [\secp]$.  Note that such $\mathbf{b}$ must exist since this step is invoked only when the challenger does not output $0$ in Step~\ref{step:check_msg_prime}.\footnote{If $\msg^*_{i,0}=\msg^*_{i,1}$ (which happens with a negligible probability), then we
            set $b_i\seteq 0$.}  
           If there is $i\in [\secp]$ such that $\widetilde{b}_{i}\ne b_{i}$,  
           the challenger outputs $1$ and otherwise $0$ as the final output of the experiment. 
        \end{itemize}
        
We prove that 
if $\Pr[\hybi{5}=1]=\negl(\secp)$, then it holds that $\Pr[\hybi{4}=1]=\negl(\secp)$.  
The intuition is as follows: If we have $\widetilde{b}_i=b_i$ with overwhelming probability, then 
$\widetilde{\qdk}$ has a negligible amplitude on $\bigotimes_{i\in[\secp]}\ket{b'_i}\ket{\cPKE.\widetilde{\dk}_{b'_i}}$ for all $\mathbf{b}'\ne \mathbf{b}$. 
In this case, the probability that $\widetilde{\qdk}$ is projected onto $\Pi_{\mathrm{vrfy}}^{\vk}$ is negligible since the right key $\qdk$ has an exponentially small amplitude on  $\bigotimes_{i\in[\secp]}\ket{b_i}\ket{\cPKE.\widetilde{\dk}_{b_i}}$.
See \cref{lem:PKE_SKL_4-to-5} for the detail.
\item[$\hybi{6}$:] This is identical to $\hybi{5}$ except that the challenger chooses $i^*\gets [\secp]$  at the beginning of the experiment and the condition to output $1$ is modified to that 
$\widetilde{b}_{i^*}\ne b_{i^*}$ 
holds for the a priori chosen $i^*$ instead of for some $i\in [\secp]$.  

Whenever there is $i$ such that $\widetilde{b}_{i}\ne b_{i}$, the probability that $i^*\gets [\secp]$ satisfies $\widetilde{b}_{i^*}\ne b_{i^*}$ is at least $\frac{1}{\secp}$. Thus, we have $\Pr[\hybi{6}=1]\ge \frac{1}{\secp}\Pr[\hybi{5}=1]$.

\item[$\hybi{7}$:] This is identical to $\hybi{6}$ except that challenger measures the register $\qregidx{DK}{i^*}$ of the decryption key $\qdk$ in the computational basis before giving $\qdk$ to $\qA$. (See the description of $\PKESKL.\qKG$ for the definition of register $\qregidx{DK}{i^*}$.)

Note that the measurement of $\qregidx{DK}{i^*}$ in the computational basis yields either $(0,\cPKE.\dk_{i^*,0})$ or $(1,\cPKE.\dk_{i^*,1})$.  In particular, there are only two possible outcomes. 
Thus, by \Cref{lem:BZ}, 
we have $\Pr[\hybi{7}=1]\ge \frac{1}{2}\Pr[\hybi{6}=1]$.

\item[$\hybi{8}$:] This is identical to $\hybi{7}$ except that the collapsing caused by measuring $\qregidx{DK}{i^*}$ is simulated by classical randomness. That is, the challenger chooses $b^*\gets \bit$ at the beginning and sets   
 \[\qdk\seteq \bigotimes_{i\in[\secp]\setminus \{i^*\}}\frac{1}{\sqrt{2}}\left(\ket{0}\ket{\cPKE.\dk_{i,0}}+\ket{1}\ket{\cPKE.\dk_{i,1}}\right)_{\qregidx{DK}{i}}\otimes (\ket{b^*}\ket{\cPKE.\dk_{i^*,b^*}})_{\qregidx{DK}{i^*}}.
\]

It is easy to see that $\hybi{7}$ and $\hybi{8}$ are identical from the view of $\qA$, and thus we have $\Pr[\hybi{7}=1]=\Pr[\hybi{8}=1]$. In \cref{lem:PKE_SKL_8}, we prove that $\Pr[\hybi{8}=1]=\negl(\secp)$ by using the OW-CPA security (which is implied by CoIC-KLA security) of $\PKESKL$. 
\end{description}

By combining the above, we have $\Pr[\hybi{0}=1]=\negl(\secp)$. This means that $\PKESKL$ is $\owKLA$ secure. 
We are left to prove \cref{lem:PKE_SKL_1-to-2,lem:PKE_SKL_4-to-5,lem:PKE_SKL_8}
\begin{lemma}\label{lem:PKE_SKL_1-to-2}
It holds that $\abs{\Pr[\hybi{1}=1] - \Pr[\hybi{2}=1]} = \negl(\secp)$ if $\cPKE$ is $\cKLA$ secure.
\end{lemma}
\begin{proof}

We define additional hybrids $\hybi{1.j}$ for $j\in [\secp+1]$ as follows.

\begin{description}
\item[$\hybi{1.j}$:] This is identical to $\hybi{1}$ except that $\cPKE.\ct^*_{i,b}$ is generated as 
\[\cPKE.\ct^*_{i,b}\gets
\begin{cases}
\cPKE.\Enc(\cPKE.\ek_{i,b},\msg^*_{i,a_i\oplus b}) & i<j
\\
\cPKE.\Enc(\cPKE.\ek_{i,b},\msg^*_{i,a_i}) & i\ge j
\end{cases}
\]
for $i\in [\secp]$. 
\end{description}
Clearly, we have $\hybi{1}=\hybi{1.1}$ and $\hybi{2}=\hybi{1.\secp+1}$. Thus, it suffices to prove that 
$
\left|\Pr[\hybi{1.j+1}=1]-\Pr[\hybi{1.j}=1]\right|=\negl(\secp). 
$
Remark that the only difference between $\hybi{1.j+1}$ and $\hybi{1.j}$ is the way of generating $\cPKE.\ct^*_{j,1}$.
To show that $\left|\Pr[\hybi{1.j+1}=1]-\Pr[\hybi{1.j}=1]\right|=\negl(\secp)$, 
we construct $\qB$ against $\cKLA$ security of $\cPKE$ as follows. 
\begin{description}
\item[$\qB(\cPKE.\ek^*_0,\cPKE.\ek^*_1,\qdk^*)$:] It works as follows.
\begin{enumerate}
            \item  Generate $(\cPKE.\ek_{i,b},\cPKE.\dk_{i,b})\gets \cPKE.\KG(1^\secp)$ for $i\in [\secp]\setminus \{j\}$ and $b\in \bit$ and
            set
            $
            \cPKE.\ek_{j,b}\seteq  \cPKE.\ek^*_b
            $
            for $b\in \bit$. 
        Set  
         $\ek\seteq \{\cPKE.\ek_{i,b}\}_{i\in[\secp],b\in\bit}$ and
         \[
         \qdk\seteq \bigotimes_{i\in[\secp]\setminus \{j\}}\frac{1}{\sqrt{2}}\left(\ket{0}\ket{\cPKE.\dk_{i,0}}+\ket{1}\ket{\cPKE.\dk_{i,1}}\right)_{\qregidx{DK}{i}}\otimes \qdk^*_{\qregidx{DK}{j}}.\] 
         This implicitly defines $\vk\seteq \{\cPKE.\dk_{i,b}\}_{i\in[\secp],b\in\bit}$ where $\cPKE.\dk_{j,b}$ is the decryption key corresponding to $\cPKE.\dk_{j,b}$ chosen by the external challenger for $b\in \bit$ (but $\qB$ cannot know $\vk$). \item   Send $\ek$ and $\qdk$ to $\qA$ and receives the verification query $\widetilde{\qdk}$ from $\qA$. 
         \item Apply a binary-outcome measurement $(\mat{I}-\Pi_{\mathrm{vrfy}}^{\vk},\Pi_{\mathrm{vrfy}}^{\vk})$ on $\widetilde{\qdk}$. 
         This is possible by simulating the projection on $\{\qregidx{DK}{i}\}_{i\neq j}$ by itself while forwarding $\qregidx{DK}{j}$ to its own verification oracle. 
         If the outcome is $0$,   output $0$. Otherwise, return $\top$ to $\qA$ as the response from the oracle. 
            \item Choose $\msg^*_{i,b}\gets \bit^{\ell}$ for $i\in [\secp]$ and $b\in \bit$ and $a_i\gets \bit$ for $i\in[\secp]\setminus \{j\}$, send $(\msg^*_{j,0},\msg^*_{j,1})$ to the external challenger, and receive $(\cPKE.\ct^*_{0},\cPKE.\ct^*_{1})$ from the challenger.  
            This implicitly defines $a_j\gets \bit$ and $\beta \gets \bit$ where the challenger generates 
            $\cPKE.\ct^*_{0}\seteq \cPKE.\Enc(\cPKE.\ek_j,m^*_{j,a_j})$
            and 
            $\cPKE.\ct^*_{1}\seteq \cPKE.\Enc(\cPKE.\ek_j,m^*_{j,a_j\oplus \beta})$ (but $\qB$ cannot know $a_j$ or $\beta$).\footnote{Here, $\beta$ plays the role of $b$ in the experiment $\expb{\cPKE,\qB}{coic}{kla}(1^\secp)$ in \cref{def:cKLA_PKE}. This is because $b$ is used in another meaning in this section. \takashi{If this is confusing, we may change the notation (for example, replace $b$ with $\beta$ in \cref{def:cKLA_PKE}).}}
    \item 
           Generate
$\cPKE.\ct^*_{i,b}\gets, \cPKE.\Enc(\cPKE.\ek_{i,b},\msg^*_{i,a_i\oplus b})$ for $i\in [\secp]\setminus \{j\}$ and $b\in \bit$, set
$\cPKE.\ct^*_{j,b}\seteq \cPKE.\ct^*_{b}$ for $b\in \bit$, 
send $\ct^*\seteq \{\cPKE.\ct^*_{i,b}\}_{i\in [\secp],b\in \bit}$  to $\qA$, and receive $\msg'=\msg'_1\concat \ldots\concat \msg'_\secp$ from $\qA$.
            \item 
           Output $1$ if $\msg'_i\in \{\msg^*_{i,0},\msg^*_{i,1}\}$ for all $i\in [\secp]$ and otherwise output $0$. 
        \end{enumerate} 
\end{description}
We have 
\begin{align}
&\advb{\cPKE,\qB}{coic}{kla}(\secp)\\
&=2\left|\Pr[\qB(\cPKE.\ek^*_0,\cPKE.\ek^*_1,\qdk^*)=\beta]-\frac{1}{2}\right|\\
&=\left|\Pr[\qB(\cPKE.\ek^*_0,\cPKE.\ek^*_1,\qdk^*)=1|\beta=0]-\Pr[\qB(\cPKE.\ek^*_0,\cPKE.\ek^*_1,\qdk^*)=1|\beta=1]\right|\\
&=\left|\Pr[\hybi{1.j+1}=1]-\Pr[\hybi{1.j}=1]\right|
\end{align}
where $(\cPKE.\ek^*_0,\cPKE.\dk^*_0)\gets \cPKE.\KG(1^\secp)$, 
     $(\cPKE.\ek^*_1,\cPKE.\dk^*_1)\gets \cPKE.\KG(1^\secp)$, and 
     $\qdk^*\seteq \frac{1}{\sqrt{2}}(\ket{\cPKE.\dk^*_0}+\ket{\cPKE.\dk^*_1})$. 
     Thus, $\left|\Pr[\hybi{1.j+1}=1]-\Pr[\hybi{1.j}=1]\right|=\negl(\secp)$ by the $\cKLA$ security of $\cPKE$.
     This completes the proof of \Cref{lem:PKE_SKL_1-to-2}. 
\end{proof}
\begin{lemma}\label{lem:PKE_SKL_4-to-5}
If $\Pr[\hybi{5}=1]=\negl(\secp)$, then it holds that $\Pr[\hybi{4}=1]=\negl(\secp)$. 
\end{lemma}
\begin{proof}
Let $\epsilon \seteq \Pr[\hybi{4}=1]$.
For 
$\vk= \{\cPKE.\dk_{i,b}\}_{i\in[\secp],b\in\bit}$ and 
$\mathbf{b}=b_1\concat\ldots \concat b_\secp\in \bit^\secp$, let $\event_{\mathbf{b}}^{\vk}$ be the event that 
$\vk$ is chosen as a verification key and 
$\msg'_i=\msg^*_{i,b_i}$ for all $i\in [\secp]$. Let $\widetilde{\qdk}_{\mathbf{b}}^{\vk}$ be the state of the returned key conditioned on $\event_{\mathbf{b}}^{\vk}$. 
Clearly, we have 
\[
\sum_{\vk,\mathbf{b}}\Pr[\event_{\mathbf{b}}^{\vk}]\cdot\Tr(\Pi_{\mathrm{vrfy}}^{\vk}\widetilde{\qdk}_{\mathbf{b}}^{\vk})=\epsilon.
\]
Let $\Good$ be a subset defined as 
\[
\Good\seteq \left\{(\vk,\mathbf{b}): \Tr(\Pi_{\mathrm{vrfy}}^{\vk}\widetilde{\qdk}_{\mathbf{b}}^{\vk})\geq \frac{\epsilon}{2}\right\}.
\]
Then, by a standard averaging argument, it holds that  
\begin{align} \label{eq:good_probability}
\sum_{(\vk,\mathbf{b})\in \Good}\Pr[\event_{\mathbf{b}}^{\vk}]\ge \frac{\epsilon}{2}.  
\end{align}

For $\vk= \{\cPKE.\dk_{i,b}\}_{i\in[\secp],b\in\bit}$ and $\mathbf{b}=b_1\concat\ldots\concat b_\secp$, let $\Pi_{\ne \mathbf{b}}^{\vk}$ be a projection defined as follows:
\begin{align}
\Pi_{\ne \mathbf{b}}^{\vk}\seteq \sum_{\mathbf{b'}\in \bit^\secp \setminus\{\mathbf{b}\}}\bigotimes_{i\in [\secp]} \ket{b'_i}\ket{\cPKE.\dk_{i,b'_i}}\bra{b'_i}\bra{\cPKE.\dk_{i,b'_i}}. 
\end{align}
Then, by the definition of $\hybi{5}$, one can see that 
\[\Pr[\hybi{5}=1]= \sum_{\vk,\mathbf{b}}\Pr[\event_{\mathbf{b}}^{\vk}]\cdot\Tr\left(\Pi_{\ne \mathbf{b}}^{\vk}\widetilde{\qdk}_{\mathbf{b}}^{\vk}\right).\]  

Then, we show the following proposition.
\begin{proposition}\label{prop:entropy}
For any $(\vk,\mathbf{b})\in \Good$, it holds that
\begin{align}
   \Tr\left(\Pi_{\ne \mathbf{b}}^{\vk}\widetilde{\qdk}_{\mathbf{b}}^{\vk}\right)\ge \frac{\epsilon}{4}-2^{-\secp}.
\end{align}
\end{proposition}
\begin{proof}[Proof of \Cref{prop:entropy}]
By diagonalization, we can write 
\[
\widetilde{\qdk}_{\mathbf{b}}^{\vk}=\sum_{j=1}^{N}p_j \ket{\psi_j}\bra{\psi_j}
\]
where $0<p_j\le 1$, $\sum_{j=1}^{N}p_j=1$, and $\left|\bra{\psi_j}\ket{\psi_j}\right|=1$.
For each $j\in [N]$, it holds that 
\begin{align}
    &\Tr(\Pi_{\mathrm{vrfy}}^{\vk}\ket{\psi_j}\bra{\psi_j})\\
    &=2^{-\secp}\left\|\left(\sum_{\mathbf{b'}\in \bit^\secp}\bigotimes_{i\in [\secp]} \bra{b'_i}\bra{\cPKE.\dk_{i,b'_i}}\right)\ket{\psi_j}\right\|^2\\ 
     &\le 2^{-\secp+1}\left(\left\|\left(\sum_{\mathbf{b'}\in \bit^\secp\setminus \{\mathbf{b}\}}\bigotimes_{i\in [\secp]} \bra{b'_i}\bra{\cPKE.\dk_{i,b'_i}}\right)\ket{\psi_j}\right\|^2+\left\|\left(\bigotimes_{i\in [\secp]} \bra{b_i}\bra{\cPKE.\dk_{i,b_i}}\right)\ket{\psi_j}\right\|^2\right)\\
      &\le 2^{-\secp+1}\left(\left(2^{\secp}-1\right)\sum_{\mathbf{b'}\in \bit^\secp\setminus \{\mathbf{b}\}}\left\|\left(\bigotimes_{i\in [\secp]} \bra{b'_i}\bra{\cPKE.\dk_{i,b'_i}}\right)\ket{\psi_j}\right\|^2+ 1\right)\\
     &\le 2\left\|\Pi_{\ne \mathbf{b}}^{\vk}\ket{\psi_j}\right\|^2+2^{-\secp+1}
\end{align} 
where the inequalities in the third and fourth lines follow from Cauchy–Schwarz inequality. 

Then, it holds that 
\begin{align}
    &\Tr(\Pi_{\mathrm{vrfy}}^{\vk}\widetilde{\qdk}_{\mathbf{b}}^{\vk})\\
    &=\sum_{j=1}^{N}p_j\Tr(\Pi_{\mathrm{vrfy}}^{\vk}\ket{\psi_j}\bra{\psi_j})\\
  &\le \sum_{j=1}^{N}p_j \left(2\left\|\Pi_{\ne \mathbf{b}}^{\vk}\ket{\psi_j}\right\|^2+2^{-\secp+1}\right)\\
  &=2\Tr\left(\Pi_{\ne \mathbf{b}}^{\vk}\widetilde{\qdk}_{\mathbf{b}}^{\vk}\right)+2^{-\secp+1}.
\end{align}
Since we assume  $(\vk,\mathbf{b})\in \Good$, it holds that $\Tr(\Pi_{\mathrm{vrfy}}^{\vk}\widetilde{\qdk}_{\mathbf{b}}^{\vk})\ge \frac{\epsilon}{2}$. 
By combining the above, \Cref{prop:entropy} is proven.
\end{proof}

Then, we have 
\begin{align}
\Pr[\hybi{5}=1]&= \sum_{\vk,\mathbf{b}}\Pr[\event_{\mathbf{b}}^{\vk}]\cdot\Tr\left(\Pi_{\ne \mathbf{b}}^{\vk}\widetilde{\qdk}_{\mathbf{b}}^{\vk}\right)\\
&\ge \sum_{(\vk,\mathbf{b})\in \Good}\Pr[\event_{\mathbf{b}}^{\vk}]\cdot\Tr\left(\Pi_{\ne \mathbf{b}}^{\vk}\widetilde{\qdk}_{\mathbf{b}}^{\vk}\right)\\
&\ge \sum_{(\vk,\mathbf{b})\in \Good}\Pr[\event_{\mathbf{b}}^{\vk}]\cdot \left(\frac{\epsilon}{4}-2^{-\secp}\right)\\
&\ge \frac{\epsilon}{2}\cdot \left(\frac{\epsilon}{4}-2^{-\secp}\right) 
\end{align}
where the second inequality follows from \Cref{prop:entropy} and the third inequality follows from Eq. \ref{eq:good_probability}. 
Recalling that $\epsilon=\Pr[\hybi{4}=1]$, the above inequality implies \Cref{lem:PKE_SKL_4-to-5}.  
\end{proof}
\begin{lemma}\label{lem:PKE_SKL_8}
It holds that $\Pr[\hybi{8}=1]=\negl(\secp)$ if $\cPKE$ is OW-CPA secure. 
\end{lemma}
\begin{proof}
For clarity, we give the full description of $\hybi{8}$ below. 
\begin{description}
\item[$\hybi{8}$:] It works as follows:
\begin{enumerate}
            \item  The challenger chooses $i^*\gets [\secp]$ and $b^*\in \bit$, 
            generates $(\cPKE.\ek_{i,b},\cPKE.\dk_{i,b})\gets \cPKE.\KG(1^\secp)$ for $i\in [\secp]$ and $b\in \bit$,
            sets 
         $\ek\seteq \{\cPKE.\ek_{i,b}\}_{i\in[\secp],b\in\bit}$ and
          \[\qdk\seteq \bigotimes_{i\in[\secp]\setminus \{i^*\}}\frac{1}{\sqrt{2}}\left(\ket{0}\ket{\cPKE.\dk_{i,0}}+\ket{1}\ket{\cPKE.\dk_{i,1}}\right)_{\qregidx{DK}{i}}\otimes (\ket{b^*}\ket{\cPKE.\dk_{i^*,b^*}})_{\qregidx{DK}{i^*}}
          \] 
         and sends $\ek$ and $\qdk$ to $\qA$. 
            \item
            $\qA$ queries  $\widetilde{\qdk}$
            to the verification oracle. The challenger returns $\top$ to $\qA$ as the response from the oracle.
            \item The challenger chooses $\msg^*_{i,b}\gets \bit^{\ell}$ for $i\in [\secp]$ and $b\in \bit$, 
           generates
$\cPKE.\ct^*_{i,b}\gets \cPKE.\Enc(\cPKE.\ek_{i,b},\msg^*_{i,b})$ for $i\in [\secp]$ and $b\in \bit$, and sends $\ct^*\seteq \{\cPKE.\ct^*_{i,b}\}_{i\in [\secp],b\in \bit}$  to $\qA$.
            \item \label{step:check_msg_prime_two}
            $\qA$ outputs $\msg'=\msg'_1\concat \ldots\concat \msg'_\secp$. 
           The challenger outputs $0$ as the final output of the experiment if $\msg'_i\notin \{\msg^*_{i,0},\msg^*_{i,1}\}$ for some $i\in [\secp]$. 
           \item 
           Otherwise, the challenger measures $\widetilde{\qdk}$ in the computational basis, and let $\{\widetilde{b}_i,\cPKE.\widetilde{\dk}_i\}_{i\in [\secp]}$  be the outcome.  
            If there is $i\in [\secp]$ such that $\cPKE.\widetilde{\dk}_i\ne \cPKE.\dk_{i,\widetilde{b}_i}$, the challenger outputs $0$ as the final outcome of this experiment. 
            Otherwise, define $\mathbf{b}=b_1\concat\ldots\concat b_\secp \in \bit^\secp$ in  such a way that $\msg'_i=\msg^*_{i,b_i}$ for $i\in [\secp]$. Note that such $\mathbf{b}$ must exist since this step is invoked only when the challenger does not output $0$ in Step~\ref{step:check_msg_prime_two}.\footnote{If $\msg^*_{i,0}=\msg^*_{i,1}$ (which happens with a negligible probability), then we
            set $b_i\seteq 0$.}  
           If  $\widetilde{b}_{i^*}\ne b_{i^*}$,  
           the challenger outputs $1$ and otherwise $0$ as the final output of the experiment. 
        \end{enumerate}
\end{description} 

Suppose that we simulate $\hybi{8}$ for $\qA$ while embedding a problem instance of the OW-CPA security of $\cPKE$ into $\cPKE.\ek_{i^*,b^*\oplus 1}$ and  $\cPKE.\ct^*_{i^*,b^*\oplus 1}$. Remark that this is possible without knowing $\cPKE.\dk_{i^*,b^*\oplus 1}$. 
Suppose that $\hybi{8}=1$ occurs in the simulated execution. Then, we in particular have $\msg'_{i^*}=\msg^*_{i^*,b_{i^*}}$,  $\cPKE.\widetilde{\dk}_{i^*}=\cPKE.\dk_{i^*,\widetilde{b}_{i^*}}$, and  $\widetilde{b}_{i^*}\ne b_{i^*}$. 
We consider the following two sub-cases.
\begin{enumerate}
    \item If $b_{i^*}=b^*$, then we have 
    $\widetilde{b}_{i^*}=b^*\oplus 1$.  This implies $\cPKE.\widetilde{\dk}_{i^*}=\cPKE.\dk_{i^*,b^*\oplus 1}$. 
    Then we can decrypt $\cPKE.\ct_{i^*,b^*\oplus 1}$ by honestly running the decryption algorithm with $\cPKE.\dk_{i^*,b^*\oplus 1}$. This contradicts the OW-CPA security of $\cPKE$. 
    \item If $b_{i^*}\ne b^*$, then we have    $\msg'_{i^*}=\msg^*_{i^*,b^*\oplus 1}$, which is the message encrypted in $\cPKE.\ek_{i^*,b^*\oplus 1}$. This means that we can break the OW-CPA security of $\cPKE$. 
\end{enumerate}
Neither of them occurs with a non-negligible probability assuming the OW-CPA security of $\cPKE$. 
Thus, $\Pr[\hybi{8}=1]=\negl(\secp)$. This completes the proof of \Cref{lem:PKE_SKL_8}. 
\end{proof}
This completes the proof of \Cref{thm:PKE_OW-KLA}.
\end{proof}

\begin{remark}[On OMUR]\label{rem:OMUR}
We can show that $\PKESKL$ constructed above also satisfies OMUR. Since there is a generic conversion to add OMUR as shown in \cref{lem:add_OMUR} anyway, we only give a proof sketch.

We reduce OMUR to 1-key OW-KLA security. Suppose that there is an adversary that breaks OMUR, i.e., passes the verification twice. Then roughly speaking, we can use it to break 1-key OW-KLA security by sending one of them to the verification oracle and using the other one to decrypt the challenge message. There is an issue that the reduction algorithm may make only one verification query while the adversary against OMUR may make arbitrarily many verification queries. To resolve this issue, we can use a similar idea to that used in the proof of \cref{lem:add_OMUR}. The reduction algorithm guesses the first two queries to be accepted. Conditioned on that the guess is correct, the reduction algorithm can simulate the verification oracle by simply returning $\bot$ to all queries except for the two queries that are guessed to be accepted until the adversary make the second guessed query. 
The guess is correct with probability $\binom{Q}{2}^{-1}$ where $Q$ is the number of queries. Thus, the reduction works with a polynomial security loss. 
Since we already proved that $\PKESKL$ is 1-query OW-KLA secure (\cref{thm:PKE_OW-KLA}), the above reduction shows that it satisfies OMUR. 

\end{remark}

\section{Attribute-Based Encryption with Secure Key Leasing}\label{sec:ibe_abe_skl}

\subsection{Definitions}\label{sec:ABE_SKL_def}
\begin{definition}[ABE with Secure Key Leasing]
An ABE-SKL scheme $\ABESKL$ is a tuple of six algorithms $(\Setup, \qKG, \Enc, \qDec,\qCert,\certvrfy)$. 
Below, let $\cX = \{ \cX_\secp \}_\secp$, $\cY= \{ \cY_\secp \}_\secp$, and $R= \{ R_\secp: \cX_\secp \times \cY_\secp \to \bit \}_\secp$ be the ciphertext space, the key attribute space, and the associated relation of $\ABESKL$, respectively.
\begin{description}
\item[$\Setup(1^\secp)\ra(\pk,\msk)$:] The setup algorithm takes a security parameter $1^\lambda$, and outputs a public key $\pk$ and master secret key $\msk$.
\item[$\qKG(\msk,y)\ra(\qusk,\vk)$:] The key generation algorithm takes a master secret key $\msk$ and a key attribute $y \in \calY$, and outputs a user secret key $\qusk$ and a verification key $\vk$.


\item[$\Enc(\pk,x,m)\ra\ct$:] The encryption algorithm takes a public key $\pk$, a ciphertext attribute $x\in \cX$, and a plaintext $m$, and outputs a ciphertext $\ct$.

\item[$\qDec(\qusk,x,\ct)\ra z$:] The decryption algorithm takes a user secret key $\qusk$, a ciphertext attribute $x$, and a ciphertext $\ct$ and outputs a value $z\in \{\bot\}\cup \bin^\ell$.


\item[$\qVrfy(\vk,\qusk^\prime)\ra\top/\bot$:] The verification algorithm takes a verification key $\vk$ and a quantum state $\qusk^\prime$, and outputs $\top$ or $\bot$.

\item[Decryption correctness:]For every $x \in \cX$ and $y \in \cY$ satisfying $R(x,y)=1$, we have
\begin{align}
\Pr\left[
\qDec(\qusk, x, \ct) \allowbreak = m
\ \middle |
\begin{array}{ll}
(\pk,\msk) \la \Setup(1^\secp)\\
(\qusk,\vk)\gets\qKG(\msk,y)\\
\ct\gets\Enc(\pk,x,m)
\end{array}
\right] 
=1-\negl(\secp).
\end{align}

\item[Verification correctness:] For every $y \in \cY$, we have 
\begin{align}
\Pr\left[
\qVrfy(\vk,\qusk)=\top
\ \middle |
\begin{array}{ll}
(\pk,\msk) \la \Setup(1^\secp)\\
(\qusk,\vk)\gets\qKG(\msk,y)\\
\end{array}
\right] 
=1-\negl(\secp).
\end{align}

\end{description}
\end{definition}

\begin{definition}[Adaptive Indistinguishability against Key Leasing Attacks]\label{def:ada_lessor_ABESKL}
We say that an ABE-SKL scheme $\ABESKL$ for relation $R:\cX\times \cY \to \bin$ is secure against adaptive indistinguishability against key leasing attacks ($\ADAINDKLA$), if it satisfies the following requirement, formalized from the experiment $\expc{\qA,\ABESKL}{ada}{ind}{kla}(1^\secp,\coin)$ between an adversary $\qA$ and a challenger:
        \begin{enumerate}
            \item At the beginning, the challenger runs $(\pk,\msk)\gets\Setup(1^\secp)$
            and initialize the list $\List{\qKG}$ to be an empty set. 
            Throughout the experiment, $\qA$ can access the following oracles.
            \begin{description}
            \item[$\Oracle{\qKG}(y)$:] Given $y$, it finds an entry of the form $(y,\vk,V)$ from $\List{\qKG}$. If there is such an entry, it returns $\bot$.
            Otherwise, it generates $(\qusk,\vk)\la\qKG(\msk,y)$, sends $\qusk$ to $\qA$, and adds $(y,\vk,\bot)$ to $\List{\qKG}$.
            
            \item[$\Oracle{\qVrfy}(y,\qusk^\prime)$:] Given $(y,\qusk^\prime)$, it finds an entry $(y,\vk,V)$ from $\List{\qKG}$. (If there is no such entry, it returns $\bot$.) 
            It then runs $\decision \seteq \qVrfy(\vk,\qusk^\prime)$ and returns $\decision$ to $\qA$.
            If $V=\bot$, it updates the entry into $(y,\vk,\decision)$.
            \end{description}
            \item \label{ada_lessor_abe_challenge}
            When $\qA$ sends $(x^*,m_0,m_1)$ to the challenger, the challenger checks if for any entry $(y,\vk,V)$ in $\List{\qKG}$ such that $R(x^*,y)=1$, it holds that $V=\top$. If so, the challenger generates $\ct^*\la\Enc(\pk,x^*,m_\coin)$ and sends $\ct^*$ to $\qA$. Otherwise, the challenger outputs $0$. 
            \item 
            $\qA$ continues to make queries to $\Oracle{\qKG}(\cdot)$ and  $\Oracle{\qVrfy}(\cdot,\cdot )$.
            However, $\qA$ is not allowed to send a key attribute $y$ such that $R(x^*,y)=1$ to $\Oracle{\qKG}$.
            \item $\qA$ outputs a guess $\coin^\prime$ for $\coin$. The challenger outputs $\coin'$ as the final output of the experiment.
        \end{enumerate}
        For any QPT $\qA$, it holds that
\ifnum\llncs=0        
\begin{align}
\advc{\ABESKL,\qA}{ada}{ind}{kla}(\secp) \seteq \abs{\Pr[\expc{\ABESKL,\qA}{ada}{ind}{kla} (1^\secp,0) \ra 1] - \Pr[\expc{\ABESKL,\qA}{ada}{ind}{kla} (1^\secp,1) \ra 1] }\leq \negl(\secp).
\end{align}
\else
\begin{align}
\advb{\PKFESKL,\qA}{ada}{lessor}(\secp) 
&\seteq \abs{\Pr[\expb{\ABESKL,\qA}{ada}{lessor} (1^\secp,0) \ra 1] - \Pr[\expb{\ABESKL,\qA}{ada}{lessor} (1^\secp,1) \ra 1] }\\
&\leq \negl(\secp).
\end{align}
\fi
\end{definition}

\begin{remark}\label{remark:same_query_remark}
In \Cref{def:ada_lessor_ABESKL}, the key generation oracle returns $\bot$ if the same $y$ is queried more than once. To handle the situation where multiple keys for the same attribute $y$ are generated,  we need to manage indices for $y$ such as $(y,1,vk_1,V_1), (y,2,vk_2,V_2)$. 
Although we can reflect the index management in the definition, it complicates the definition and prevents readers from understanding the essential idea.
Thus, we use the simplified definition above.
\end{remark}

We also consider relaxed versions of the above security notion. 

\begin{definition}[Selective indistinguishability against key leasing attacks]\label{def:sel_ind_ABE_SKL}
We consider selective indistinguishability against key leasing attacks ($\SELINDKLA$). For doing so, we consider the same security game as that for $\ADAINDKLA$ except that the adversary $\qA$ should declare its target $x^*$ at the beginning of the game (even before it is given $\pk$).
We then define the advantage $\advc{\ABESKL,\qA}{sel}{ind}{kla}(\secp)$ for the selective case similarly. We say $\ABESKL$ is secure against selective indistinguishability against key leasing attack if for any QPT adversary $\qA$, $\advc{\ABESKL,\qA}{sel}{ind}{kla}(\secp)$ is negligible.    
\end{definition}
We also consider the following security notion where we introduce additional restriction that the number of distinguishing keys that are issued (and eventually returned) before $\ct^*$ is generated is bounded by some predetermined parameter $q$.
Here, distinguishing key refers to a key that can decrypt the challenge ciphertext if it is not returned. 
\begin{definition}[Bounded Distinguishing Key $\ADAINDKLA$/$\SELINDKLA$ for ABE]\label{def:bounded_ABE_SKL}
For defining bounded distinguishing key $\ADAINDKLA$ security, we consider the same security game as that for $\ADAINDKLA$ (i.e., $\expc{\qA,\ABESKL}{ada}{ind}{kla}(1^\secp,\coin)$) except that we change the step \ref{ada_lessor_abe_challenge} in \cref{def:ada_lessor_ABESKL} with the following:
\begin{itemize}
\item[2'] When $\qA$ sends $(x^*,m_0,m_1)$ to the challenger, the challenger checks if there are at most $q$ entries $(y,\vk,V)$ in $\List{\qKG}$ such that $R(x^*,y)=1$ and for all these entries, $V=\top$. If so, the challenger generates $\ct^*\la\Enc(\pk,x^*,m_\coin)$ and sends $\ct^*$ to $\qA$. Otherwise, the challenger outputs $0$. 
\end{itemize}
We then define the advantage $\advc{\ABESKL,\qA,{\color{red}{q}}}{ada}{ind}{kla}(\secp)$ similarly to $\advc{\ABESKL,\qA}{ada}{ind}{kla}(\secp)$. 
We say $\ABESKL$ is $q$-bounded distinguishing key $\ADAINDKLA$ secure if for any QPT adversary $\qA$, $\advc{\ABESKL,\qA, q}{ada}{ind}{kla}(\secp)$ is negligible.
We also define $q$-bounded distinguishing key $\SELINDKLA$ security analogously by enforcing the adversary to output its target $x^*$ at the beginning of the game. 
\end{definition}
We emphasize that while the number of distinguishing keys that the adversary can obtain in the game is bounded by a fixed polynomial, the number of non-distinguishing keys (i.e., keys for $y$ with $R(x^*,y)=0$) can be unbounded.

\subsection{1-Bounded Distinguishing Key Construction}
\label{sec:ABE_SKL_one_bounded}

We construct an ABE-SKL scheme $\oneABESKL =(\Setup,\qKG,\Enc,\qDec,\qVrfy)$ for relation $R:\cX\times \cY \to \bit$ with $1$-bounded distinguishing key $\ADAINDKLA$/$\SELINDKLA$ security whose message space is $\bit^\ell$
by using the following building blocks.
\begin{itemize}
\item IND-KLA secure 
PKE-SKL 
$\SKL.(\qKG,\Enc,\qDec,\qVrfy)$.
Without loss of generality, 
we assume that $\skl.\ek \in \bit^{\sklpklen}$ and the randomness space used by $\SKL.\Enc$ is $\bit^{\sklrandlen}$ for some $\sklpklen(\secp)$ and $\sklrandlen(\secp)$. 
We also assume that the message space of $\SKL$ is $\bit^\ell$. 

\item Adaptively/Selectively secure ABE $\ABE.(\Setup,\KG,\Enc,\Dec)$ for relation $R$ with message space $\bit^\secp$.
\item A garbling scheme $\GC=(\Garble,\GCEval )$.
Without loss of generality, we assume that the labels of $\GC$ are in $\bit^\secp$.
\end{itemize}

\begin{description}
\item[$\Setup(1^\secp)$:] $ $
\begin{itemize}
    \item 
    For $i\in [\sklpklen]$ and $b\in \bin$, run $(\sfabe.\pk_{i,b},\sfabe.\msk_{i,b}) \gets \ABE.\Setup(1^\secp)$.
    \item Output $(\pk,\msk)\seteq (\{ \sfabe.\pk_{i,b}\}_{i\in [\sklpklen], b\in \bin },\{ \sfabe.\msk_{i,b}\}_{i\in [\sklpklen], b\in \bin })$.
\end{itemize}
\item[$\qKG(\msk,y)$:] $ $
\begin{itemize}
	\item Generate $(\skl.\ek,\skl.\qdk,\skl.\vk)\gets \SKL.\qKG(1^\secp)$.
	\item Run $\sfabe.\sk_{i } \gets \ABE.\KG(\ABE.\msk_{i,\skl.\ek[i]},y)$
	for $i\in [\sklpklen]$, where $\skl.\ek[i]$ denotes the $i$-th bit of the binary string $\skl.\ek$.
	\item Output $\qusk \seteq (\{ \sfabe.\sk_i \}_{i\in [\sklpklen] },\skl.\ek, \skl.\qdk)$ and $\vk \seteq \skl.\vk$.
\end{itemize}
\item[$\Enc(\pk,x,m)$:] $ $
\begin{itemize}
\item Choose $\sfR \chosen \bin^{\sklrandlen}$.
\item Construct circuit $E[m,\sfR]$, which is a circuit that takes as input an encryption key $\skl.\ek$ of $\SKL$ and outputs $\SKL.\Enc(\skl.\ek, m;\sfR)$.
\item Compute $(\{ \lab_{i,b} \}_{i\in [\sklpklen], b\in \bit }, \tildeE ) \gets \Garble(1^\secp, E[m,\sfR])$.
\item Run $\sfabe.\ct_{i,b} \la \ABE.\Enc(\sfabe.\pk_{i,b}, x, \lab_{i,b}  )$
for $i\in [\sklpklen]$ and $b\in \bit$.
\item Output $\ct \seteq (\{ \sfabe.\ct_{i,b} \}_{i\in [\sklpklen], b\in \bit }, \tildeE )$.
\end{itemize}
\item[$\qDec(\qusk,x,\ct)$:] $ $
\begin{itemize}
\item Parse $\qusk = (\{ \sfabe.\sk_i \}_{i\in [\sklpklen] },\skl.\ek, \skl.\qdk)$ and $\ct = (\{ \sfabe.\ct_{i,b} \}_{i\in [\sklpklen], b\in \bit }, \tildeE )$.
\item Compute $\lab_{i} \gets \ABE.\Dec(\ABE.\sk_{i},x,\sfabe.\ct_{i,\skl.\ek[i]})$ for $i\in [\sklpklen]$.
\item Compute $\skl.\ct=\GCEval(\tildeE, \{ \lab_i\}_{i\in [\sklpklen]} )$.
\item Compute and output $m' \gets \SKL.\qDec(\skl.\qdk,\skl.\ct)$.
\end{itemize}
\item[$\qVrfy(\vk,\qusk^\prime 
)$:] $ $
\begin{itemize}
\item Parse $\vk = \skl.\vk$ and $\qusk^\prime=(\{ \sfabe.\sk_i \}_{i\in [\sklpklen] },\skl.\ek', \skl.\qdk')$.
\item Compute and output $\SKL.\qVrfy(\skl.\vk,\skl.\qdk^\prime)$.
\end{itemize}
\end{description}

We show that the scheme satisfies decryption correctness. To see this, we first observe that the decryption algorithm correctly recovers labels of $\tildeE$ corresponding to the input $\skl.\ek$ by the correctness of $\ABE$. Therefore, $\skl.\ct$ recovered by the garbled circuit evaluation equals to $\SKL.\enc(\skl.\ek, m;\sfR)$ by the correctness of $\GC$. Then, the message $m$ is recovered in the last step by the correctness of $\SKL$.
We can also see that the verification correctness follows from that of $\SKL$. 

\begin{theorem}\label{thm:ABESKL_from_ABE}
If $\ABE$ is adaptively (resp., selectively) secure, $\GC$ is secure, and $\SKL$ is IND-KLA secure, then $\oneABE$ above is $1$-bounded distinguishing key $\ADAINDKLA$ (resp., $\SELINDKLA$) secure.
\end{theorem}
\begin{proof}[Proof of~\cref{thm:ABESKL_from_ABE}]
Here, we first focus on the proof for the case of $\ADAINDKLA$ and later mention the necessary modifications for the case of $\SELINDKLA$. 
Let $Q$ be the upper bound on the number of key queries to $\Oracle{\qKG}$ before the challenge phase.
We define a sequence of hybrid games.
\begin{description}
\item[$\hybi{0}$:] This is the same as $\expc{\oneABE,\qA, 1}{ada}{ind}{kla}(1^\secp,0)$. More specifically, it is as follows.
        \begin{enumerate}
            \item The challenger generates $(\sfabe.\pk_{i,b},\sfabe.\msk_{i,b}) \gets \ABE.\Setup(1^\secp)$
            for $i\in [\sklpklen]$ and $b\in \bin$ and sends $\pk\seteq \{\sfabe.\pk_{i,b}\}_{i,b}$ to the adversary $\qA$.
            The challenger then initializes the list $\List{\qKG}$ to be an empty set. 
            $\qA$ can access the following oracles.
            \begin{description}
            \item[$\Oracle{\qKG}(y^{(j)})$:] Given the $j$-th query $y^{(j)}$ with $j\in [Q]$, if there is an entry of the form $(y^{(j)},\vk,V)$, it outputs $\bot$. Otherwise, it generates $(\skl.\ek^{(j)},\skl.\qdk^{(j)},\skl.\vk^{(j)})\la\SKL.\qKG(1^\secp)$ and $\sfabe.\sk^{(j)}_{i} \la \ABE.\KG(\sfabe.\msk_{i,\skl.\ek^{(j)}[i]}, y^{(j)})$ for $i\in [\sklpklen]$,
            where $\skl.\ek^{(j)}[i]$ is the $i$-th bit of the binary string $\skl.\ek^{(j)}$. 
            It then sends $\qusk^{(j)} \seteq (\{ \sfabe.\sk^{(j)}_{i} \}_i ,\skl.\qdk^{(j)})$ to $\qA$ and adds $(y^{(j)},\skl.\vk^{(j)},\bot)$ to $\List{\qKG}$.
            \item[$\Oracle{\qVrfy}(y,\qusk^\prime)$:] Given $(y,\qusk^\prime)$, it finds an entry $(y,\vk,V)$ from $\List{\qKG}$ and parse $\qusk^\prime = (\{\sfabe.\sk^\prime_i\}_i,\skl.\qusk^\prime)$. (If there is no such entry, it returns $\bot$.) 
            It then parses $\vk=\skl.\vk$
            and returns $\decision \seteq \SKL.\qVrfy(\skl.\vk,\skl.\qdk^\prime)$ to $\qA$. It finally updates the entry into $(y,\vk,\decision )$ if $V=\bot$.
            \end{description}
            \item When $\qA$ sends $(x^*,m_0, m_1)$ to the challenger, the challenger checks whether there is at most one entry $(y,\vk,V)$ in $\List{\qKG}$ such that $R(x^*,y)=1$ and for that entry $V=\top$ holds.
            If so, the challenger generates $(\{ \lab_{i,b} \}_{i\in [\sklpklen], b\in \bit }, \tildeE ) \gets \Garble(1^\secp, E[m_0,\sfR])$
            and computes $\sfabe.\ct_{i,b} \la \ABE.\Enc(\sfabe.\pk_{i,b}, x^*, \lab_{i,b}  )$ for $i\in [\sklpklen]$ and $b\in \bit$.
            It then sends $\ct^* \seteq (\{ \sfabe.\ct_{i,b} \}_{i,b}, \tildeE )$ to $\qA$. 
            Otherwise (i.e., if there are multiple entries with $R(x^*,y)=1$ or if there is an entry with $R(x^*,y)=1$ and $V=\bot$), it aborts the game and outputs $0$. 
            
            \item $\qA$ continues to make queries to $\Oracle{\qKG}(\cdot)$ and  $\Oracle{\qVrfy}(\cdot,\cdot )$.
            However, $\qA$ is not allowed to send a key attribute $y$ such that $R(x^*,y)=1$ to $\Oracle{\qKG}$.
            \item $\qA$ outputs a guess $\coin^\prime$ for $\coin$. The challenger outputs $\coin'$ as the final output of the experiment.
        \end{enumerate}
\item[$\hybi{1}$:] 
            This game is the same as $\hybi{0}$ except that the challenger chooses random $\tilde{j} \gets [Q]$ at the beginning of the game. 
            Then, right before it computes the challenge ciphertext,  
            the challenger finds an index $j^*\in [Q]$ such that $R(x^*,y^{(j^*)})=1$. 
            If there is no such a query, we define $j^*\seteq 1$.%
            \footnote{Note that if there are multiple indices $j^*$ satisfying the above, the challenger aborts and outputs $0$ as specified in the previous game.
            Therefore, there is at most one such $j^*$.
            }
            The challenger then checks whether $\tilde{j}= j^*$.
            If so, the challenger continues the game until $\qA$ outputs its guess.
            Otherwise, it aborts the game and outputs $0$ as the outcome of the game. 
            
            Since the choice of $\tilde{j}$ is independent from the view of $\qA$ and the outcome of the game is $1$ only when $\tilde{j}=j^*$,
            we can easily see that 
            $\Pr[\hybi{1}=1] = \Pr[\hybi{0} =1]/Q $.
\item[$\hybi{2}$:]
            This game is the same as $\hybi{1}$ except for the way $\{ \sfabe.\ct_{i,b} \}_{i,b}$ is generated. 
            Namely, we generate $\sfabe.\ct_{i,b}$ as 
            $\sfabe.\ct_{i,b} \la \ABE.\Enc(\sfabe.\pk_{i,b}, \lab_{i, \skl.\ek^{(j^*)}[i]})$
            for $i\in [\sklpklen]$ and $b\in \bit$.
            
            We observe that the labels being encrypted are changed only for positions of the form $(i, 1\oplus \lab_{i, \skl.\ek^{(j^*)}[i]} )$.
            The adversary $\qA$ cannot notice the change since it is not given any secret key that can decrypt the ABE ciphertexts for these positions. 
            To check this, recall that there is at most one index $j^*$ such that $R(x^*,y^{(j^*)})=1$ and for the corresponding key query, the adversary is given ABE secret keys for positions of the form $(i, \lab_{i, \skl.\ek^{(j^*)}[i]} )$, but not for $(i, \lab_{i, 1\oplus\skl.\ek^{(j^*)}[i]} )$. 
            Hence, we obtain $\abs{\Pr[\hybi{1}=1] - \Pr[\hybi{2} =1]} = \negl(\secp)$ by the adaptive security of $\ABE$.
            See \cref{lem:abeskl_abe_step} for the detail.
\item[$\hybi{3}$:] 
            This game is the same as $\hybi{2}$ except for the way $\ct^*$ is generated. 
            In particular, to generate $\ct^*$, we first run 
            $(\{ \lab_{i} \}_{i\in [\sklpklen]}, \tildeE ) \gets \SimGC(1^\secp, \SKL.\Enc(\skl.\ek^{(j^*)}, m_0; \sfR) )$
            and then compute
            $\sfabe.\ct_{i,b} \la \ABE.\Enc(\sfabe.\pk_{i,b}, \lab_{i})$
            for $i\in [\sklpklen]$ and $b\in \bit$.
            
            We claim that this game is indistinguishable from the previous one.
            To see this, it suffices to show that 
            $(\{ \lab_{i, \skl.\ek^{(j^*)}[i]} \}_i, \tildeE)$ 
            computed by $(\{ \lab_{i,b} \}_{i,b}, \tildeE ) \gets \Garble(1^\secp, E[m_0,\sfR])$
            and 
            $(\{ \lab_{i} \}_{i\in [\sklpklen]}, \tildeE )$
            computed by
            $(\{ \lab_{i} \}_{i\in [\sklpklen]}, \tildeE ) \gets \SimGC(1^\secp, \SKL.\Enc(\skl.\ek^{(j^*)}, m_0; \sfR) )$
            are computationally indistinguishable.
            This immediately follows from the security of the garbled circuit, since we have 
            \[
            E[m_0,\sfR](\skl.\ek^{(j^*)}) 
            = \SKL.\Enc(\skl.\ek^{(j^*)}, m_0; \sfR)
            \]
            by the definition of $E$.
            Hence, we obtain $\abs{\Pr[\hybi{2}=1] - \Pr[\hybi{3} =1]} = \negl(\secp)$.

\item[$\hybi{4}$:] 
            This game is the same as $\hybi{3}$ except that the challenger
            chooses $(\{\lab_i\}_i,\tildeE)$ by $(\{ \lab_{i} \}_{i\in [\sklpklen]}, \tildeE ) \gets \SimGC(1^\secp, \SKL.\Enc(\skl.\ek^{(j^*)}, m_1; \sfR) )$
            instead of 
            $(\{ \lab_{i} \}_{i\in [\sklpklen]}, \tildeE ) \gets \SimGC(1^\secp, \SKL.\Enc(\skl.\ek^{(j^*)}, m_0; \sfR) )$.
            
           To show that $\abs{\Pr[\hybi{3}=1]-\Pr[\hybi{4} =1]}=\negl(\secp)$, it suffices to show that $\SKL.\Enc(\skl.\ek^{(j^*)}, m_0;\sfR) $ is indistinguishable from $\SKL.\Enc(\skl.\ek^{(j^*)}, m_1; \sfR)$ for $\qA$, if it makes $\Oracle{\qVrfy}$ output $\top$ on input $(y^{(j^*)},\qusk')$ for some $\qusk'$
           before the challenge ciphertext is given to $\qA$. 
           The indistinguishability follows from the security of $\SKL$, 
           since the fact that $\qA$ passes the verification $\Oracle{\qVrfy}$ implies that $\qA$ submitted $\skl.\qdk^\prime$ such that $\SKL.\qVrfy(\skl.\vk^{(j^*)},\skl.\qdk^\prime)=\top$
           before it is given the challenge ciphertext and therefore it has no longer the ability to decrypt the ciphertext.
           To turn this intuition into a formal reduction, we have to embed the public key of $\SKL$ into the answer to the $j^*$-th key generation query. 
           Since the reduction algorithm does not know $j^*$ until $\qA$ submits $(x^*,m_0,m_1)$, it can only guess it.
           The change in $\hybi{1}$ is introduced in order to incorporate the guess into the game so that the reduction is possible. 
           We refer to \cref{lem:abeskl_skl_pke_step} for the formal proof for $\abs{\Pr[\hybi{3}=1] - \Pr[\hybi{4} =1]} = \negl(\secp)$.
\item[$\hybi{5}$:] 
This is the same as $\expc{\oneABE,\qA,1}{ada}{ind}{kla}(1^\secp,1)$.
\end{description}
From the above discussion, we have
\begin{eqnarray}\label{Eq:abe_skl_semifinal_step}
   \abs{ Q \Pr[\hybi{3}=1] -  \Pr[\hybi{0}=1]}
    &=& Q \abs{ \Pr[\hybi{3}=1] - \Pr[\hybi{1}=1]} \nonumber
    \\
    &\leq & Q \sum_{i\in [0,2] }\abs{\Pr[\hybi{i+1}=1] - \Pr[\hybi{i}=1]} 
    \leq \negl(\secp).
\end{eqnarray}
We then observe that $\hybi{5}$ (resp., $\hybi{4}$) is the same as $\hybi{0}$ (resp., $\hybi{3}$) except that $m_1$ is used for the encryption instead of $m_0$. Therefore, we obtain 
$
     \abs{ Q \Pr[\hybi{4}=1] -  \Pr[\hybi{5}=1]}
    \leq \negl(\secp)
$
analogously to Eq.~\eqref{Eq:abe_skl_semifinal_step} by considering similar sequence of the games with $m_0$ being replaced by $m_1$ in reverse order.
We therefore have
\begin{eqnarray*}
   &&\abs{\expc{\oneABE,\qA,1}{ada}{ind}{kla}(1^\secp,0)-\expc{\oneABE,\qA,1}{ada}{ind}{kla}(1^\secp,1)}
   \\
   &&= \abs{ \Pr[\hybi{0}=1] -  \Pr[\hybi{5}=1]}
   \\
   && \leq \abs{ \Pr[\hybi{0}=1] -  Q \Pr[\hybi{3}=1]} +
   Q \cdot \abs{ \Pr[\hybi{3}=1] - \Pr[\hybi{4}=1]} +
   \abs{ \Pr[\hybi{5}=1] -  Q \Pr[\hybi{4}=1]} 
   \\
   && \leq \negl(\secp)
\end{eqnarray*}
as desired. 
It remains to prove \cref{lem:abeskl_abe_step,lem:abeskl_skl_pke_step}. 

\begin{lemma}\label{lem:abeskl_abe_step}
$\abs{\Pr[\hybi{1}=1] - \Pr[\hybi{2} =1]} = \negl(\secp)$ if $\ABE$ is adaptively secure.
\end{lemma}
\begin{proof}
This can be reduced to the adaptive security of $\ABE$ by a standard hybrid argument where we modify the way of generating $\ABE.\ct_{i,1\oplus \skl.\ek^{(j^*)}[i]}$ for each $i\in[\sklpklen]$ one by one. 
More precisely, the reduction works as follows.

We define additional hybrids $\hybi{1.k}$ for $k\in [\sklpklen]$ as follows.

\begin{description}
\item[$\hybi{1.k}$:] This is identical to $\hybi{1}$ except that $\sfabe.\ct_{i,1\oplus \skl.\ek^{(j^*)}[i]}$ is generated as 
\begin{equation}
\label{eq:ABECT}
\sfabe.\ct_{i,1\oplus \skl.\ek^{(j^*)}[i]}\gets
\begin{cases}
\ABE.\Enc(\sfabe.\pk_{i,1\oplus \skl.\ek^{(j^*)}[i]},\lab_{i,\skl.\ek^{(j^*)}[i]}) & i<k
\\
\ABE.\Enc(\sfabe.\pk_{i,1\oplus \skl.\ek^{(j^*)}[i]},\lab_{i,1\oplus \skl.\ek^{(j^*)}[i]}) & i\ge k
\end{cases}
\end{equation}
for $i\in [\secp]$. 
\end{description}
Clearly, we have $\hybi{1}=\hybi{1.1}$ and $\hybi{2}=\hybi{1.\sklpklen+1}$. Thus, it suffices to prove that 
$
\left|\Pr[\hybi{1.k+1}=1]-\Pr[\hybi{1.k}=1]\right|=\negl(\secp) 
$
for all $k\in [\sklpklen]$.
Remark that the only difference between $\hybi{1.k+1}$ and $\hybi{1.k}$ is the way of generating $\sfabe.\ct_{k,1\oplus \skl.\ek^{(j^*)}[i]}$.
To show that $\left|\Pr[\hybi{1.k+1}=1]-\Pr[\hybi{1.k}=1]\right|=\negl(\secp)$, 
we construct $\qB$ against adaptive security of $\ABE$ as follows. 

\begin{description}
\item[$\qB(\sfabe.\pk)$:] It works as follows.
\begin{enumerate}
        \item 
        It chooses $\tilde{j} \gets [Q]$ and $(\skl.\ek^{(j)},\skl.\qdk^{(j)},\skl.\vk^{(j)})\la\SKL.\qKG(1^\secp)$
        for $j\in [Q]$.
        
        \item  Generate $(\sfabe.\pk_{i,b},\sfabe.\msk_{i,b})\gets \ABE.\KG(1^\secp)$ for $(i,b)\in [\sklpklen]\times \bit \backslash \{(k, 1\oplus \skl.\ek^{(\tilde{j})}[k] )\}$. 
        Set  
        $\sfabe.\pk_{k, 1\oplus \skl.\ek^{(\tilde{j})}[k] } \seteq \sfabe.\pk$
        and send $\pk\seteq \{ \sfabe.\pk_{i,b}\}_{i,b}$ to $\qA$.
        
        \item 
        $\qB$ initializes the list $\List{\qKG}$ to be an empty set and 
        simulates the following oracles for $\qA$.
            \begin{description}
            \item[$\Oracle{\qKG}(y^{(j)})$:] Given the $j$-th query $y^{(j)}$ with $j\in [Q]$, if there is an entry of the form $(y^{(j)},\vk,V)$, it outputs $\bot$. Otherwise, it generates $\sfabe.\sk^{(j)}_{i} \la \ABE.\KG(\sfabe.\msk_{i,\skl.\ek^{(j)}[i]}, y^{(j)})$ for $i\in [\sklpklen]\backslash\{k\}$. 
            To simulate $\sfabe.\sk^{(j)}_{k}$, $\qB$ preceeds as follows.
            If $\skl.\ek^{(j)}[k] = 1\oplus \skl.\ek^{(\tilde{j})}[k]$, 
            it queries $y^{(j)}$ to its challenger.
            The challenger runs 
            \[
            \sfabe.\sk \gets \ABE.\KG(\sfabe.\msk,y^{(j)})
            \]
            and returns it to $\qB$.
            $\qB$ then sets $\sfabe.\sk^{(j)}_k \seteq \sfabe.\sk$.
            Otherwise (i.e., if $\skl.\ek^{(j)}[k] =  \skl.\ek^{(\tilde{j})}[k]$), it runs $\sfabe.\sk^{(j)}_{k} \la \ABE.\KG(\sfabe.\msk_{k,\skl.\ek^{(j)}[k]}, y^{(j)})$.
            It then sends $\qusk^{(j)} \seteq (\{ \sfabe.\sk^{(j)}_{i} \}_i ,\skl.\qdk^{(j)})$ to $\qA$ and adds $(y^{(j)},\skl.\vk^{(j)},\bot)$ to $\List{\qKG}$.
            
            \item[$\Oracle{\qVrfy}(y,\qusk^\prime)$:] Given $(y,\qusk^\prime)$, it finds an entry $(y,\vk,V)$ from $\List{\qKG}$ and parse $\qusk^\prime = (\{\sfabe.\sk^\prime_i\}_i,\skl.\qusk^\prime)$. (If there is no such entry, it returns $\bot$.) 
            It then parses $\vk=\skl.\vk$
            and returns $\decision \seteq \SKL.\qVrfy(\skl.\vk,\skl.\qdk^\prime)$ to $\qA$. It finally updates the entry into $(y,\vk,\decision)$ if $V=\bot$.
            \end{description}
        
         \item 
         When $\qA$ sends $(x^*,m_0, m_1)$ to the challenger, $\qB$ checks whether there are multiple entries $(y,\vk,V)$ in $\List{\qKG}$ such that $R(x^*,y)=1$ or there is an entry $(y,\vk,V)$ in $\List{\qKG}$ with $R(x^*,y)=1$ and $V=\top$.
         If so, $\qB$ aborts the game and outputs $0$ as its guess.
         Otherwise, $\qB$ defines $j^*\in [Q]$ as in $\hybi{1}$.
         It then aborts and outputs $0$ if $j^*\neq \tilde{j}$.
         Otherwise, $\qB$ computes $\ct^*$ as follows.
         It first chooses $\sfR\gets \bit^{\sklrandlen}$ and 
         computes $(\{ \lab_{i,b} \}_{i\in [\sklpklen], b\in \bit }, \tildeE ) \gets \Garble(1^\secp,E[m_0,\sfR])$.
         It then computes 
         $\sfabe.\ct_{i,b}$
         for $(i,b)\in [\sklpklen]\times \bit \backslash \{(k, 1\oplus \skl.\ek^{(j^*)}[k] )\}$
         as in \cref{eq:ABECT}.
         $\qB$ then submits $( \lab_{k,\skl.\ek^{(j^*)}[k]}, \lab_{k, 1\oplus \skl.\ek^{(j^*)}[k]} )$ to its challenger.
         Then, the challenger runs
         \[
         \sfabe.\ct \gets \ABE.\Enc(\sfabe.\pk, \lab_{\overline{\coin}\oplus \skl.\ek^{(j^*)}[k] } )
         \]
         and gives $\sfabe.\ct$ to $\qB$,
         where $\overline{\coin} \in \bit$ is the coin chosen by the challenger.
         Then, $\qB$ sets $\sfabe.\ct_{k, 1\oplus \skl.\ek^{(j^*)}[k]} \seteq \sfabe.\ct$ and gives 
         $\ct^*\seteq ( \{\sfabe.\ct_{i,b}\}_{i,b}, \tildeE )$ to $\qA$.
         
         \item $\qA$ then continues to make queries to $\Oracle{\qKG}(\cdot)$ and $\Oracle{\qVrfy}(\cdot,\cdot)$.
         $\qB$ answers the queries in the same manner as before the challenge query.
         
         \item $\qA$ finally outputs its guess.
         $\qB$ outputs the same bit as its guess.
        \end{enumerate} 
\end{description}
We first argue that $\qB$ does not make any prohibited key query. 
To see this, we first observe that for every key query $y$ that $\qB$ makes, 
there exists $j$ such that $y=y^{(j)}$.
We then observe that $R(x^*,y^{(j)})=0$ for $j\neq j^*$ and $\qB$ does not make a key query for $y^{(j^*)}$ in the above simulation. 

We have 
%
%
\begin{align}
\advb{\ABE,\qB}{ada}{ind}(\secp) 
&=2\left|\Pr[\mbox{ $\qB$ outputs $\overline{\coin}$ }]-\frac{1}{2}\right|\\
&=\left|\Pr[\mbox{ $\qB$ outputs $1$ }|\overline{\coin}=0]-\Pr[\mbox{ $\qB$ outputs $1$ }|\overline{\coin}=1]
\right|\\
&=\left|\Pr[\hybi{1.k+1}=1]-\Pr[\hybi{1.k}=1]\right|
\end{align}
where the probabilities are taken over the randomness used in the respective games. Thus, $\left|\Pr[\hybi{1.k+1}=1]-\Pr[\hybi{1.k}=1]\right|=\negl(\secp)$ by the adaptive security of $\ABE$.
This completes the proof of \Cref{lem:abeskl_abe_step}. 

\end{proof}

\begin{lemma}\label{lem:abeskl_skl_pke_step}
$\abs{\Pr[\hybi{3}=1] - \Pr[\hybi{4} =1]} = \negl(\secp)$ if $\SKL$ is IND-KLA secure.
\end{lemma}
\begin{proof}
This can be reduced to the IND-KLA security security of $\SKL$.  
To do so, we construct an adversary $\qB$ against IND-KLA security of the scheme with advantage $\abs{\Pr[\hybi{3}=1] - \Pr[\hybi{4} =1]}$ as follows.
\begin{description}
\item[$\qB(\skl.\ek,\skl.\qdk)$:] It works as follows.
\begin{enumerate}
        \item 
        It chooses $\tilde{j} \gets [Q]$ and $(\skl.\ek^{(j)},\skl.\qdk^{(j)},\skl.\vk^{(j)})\la\SKL.\qKG(1^\secp)$
        for $j\in [Q]\backslash\{\tilde{j}\}$.
        It then sets $(\skl.\ek^{(\tilde{j})},\skl.\qdk^{(\tilde{j})})\seteq (\skl.\ek,\skl.\qdk)$.
        It then generates $(\sfabe.\pk_{i,b},\sfabe.\msk_{i,b}) \gets \ABE.\Setup(1^\secp)$
            for $i\in [\sklpklen]$ and $b\in \bin$ and sends $\pk\seteq \{\sfabe.\pk_{i,b}\}_{i,b}$ to the adversary $\qA$.
        
        \item 
        $\qB$ initializes the list $\List{\qKG}$ to be an empty set and 
        simulates the following oracles for $\qA$.
            \begin{description}
            \item[$\Oracle{\qKG}(y^{(j)})$:] Given the $j$-th query $y^{(j)}$ with $j\in [Q]$, if there is an entry of the form $(y^{(j)},\vk,V )$, it outputs $\bot$. Otherwise, it generates $\sfabe.\sk^{(j)}_{i} \la \ABE.\KG(\sfabe.\msk_{i,\skl.\ek^{(j)}[i]}, y^{(j)})$ for $i\in [\sklpklen]$.
            It then returns $\qusk^{(j)} \seteq (\{ \sfabe.\sk^{(j)}_{i} \}_i ,\skl.\qdk^{(j)})$ to $\qA$ and adds $(y^{(j)},\skl.\vk^{(j)},\bot)$ to $\List{\qKG}$.
            
            \item[$\Oracle{\qVrfy}(y,\qusk^\prime)$:] Given $(y,\qusk^\prime)$, it finds an entry $(y,\vk,V)$ from $\List{\qKG}$ and parses $\qusk^\prime = (\{\sfabe.\sk^\prime_i\}_i,\skl.\qdk^\prime)$. (If there is no such entry, it returns $\bot$.) 
            If $y=y^{(j)}$ for $j\neq \tilde{j}$, $\qB$ returns $\decision \seteq \SKL.\qVrfy(\skl.\vk^{(j)}, \skl.\qdk^\prime)$ to $\qA$. 
            Otherwise (i.e., if $y=y^{(\tilde{j})}$), $\qB$ submits $\skl.\qdk^\prime$ to its verification oracle.
            Then, 
            \[ \decision \seteq \SKL.\qVrfy(\skl.\vk,\skl.\qdk^\prime) \]
            is computed and returned to $\qB$.
            $\qB$ then returns $\decision$ to $\qA$.
            It finally updates the entry into $(y,\vk,\decision)$ if $V=\bot$.
            
            \end{description}
        \item 
        When $\qA$ sends $(x^*,m_0, m_1)$ to the challenger, 
        $\qB$ checks whether there are multiple entries $(y,\vk,V)$ in $\List{\qKG}$ such that $R(x^*,y)=1$ or there is an entry $(y,\vk,V)$ in $\List{\qKG}$ with $R(x^*,y)=1$ and $V=\top$.
         If so, $\qB$ aborts the game and outputs $0$ as its guess.
         Otherwise, $\qB$ defines $j^*\in [Q]$ as in $\hybi{1}$.
         It then aborts and outputs $0$ if $j^*\neq \tilde{j}$.
         Otherwise, $\qB$ computes $\ct^*$ as follows.
         It first submits $(m_0,m_1)$ to its challenger.
         Then, the challenger runs 
         \[
         \skl.\ct \la \SKL.\Enc(\skl.\ek, m_{\overline{\coin}})
         \]
         and returns it to $\qB$, where $\overline{\coin}\in \bit$ is the coin chosen by the challenger.
         $\qB$ then runs $(\{ \lab_{i} \}_{i\in [\sklpklen]}, \tildeE ) \gets \SimGC(1^\secp, \skl.\ct )$ and computes
        $\sfabe.\ct_{i,b} \la \ABE.\Enc(\sfabe.\pk_{i,b}, \lab_{i})$
        for $i\in [\sklpklen]$ and $b\in \bit$.
         Then, $\qB$ sets $\ct^*\seteq ( \{\sfabe.\ct_{i,b}\}_{i,b}, \tildeE )$ and gives it to $\qA$.
         
         \item $\qA$ then continues to make queries to $\Oracle{\qKG}(\cdot)$ and $\Oracle{\qVrfy}(\cdot,\cdot)$.
         $\qB$ answers the queries in the same manner as before the challenge query.
         
         \item $\qA$ finally outputs its guess.
         $\qB$ outputs the same bit as its guess.
\end{enumerate}
\end{description}
We then have 
\begin{align}
\advb{\SKL,\qB}{ind}{kla}(\secp) 
&=\left|\Pr[\mbox{ $\qB$ outputs $1$ }|\overline{\coin}=0]-\Pr[\mbox{ $\qB$ outputs $1$ }|\overline{\coin}=1]
\right|\\
&=\left|\Pr[\hybi{3}=1]-\Pr[\hybi{4}=1]\right|
\end{align}
where the probabilities are taken over the randomness used in the respective games. Thus, $\left|\Pr[\hybi{3}=1]-\Pr[\hybi{4}=1]\right|=\negl(\secp)$ by the security of $\SKL$.
This completes the proof of \Cref{lem:abeskl_skl_pke_step}. 
\end{proof}

This completes the proof of \cref{thm:ABESKL_from_ABE} for the case of adaptive security. 

\paragraph{The proof for selective security.}
The statement for selective security can be obtained immediately by considering the same sequence of games as adaptive security case with natural adaptations. In particular, we modify the reduction algorithm in \cref{lem:abeskl_abe_step} so that it outputs $x^*$ at the beginning of the game right after given $x^*$ from $\qA$.

An alternative option is to consider a simpler proof that is tailored to selective setting. This is possible because the proof obtained by adapting the adaptive setting to the selective setting includes a redundant step. 
In particular, we consider a sequence of games without $\hybi{1}$. The reason why $\hybi{1}$ is not necessary is that in the selective setting, the reduction algorithm obtains $x^*$ at the beginning of the game and can use this information throughout the game. 
In particular, whenever $\qA$ makes a key query $y^{(j)}$, the reduction algorithm can check whether $j^*=j$ holds or not by computing the value of $R(x^*,y^{(j)})$ and there is no need to guess it.
By introducing this change, we can improve the reduction cost to be independent of $Q$.
\end{proof}

\subsection{$Q$-Bounded Distinguishing Key Construction}
\label{sec:ABE_SKL_q_bounded}
We construct an ABE-SKL scheme $\qABESKL =(\Setup,\qKG,\Enc,\qDec,\qVrfy)$ for relation $R:\cX\times \cY \to \bit$ with $q$-bounded distinguishing key $\ADAINDKLA$ (resp., $\SELINDKLA$) security
from an ABE-SKL scheme $\oneABESKL =\oneABESKL.(\Setup,\qKG,\Enc,\qDec,\qVrfy)$ for the same relation $R$ with $1$-bounded distinguishing key $\ADAINDKLA$ (resp., $\SELINDKLA$) security.
We note that the construction here is essentially the same as \cite{PKC:ISVWY17}, which converts a single collusion secure ABE scheme into a $q$-bounded collusion secure ABE.
However, our proof is more complex reflecting the fact that the adversary is allowed to make unbounded number of key queries (though the number of distinguishing keys is bounded).

The following construction uses parameters $v\seteq v(\secp)$ and $w\seteq w(\secp)$.
We will set the parameters in \cref{thm:qABESKL_from_1ABE}. 

\begin{description}
\item[$\Setup(1^\secp)$:] $ $
\begin{itemize}
    \item 
    For $i\in [v]$ and $j\in [w]$, run $(\oneabeskl.\pk_{i,j},\oneabeskl.\msk_{i,j}) \gets \oneABESKL.\Setup(1^\secp)$.
    \item Output $(\pk,\msk)\seteq (\{ \oneabeskl.\pk_{i,j} \}_{i\in [v], j\in [w] },\{ \oneabeskl.\msk_{i,j}\}_{i\in [v], j\in [w] })$.
\end{itemize}
\item[$\qKG(\msk,y)$:] $ $
\begin{itemize}
	\item For $i\in [v]$, choose $j_i \gets [w]$. 
	\item Run $(\oneabe.\vk_i, \oneabe.\qusk_{i }) \gets \oneABE.\qKG(\oneabe.\msk_{i,j_i},y)$
	for $i\in [v]$.
	\item Output $\qusk \seteq \{ j_i, \oneabe.\qusk_i \}_{i\in [v] } $ and $\vk \seteq \{ \oneabe.\vk_i \}_{i\in [v] } $.
\end{itemize}
\item[$\Enc(\pk,x,m)$:] $ $
\begin{itemize}
\item Choose $\mu_1,\ldots, \mu_{v-1} \chosen \bin^{\ell}$ and set 
$\mu_v \seteq (\oplus_{i\in [v-1]}\mu_i) \oplus m$, where $\oplus$ denotes bit-wise XOR here.
\item Run $\oneabe.\ct_{i,j} \la \oneABE.\Enc(\oneabe.\pk_{i,j}, x, \mu_i  )$
for $i\in [v]$ and $j \in [w]$.
\item Output $\ct \seteq \{ \oneabe.\ct_{i,j} \}_{i\in [v], j\in [w] }$.
\end{itemize}
\item[$\qDec(\qusk,x,\ct)$:] $ $
\begin{itemize}
\item Parse $\qusk \seteq \{ j_i, \oneabe.\qusk_i \}_{i\in [v] }$ and $\ct \seteq \{ \oneabe.\ct_{i,j} \}_{i\in [v], j\in [w] }$.
\item Compute $\mu^\prime_{i} \gets \oneABE.\Dec(\oneabe.\qusk_{i},x,\oneabe.\ct_{i,j_i})$ for $i\in [v]$.
\item Compute and output $m' \seteq \oplus_{i\in [v]}\mu^\prime_i$.
\end{itemize}
\item[$\qVrfy(\vk,\qusk^\prime)$:] $ $
\begin{itemize}
\item Parse $\vk = \{ \oneabe.\vk_i \}_{i\in [v] } $ and $\qusk' \seteq \{ j_i, \oneabe.\qusk'_i \}_{i\in [v] } $.
\item Compute $\decision_i\la \oneABE.\qVrfy(\oneabe.\vk_i,\oneabe.\qusk^\prime_i)$
for $i\in [v]$.
\item If $\decision_i=\top$ for all $i\in [v]$, output $\top$. Otherwise, output $\bot$.
\end{itemize}
\end{description}
It is straightforward to see that the decryption correctness and the verification correctness of the above scheme follow from those of $\oneABE$.

\begin{theorem}\label{thm:qABESKL_from_1ABE}
Assuming $\oneABESKL$ is $1$-bounded distinguishing key $\ADAINDKLA$ (resp., $\SELINDKLA$) secure, $\qABESKL$ is $q$-bounded distinguishing key $\ADAINDKLA$ (resp., $\SELINDKLA$) secure if we set the parameters as follows: 
\begin{itemize}
    \item 
    For the adaptive case, we assume that the size of the ciphertext attribute space $|\cX_\secp|$ is bounded by $2^{n(\secp)}$ for some polynomial function $n(\secp)$.
    We then set $v= 2(\secp + n )$ and $w= q^2$.
    \item For the selective case, we set $v=\secp$ and $w=q^2$.
\end{itemize}
\end{theorem}
\begin{proof}[Proof of~\cref{thm:qABESKL_from_1ABE}]
Here, we first focus on the proof for the case of $q$-bounded distinguishing key $\ADAINDKLA$ and later mention the difference for the case of $q$-bounded distinguishing key $\SELINDKLA$.
We define a sequence of hybrid games.
\begin{description}
\item[$\hybi{0}$:] This is the same as $\expc{\qABE,\qA,q}{ada}{ind}{kla}(1^\secp,0)$. More specifically, it is as follows.
        \begin{enumerate}
            \item The challenger generates $(\oneabe.\pk_{i,j},\oneabe.\msk_{i,j}) \gets \oneABE.\Setup(1^\secp)$
            for $i\in [v]$ and $j \in [w]$ and sends $\pk\seteq \{\oneabe.\pk_{i,j}\}_{i,j}$ to the adversary $\qA$.
            The challenger then initializes the list $\List{\qKG}$ to be an empty set. 
            $\qA$ can access the following oracles.
            \begin{description}
            \item[$\Oracle{\qKG}(y^{(k)})$:] Given the $k$-th query $y^{(k)}$ with $k\in [Q]$, if there is an entry of the form $(y^{(k)},\vk,V)$, it outputs $\bot$. Otherwise, it chooses $j_i^{(k)} \gets [w]$ for $i\in [v]$ and runs $(\oneabe.\vk_i^{(k)}, \oneabe.\qusk_{i }^{(k)}) \gets \oneABE.\qKG(\oneabe.\msk_{i,j_i},y^{(k)})$ for $i\in [v]$. It then returns $\qusk^{(k)} \seteq \{j_i^{(k)}, \oneabe.\qusk_i^{(k)} \}_{i\in [v] } $ and $\vk^{(k)} \seteq \{ \oneabe.\vk_i^{(k)} \}_{i\in [v] } $ to $\qA$ and adds $(y^{(k)},\vk^{(k)},\bot)$ to $\List{\qKG}$.
            \item[$\Oracle{\qVrfy}(y,\qusk^\prime)$:] Given $(y,\qusk^\prime)$, it finds an entry $(y,\vk,V)$ from $\List{\qKG}$ and parses $\qusk^\prime = \{j_i, \qusk_i^\prime\}_i$. (If there is no such entry, it returns $\bot$.) 
            It then computes $\decision_i \seteq \oneABE.\qVrfy(\oneabe.\vk_i,\qusk^\prime_i)$ for $i\in [v]$ and checks if $\decision_i=\top$ for all $i\in [v]$.
            If so, it returns $\decision \seteq \top$ to $\qA$. 
            Otherwise, it returns $\decision \seteq \bot$ to $\qA$.
            It finally updates the entry into $(y,\vk,\decision)$ if $V=\bot$.
            \end{description}
            \item When $\qA$ sends $(x^*,m_0, m_1)$ to the challenger, the challenger computes the set $K_{x^*}\seteq \{ k\in [Q_1]: R(x^*,y^{(k)})=1 \}$, where $Q_1 \leq Q$ is the number of key queries made by $\qA$ so far.
            If we have $V=\top$ for all entries of the form $(y^{(k)}, \vk, V)$ in $\List{\qKG} $ with $k\in K_{x^*}$ and $|K_{x^*}|\leq q$, 
            the challenger chooses $\mu_1,\ldots, \mu_{v-1} \chosen \bin^{\ell}$, sets $\mu_v \seteq (\oplus_{i\in [v-1]}\mu_i) \oplus m_0$, and computes $\oneabe.\ct_{i,j} \la \oneABE.\Enc(\oneabe.\pk_{i,j}, x^*, \mu_i  )$ for $i\in [v]$ and $j\in [w]$.
            It then sends $\ct^* \seteq \{ \oneabe.\ct_{i,j} \}_{i,j}$ to $\qA$. 
            Otherwise (i.e., if $|K_{x^*}|> q$ or if there is an entry of the form $(y^{(k)}, \vk, \bot)$ for some $k\in K_{x^*}$), it aborts the game and outputs $0$. 
            
            \item $\qA$ continues to make queries to $\Oracle{\qKG}(\cdot)$ and  $\Oracle{\qVrfy}(\cdot,\cdot )$.
            However, $\qA$ is not allowed to send a key attribute $y$ such that $R(x^*,y)=1$ to $\Oracle{\qKG}$.
            \item $\qA$ outputs a guess $\coin^\prime$ for $\coin$. The challenger outputs $\coin'$ as the final output of the experiment.
        \end{enumerate}
\item[$\hybi{1}$:] 
This game is the same as $\hybi{0}$ except for the way $\ct^*$ is generated.
In particular, when $\qA$ submits $(x^*,m_0,m_1)$, the challenger aborts the game and outputs $0$ as the outcome of the game if there is no $i^*$ such that $\{ j_{i^*}^{(k)} \}_{k\in K_{x^*}}$ are all distinct. 
Otherwise, the challenger continues the game as specified in $\hybi{0}$ .

We observe that unless there is no such $i^*$, the game is the same as the previous one.
We bound the probability of this occuring. 
Let us first consider the case where $\qA$ fixes its target $x^*$ at the beginning of the game (i.e., selective security setting). 
In this case, by simple probability calculation, we can show that the probability that $i^*$ does not exist is exponentially small in the parameter $v$. 
However, in the adaptive case, the adversary can adaptively choose $x^*$ dependent on the values of $\{j_i^{(k)}\}_{i\in [v], k\in [Q]}$ and the proof for the selective case no longer works. To deal with the added flexibility given to the adversary, we use the 
union bound over all $x\in \cX$ and then use the above bound for each fixed $x$. This requires the parameter $v$ to grow dependent on the size of $\log|\cX_\secp|$ so that the sum of the probabilities is still small enough even after taking the union bound.
Based on the above discussion, we can prove $\abs{ \Pr[\hybi{0}= 1] - \Pr[\hybi{1}= 1] }=\negl(\secp)$. 
We refer to \cref{thm:upperbound_on_probability_of_BAD} for the detail.

\item[$\hybi{2}$:] 
This game is the same as $\hybi{1}$ except that the challenger chooses random $\tilde{i} \gets [v]$ at the beginning of the game. Then, right before the challenger computes $\ct^*$, it checks whether $\tilde{i}= i^*$, where $i^*$ is the smallest index such that $\{ j_{i^*}^{(k)} \}_{k\in K_{x^*}}$ are all distinct. 
\footnote{Note that $i^*$ is not defined until $\qA$ chooses $x^*$.%
}
If so, the challenger continues the game until $\qA$ outputs its guess.
Otherwise, it aborts the game and outputs $0$ as the outcome of the game. 

Since the choice of $\tilde{i}$ is independent from the view of $\qA$ and the outcome of the game is $1$ only when $\tilde{i}=i^*$,
we can easily see that $\Pr[\hybi{2}=1] = \Pr[\hybi{1} =1]/v $.

\item[$\hybi{3}$:] 
This is the same as $\hybi{2}$ except for how $\mu_1,\ldots, \mu_v$ are generated. 
In particular, $\qA$ first chooses $\mu_1,\ldots, \mu_v \la \bit^\ell$ and discards $\mu_{i^*}$. 
It then sets $\mu_{i^*}\seteq (\oplus_{i\in [v]\backslash\{i^*\}})\oplus m_{0}$.
It can be easily seen that the distribution of $\mu_1,\ldots, \mu_v$ is unchanged from the previous game and thus we have $\Pr[\hybi{2}= 1] = \Pr[\hybi{3}= 1] $. 

\item[$\hybi{4}$:] 
This is the same as $\hybi{3}$ except that $\mu_{i^*}$ is set as $\mu_{i^*}\seteq (\oplus_{i\in [v]\backslash\{i^*\}}\mu_i)\oplus m_{1}$.

We claim that this change is not noticed by $\qA$ by the security of the underlying $\oneABE$. 
To show this, we first observe that the game differs from the previous one only in how $\{ \oneabe.\ct_{i^*,j} \}_{j\in[w]}$ are generated.
We then change each plaintext encrypted in $\{ \oneabe.\ct_{i^*,j}\}_{j}$ one by one by using the security of the underlying $\oneABE$.
This is possible since for each $\oneABE$ instance with index $(i^*,j)$,
$\qA$ is given only at most one distinguishing key by the change we introduced in $\hybi{1}$ and thus we can use the security of $\oneABE$ for such instances.
We therefore have $\abs{\Pr[\hybi{3}= 1] - \Pr[\hybi{4}= 1]} = \negl(\secp)$.
We refer to \cref{thm:qABE_from_1ABE_1ABE_step} for the detail. 

\item[$\hybi{5}$:] 
This is the same as $\expc{\qABE,\qA,q}{ada}{ind}{kla}(1^\secp,1)$.
\end{description}
From the above discussion, we have
\begin{eqnarray}\label{Eq:qabe_skl_semifinal_step}
   \abs{ v \Pr[\hybi{3}=1] -  \Pr[\hybi{0}=1]}
    = \abs{ \Pr[\hybi{1}=1] - \Pr[\hybi{0}=1] 
    }
    \leq   \negl(\secp).
\end{eqnarray}
We then observe that $\hybi{5}$ (resp., $\hybi{4}$) is the same as $\hybi{0}$ (resp., $\hybi{3}$) except that $m_1$ is used for the encryption instead of $m_0$. Therefore, we obtain 
$
     \abs{ v \Pr[\hybi{4}=1] -  \Pr[\hybi{5}=1]}
    \leq \negl(\secp)
$
analogously to Eq.~\eqref{Eq:qabe_skl_semifinal_step} by considering similar sequence of games with $m_0$ being replaced by $m_1$ in a reverse order.
We therefore have
\begin{eqnarray*}
   &&\abs{\expc{\qABE,\qA,q}{ada}{ind}{kla}(1^\secp,0)-\expc{\qABE,\qA, q}{ada}{ind}{kla}(1^\secp,1)}
   \\
   &&= \abs{ \Pr[\hybi{0}=1] -  \Pr[\hybi{5}=1]}
   \\
   && \leq \abs{ \Pr[\hybi{0}=1] -  v \Pr[\hybi{3}=1]} +
   v \cdot \abs{ \Pr[\hybi{3}=1] - \Pr[\hybi{4}=1]} +
   \abs{ \Pr[\hybi{5}=1] -  v \Pr[\hybi{4}=1]} 
   \\
   && \leq \negl(\secp)
\end{eqnarray*}
as desired. 
It remains to prove \cref{thm:upperbound_on_probability_of_BAD,thm:qABE_from_1ABE_1ABE_step}.

\begin{lemma}\label{thm:upperbound_on_probability_of_BAD}
$\abs{\Pr[\hybi{0}= 1] - \Pr[\hybi{1}= 1]} = \negl(\secp)$ holds both for selective and adaptive settings. 
\end{lemma}
\begin{proof}
We first show the statement for the selective case.
The proof for this case is the same as \cite[Lemma 1]{PKC:ISVWY17}, but we provide the proof here for completeness.
In the selective case, the probability that $\{j_i^{(k)}\}_{k\in K_{x^*}}$ are not all distinct for some fixed $i$ is 
\[
1- \frac{w(w-1)\cdots (w-q+1)}{w^q} 
\leq 1 - \left( 1-\frac{q-1}{w} \right)^q.
\]
Therefore, the probability that there is no $i^*$ satisfying the requirement is at most
\[\left( 1 - \left( 1-\frac{q-1}{w} \right)^q \right)^v\]
which is negligible when $v=\secp$ and $w=q^2$ since 
\[
\left( 1 - \left( 1-\frac{q-1}{w} \right)^q \right)^v
\leq 
\left( 1 - e^{-1} \right)^\secp = 2^{-O(\secp)}.
\]
We then consider the adaptive case. 
We have 
\begin{eqnarray*}
\Pr[ \mbox{$\{j_i^{(k)}\}_{k\in K_{x^*}}$ are not all distinct} ]
& = & \sum_{x\in \cX_\secp}\Pr[ x^*=x \land \mbox{$\{j_i^{(k)}\}_{k\in K_{x}}$ are not all distinct} ] 
\\
&\leq & \sum_{x\in \cX_\secp}\Pr[  \mbox{$\{j_i^{(k)}\}_{k\in K_{x}}$ are not all distinct} ] 
\\
&\leq & \sum_{x\in \cX_\secp}  \left( 1 - \left( 1-\frac{q-1}{w} \right)^q \right)^v
\\
&\leq & |\cX_\secp| \left( 1 - e^{-1} \right)^v
\\
&\leq & 2^{-\secp},
\end{eqnarray*}
where the probabilities are taken over all randomness used in the game.
In the above, third line follows from the same analysis as the selective case and the forth and the fifth lines follow from our parameter setting.
\end{proof}

\begin{lemma}\label{thm:qABE_from_1ABE_1ABE_step}
If $\oneABE$ is $1$-bounded distinguishing key $\ADAINDKLA$, $\abs{\Pr[\hybi{3}= 1] - \Pr[\hybi{4}= 1]} = \negl(\secp)$. 
\end{lemma}
\begin{proof}
This can be reduced to the $1$-bounded distinguishing key $\ADAINDKLA$ security of $\oneABE$ by a standard hybrid argument, where we modify the plaintext encrypted in $\oneabe.\ct_{i^*,j}$ for each $j\in[w]$ one by one. 
More precisely, the reduction works as follows.

We define additional hybrids $\hybi{3.k}$ for $k\in [w]$ as follows.
In the following, let $\xi_b \seteq (\oplus_{i\in[v]\backslash\{i^*\}}\mu_i)\oplus m_b$ for $b\in \bit$.
\begin{description}
\item[$\hybi{3.\tau}$:] This is identical to $\hybi{3}$ except that $\oneabe.\ct_{i^*,j}$ is generated as 
\begin{equation}
\label{eq:qABECT}
\oneabe.\ct_{i^*,j}\gets
\begin{cases}
\oneABE.\Enc(\oneabe.\pk_{i^*,j}, \xi_1 ) & j<\tau
\\
\oneABE.\Enc(\oneabe.\pk_{i^*,j}, \xi_0 )  & j\ge \tau
\end{cases}
\end{equation}
for $j\in [\secp]$. 
\end{description}
Clearly, we have $\hybi{3}=\hybi{3.1}$ and $\hybi{4}=\hybi{3.w+1}$. Thus, it suffices to prove that 
$
\left|\Pr[\hybi{3.\tau+1}=1]-\Pr[\hybi{3.\tau}=1]\right|=\negl(\secp) 
$
for all $\tau\in [w]$.
Remark that the only difference between $\hybi{3.\tau+1}$ and $\hybi{3.\tau}$ is the way of generating $\oneabe.\ct_{i^*,\tau}$.
To show that $\left|\Pr[\hybi{3.\tau+1}=1]-\Pr[\hybi{3.\tau}=1]\right|=\negl(\secp)$, 
we construct $\qB$ against the security of $\oneABE$ as follows. 

\begin{description}
\item[$\qB(\oneabe.\pk)$:] It works as follows.
\begin{enumerate}
        \item It first chooses random $\tilde{i} \gets [v]$.
        \item The challenger generates $(\oneabe.\pk_{i,j},\oneabe.\msk_{i,j}) \gets \oneABE.\Setup(1^\secp)$
            for $(i,j) \in ([v]\times [w]) \backslash \{ (\tilde{i},k )\}$.
            It then sets $\oneabe.\pk_{\tilde{i},k} \seteq \oneabe.\pk$
            and sends $\pk\seteq \{\oneabe.\pk_{i,j}\}_{i,j}$ to the adversary $\qA$.
            It then initializes the list $\List{\qKG}$ to be an empty set. 
            $\qB$ then simulates the following oracles for $\qA$.
            \begin{description}
            \item[$\Oracle{\qKG}(y^{(k)})$:] Given the $k$-th query $y^{(k)}$ with $k\in [Q]$ from $\qA$, $\qB$ returns $\bot$ to $\qA$ if there is an entry of the form $(y^{(k)},\vk,V)$. Otherwise, it chooses $j_i^{(k)} \gets [w]$ for $i\in [v]$ and runs $(\oneabe.\vk_i^{(k)}, \oneabe.\qusk_{i }^{(k)}) \gets \oneABE.\qKG(\oneABE.\msk_{i,j_i},y^{(k)})$ for $i\in [v]\backslash\{ \tilde{i}\}$. 
            If $j_{\tilde{i}}^{(k)}=\tau$, 
            it sends $y^{(k)}$ to its key generation oracle and is given
            \[
            \oneabe.\qusk \gets \oneABE.\qKG( \oneabe.\msk, y^{(k)} ).
            \]
            Then, it sets $\oneabe.\qusk^{(k)}_{\tilde{i}}\seteq \oneabe.\qusk $.
            Otherwise (i.e., if $j_{\tilde{i}}^{(k)}\neq \tau$), it runs 
            $(\oneabe.\vk_{\tilde{i}}^{(k)}, \oneabe.\qusk_{\tilde{i} }^{(k)}) \gets \ABE.\qKG(\oneABE.\msk_{{\tilde{i}},j_{\tilde{i}}},y^{(k)})$ by itself.
            Finally, $\qB$ returns $\qusk^{(k)} \seteq \{j_i^{(k)}, \oneabe.\qusk_i^{(k)} \}_{i\in [v] } $ to $\qA$ and adds $(y^{(k)},\vk^{(k)},\bot)$ to $\List{\qKG}$, where $\vk^{(k)} \seteq \{ \oneabe.\vk_i^{(k)} \}_{i\in [v] } $.

            \item[$\Oracle{\qVrfy}(y,\qusk^\prime)$:] Given $(y,\qusk^\prime)$, it finds an entry $(y,\vk,V)$ from $\List{\qKG}$ and parses $\qusk^\prime = \{j_i, \qusk_i^\prime\}_i$. (If there is no such entry, it returns $\bot$.) 
            It then computes $\decision_i \seteq \oneABE.\qVrfy(\oneabe.\vk_i,\qusk^\prime_i)$ for $i\in [v]$.
            If $j_{\tilde{i}}=\tau$, $\qB$ makes a query to its own verification oracle to obtain 
            \[
            \decision_{\tilde{i}} \seteq \oneABE.\qVrfy(\oneabe.\vk,\qusk^\prime_{\tilde{i}}).
            \]
            Otherwise, $\qB$ runs $\decision_{\tilde{i}} \seteq \oneABE.\qVrfy(\oneabe.\vk_{\tilde{i},\tau},\qusk^\prime_{\tilde{i}})$ by itself.
            Finally, it checks if $\decision_{i}=\top$ for all $i\in [v]$.
            If so, it returns $\decision \seteq \top$ to $\qA$. 
            Otherwise, it returns $\decision \seteq \bot$ to $\qA$.
            It finally updates the entry into $(y,\vk,\decision)$ if $V=\bot$.
            \end{description}
            \item When $\qA$ sends $(x^*,m_0, m_1)$ to the challenger, 
            $\qB$ aborts and outputs $0$ 
            if either $|K_{x^*}|> q$ or there is an entry of the form $(y^{(k)}, \vk, \bot)$ for some $k\in K_{x^*}$.
            It also aborts and outputs $0$ if $i^*\neq \tilde{i}$, which includes the case that there is no $i^*$ satisfying the properties we defined in $\hybi{1}$.
            Otherwise, it chooses $\mu_1,\ldots, \mu_v \la \bit^\ell$ and sets 
            $\xi_0 \seteq (\oplus_{i\in [v]\backslash\{i^*\}}\mu_i)\oplus m_{0}$ and $\xi_1 \seteq (\oplus_{i\in [v]\backslash\{i^*\}}\mu_i)\oplus m_{1}$.
            It then computes $\oneabe.\ct_{i,j} \la \oneABE.\Enc(\oneabe.\pk_{i,j}, x^*, \mu_i  )$ for $i\in [v]\backslash\{i^*\}$ and $j\in [w]$
            and
            $\oneabe.\ct_{i^*,j}$ for $j\in [w]\backslash \{\tau\}$ as 
            \cref{eq:qABECT}. 
            It then submits $(\xi_0,\xi_1)$ to its challenger.
            Then,
            \[
            \oneabe.\ct \gets \Enc(\oneabe.\pk, \xi_{\overline{\coin}} ) 
            \]
            is run and $\oneabe.\ct$ is returned to $\qB$, where $\overline{\coin}$ is the random bit chosen by $\qB$'s challenger.
            Finally, $\qB$ sets $\oneabe.\ct_{i^*,\tau} \seteq \oneabe.\ct$
            and sends $\ct^* \seteq \{ \oneabe.\ct_{i,j} \}_{i,j}$ to $\qA$. 
            
            \item $\qA$ continues to make queries to $\Oracle{\qKG}(\cdot)$ and  $\Oracle{\qVrfy}(\cdot,\cdot )$.
            However, $\qA$ is not allowed to send a key attribute $y$ such that $R(x^*,y)=1$ to $\Oracle{\qKG}$.
            \item $\qA$ outputs a guess $\coin^\prime$ for $\coin$. The challenger outputs $\coin'$ as the final output of the experiment.
        \end{enumerate} 
\end{description}
We first argue that $\qB$ does not make more than two distinguishing key queries.
This is because $\qB$ aborts and outputs $0$ before it makes a challenge query if there is no $i^*$ with the required conditions. 
For such $i^*$, we have that $\{ j^{(k)}_{i^*} \}_{k\in K_{x^*}}$ are all distinct and thus in particular, $\qB$ needs to simulate only single distinguishing key for the  $(i^*,\tau)$-th instance, to which the reduction algorithm embeds the $\oneABE$ instance.

We then have 
\begin{align}
\advc{\oneABE,\qB,1}{ada}{ind}{kla}(\secp) 
&=\left|\Pr[\mbox{ $\qB$ outputs $1$ }|\overline{\coin}=0]-\Pr[\mbox{ $\qB$ outputs $1$ }|\overline{\coin}=1]
\right|\\
&=\left|\Pr[\hybi{3}=1]-\Pr[\hybi{4}=1]\right|
\end{align}
where the probabilities are taken over the randomness used in the respective games. Thus, $\left|\Pr[\hybi{3}=1]-\Pr[\hybi{4}=1]\right|=\negl(\secp)$ by the adaptive security of $\oneABE$.
This completes the proof of \Cref{thm:qABE_from_1ABE_1ABE_step}.

\end{proof}
This completes the proof of \cref{thm:qABESKL_from_1ABE} for the case of adaptive security.

\paragraph{The proof for selective security.}
The proof for selective security can be obtained immediately by considering the same sequence of games as adaptive security case with natural adaptations. 
There are two main differences. 
The proof for \cref{thm:upperbound_on_probability_of_BAD} requires different parameters for selective and adaptive cases. We refer to the proof of the lemma for the detail. 
Another difference is that we modify the reduction algorithm in \cref{thm:qABE_from_1ABE_1ABE_step} so that it outputs $x^*$ at the beginning of the game right after given $x^*$ from $\qA$.
\end{proof}

\subsection{Instantiations}
Here, we explain new schemes that can be obtained by applying the conversions that we showed in \cref{sec:ABE_SKL_one_bounded,sec:ABE_SKL_q_bounded} to existing IBE/ABE schemes. 
Our constructions are fully generic and can upgrade almost all ABE schemes\footnote{
Our conversion in \cref{sec:ABE_SKL_q_bounded} for the adaptive security case poses the restriction that the size of the cipheretxt attribute space of the ABE should be bounded by $2^{\poly(\secp)}$ for some polynomial $\poly(\secp)$. 
This means that we cannot apply the conversion for adaptively secure ABE for DFA for example,
since the ciphertext attribute is of unbounded length and there is no such bound for the size of the ciphertext attribute space.
However, we do not know any concrete ABE scheme from standard assumptions for which we cannot apply our conversion. 
} into the one with the security against key leasing attacks with the help of IND-KLA secure PKE-SKL scheme, which can be instantiated from any (post quantum) PKE. 
Here, we mention some instantiations, all of which are obtained from the standard LWE assumption.  
\begin{itemize}
    \item If we start from selectively secure ABE scheme for circuits \cite{STOC:GorVaiWee13,EC:BGGHNS14} and apply the conversions in \cref{sec:ABE_SKL_one_bounded,sec:ABE_SKL_q_bounded}, 
    we obtain an ABE-SKL scheme for circuits with $q$-bounded distinguishing key $\SELINDKLA$ security for any $q=\poly(\secp)$.
    \item If we start from adaptively secure ABE for inner products over the integer \cite{AC:KNYY20} and apply the conversions in \cref{sec:ABE_SKL_one_bounded,sec:ABE_SKL_q_bounded}, we obtain an ABE-SKL scheme for the same predicate with $q$-bounded distinguishing key $\ADAINDKLA$ security for any $q=\poly(\secp)$.
    We note that the conversion in \Cref{sec:ABE_SKL_q_bounded} for adaptive security case can be applied for the scheme, since the size of the ciphertext attribute space is bounded by $2^{\poly(\secp)}$ for the primitive. 
    Similar implications can be obtained for adaptively secure $t$-CNF formulae for $t=O(1)$~\cite{C:Tsabary19} and fuzzy IBE for small universe \cite{AC:KNYY20}. 
    \item  If we start from adaptively (resp., selectively) secure IBE \cite{EC:AgrBonBoy10,EC:CHKP10} and apply the conversion in \cref{sec:ABE_SKL_one_bounded}, we obtain IBE-SKL scheme with $1$-bounded distinguishing key $\ADAINDKLA$ (resp., $\SELINDKLA$) security. 
    We note that $1$-bounded distinguishing key security for the case of IBE is a more natural security notion than that for the case of ABE with other relations since there is only one attribute that is eligible for decrypting a ciphertext in the case of IBE (i.e., the identity that is associated with the ciphertext), whereas there can be exponentially many such attributes in general.
\end{itemize}


\section{Public-Key Functional Encryption with Secure Key Leasing}\label{sec:PKFE_SKL}

\subsection{Definitions}\label{sec:PKFE_SKL_def}
\begin{definition}[PKFE with Secure Key Leasing]
A PKFE-SKL scheme $\PKFESKL$ is a tuple of six algorithms $(\Setup, \qKG, \Enc, \qDec,\qCert,\certvrfy)$. 
Below, let $\cX$, $\cY$, and $\cF$ be the plaintext, output, and function spaces of $\PKFESKL$, respectively.
\begin{description}
\item[$\Setup(1^\secp)\ra(\pk,\msk)$:] The setup algorithm takes a security parameter $1^\lambda$, and outputs a public key $\pk$ and master secret key $\msk$.
\item[$\qKG(\msk,f)\ra(\qfsk,\vk)$:] The key generation algorithm takes a master secret key $\msk$ and a function $f \in \calF$, and outputs a functional decryption key $\qfsk$ and a verification key $\vk$.

\item[$\Enc(\pk,x)\ra\ct$:] The encryption algorithm takes a public key $\pk$ and a plaintext $x \in \cX$, and outputs a ciphertext $\ct$.

\item[$\qDec(\qfsk,\ct)\ra\tlx$:] The decryption algorithm takes a functional decryption key $\qfsk$ and a ciphertext $\ct$, and outputs a value $\tilde{x}$.


\item[$\qVrfy(\vk,\qfsk^\prime)\ra\top/\bot$:] The verification algorithm takes a verification key $\vk$ and a quantum state $\qfsk^\prime$, and outputs $\top$ or $\bot$.

\item[Decryption correctness:]For every $x \in \cX$ and $f \in \calF$, we have
\begin{align}
\Pr\left[
\qDec(\qfsk, \ct) \allowbreak = f(x)
\ \middle |
\begin{array}{ll}
(\pk,\msk) \la \Setup(1^\secp)\\
(\qfsk,\vk)\gets\qKG(\msk,f)\\
\ct\gets\Enc(\pk,x)
\end{array}
\right] 
=1-\negl(\secp).
\end{align}

\item[Verification correctness:] For every $f \in \calF$, we have 
\begin{align}
\Pr\left[
\qVrfy(\vk,\qfsk)=\top
\ \middle |
\begin{array}{ll}
(\pk,\msk) \la \Setup(1^\secp)\\
(\qfsk,\vk)\gets\qKG(\msk,f)\\
\end{array}
\right] 
=1-\negl(\secp).
\end{align}

\end{description}
\end{definition}
\begin{remark}
Although Kitagawa and Nishimaki~\cite{AC:KitNis22} require SKFE-SKL to have classical certificate generation algorithm for deletion, we do not since it is optional.
If there exists a PKE-SKL scheme that has a classical certificate generation algorithm, our PKFE-SKL scheme in~\cref{sec:CR_PKFE_SKL} also has a classical certificate generation algorithm.
\end{remark}

\begin{definition}[Adaptive Indistinguishability against Key Leasing Attacks]\label{def:ada_ind_kla_PKFESKL}
We say that a PKFE-SKL scheme $\PKFESKL$ for $\Xs,\Ys$, and $\Fs$ is an adaptively indistinguishable secure against key leasing attacks (Ada-IND-KLA), if it satisfies the following requirement, formalized from the experiment $\expc{\qA,\PKFESKL}{ada}{ind}{kla}(1^\secp,\coin)$ between an adversary $\qA$ and a challenger:
        \begin{enumerate}
            \item At the beginning, the challenger runs $(\pk,\msk)\gets\Setup(1^\secp)$.
            Throughout the experiment, $\qA$ can access the following oracles.
            \begin{description}
            \item[$\Oracle{\qKG}(f)$:] Given $f$, it finds an entry $(f,\vk,V)$ from $\List{\qKG}$. If there is such an entry, it returns $\bot$. Otherwise, it generates $(\qfsk,\vk)\la\qKG(\msk,f)$, sends $\qfsk$ to $\qA$, and adds $(f,\vk,\bot)$ to $\List{\qKG}$.
            \item[$\Oracle{\qVrfy}(f,\qfsk^\prime)$:] Given $(f,\qfsk^\prime)$, it finds an entry $(f,\vk,V)$ from $\List{\qKG}$. (If there is no such entry, it returns $\bot$.) It computes $d \gets \qVrfy(\vk,\qfsk^\prime)$ and sends $d$ to $\qA$. If $V = \top$, it does not update $\List{\qKG}$. Else if $V = \bot$, it updates the entry by setting $V\seteq d$.
            \end{description}
            \item When $\qA$ sends $(x_0^*,x_1^*)$ to the challenger, the challenger checks if for any entry $(f,\vk,V)$ in $\List{\qKG}$ such that $f(x_0^*)\ne f(x_1^*)$, it holds that $V=\top$. If so, the challenger generates $\ct^*\la\Enc(\pk,x_\coin^*)$ and sends $\ct^*$ to $\qA$. Otherwise, the challenger outputs $0$. Hereafter, $\qA$ is not allowed to send a function $f$ such that $f(x_0^*)\ne f(x_1^*)$ to $\Oracle{\qKG}$.
            \item $\qA$ outputs a guess $\coin^\prime$ for $\coin$. The challenger outputs $\coin'$ as the final output of the experiment.
        \end{enumerate}
        For any QPT $\qA$, it holds that
\ifnum\llncs=0        
\begin{align}
\advc{\PKFESKL,\qA}{ada}{ind}{kla}(\secp) \seteq \abs{\Pr[\expc{\PKFESKL,\qA}{ada}{ind}{kla} (1^\secp,0) \ra 1] - \Pr[\expc{\PKFESKL,\qA}{ada}{ind}{kla} (1^\secp,1) \ra 1] }\leq \negl(\secp).
\end{align}
\else
\begin{align}
\advc{\PKFESKL,\qA}{ada}{ind}{kla}(\secp) 
&\seteq \abs{\Pr[\expc{\PKFESKL,\qA}{ada}{ind}{kla} (1^\secp,0) \ra 1] - \Pr[\expc{\PKFESKL,\qA}{ada}{ind}{kla} (1^\secp,1) \ra 1] }\\
&\leq \negl(\secp).
\end{align}
\fi
\end{definition}
\begin{remark}
\Cref{def:ada_ind_kla_PKFESKL} assumes that the adversary does not get more than one decryption key for the same $f$ for simplification as~\cref{remark:same_query_remark}.  
\end{remark}

\subsection{Constructions}\label{sec:CR_PKFE_SKL}

We describe our PKFE-SKL scheme in this section.
We construct a PKFE-SKL scheme $\PKFESKL =(\Setup,\qKG,\Enc,\qDec,\qVrfy)$ by using the following building blocks.
\begin{itemize}
\item IND-KLA secure PKE-SKL $\SKL = \SKL.(\qKG,\Enc,\qDec,\qVrfy)$.
\item Adaptively secure PKFE $\FE = \FE.(\Setup,\KG,\Enc,\Dec)$.
\item Adaptively single-ciphertext function private SKFE $\SKFE=\SKFE.(\Setup,\KG,\Enc,\Dec)$.
\item Pseudorandom-secure SKE $\SKE= \SKE.(\Enc,\Dec)$.
\item Puncturable PRF $\PRF=(\prfgen,\prf,\Puncture)$.
\end{itemize}
We set $\xpadlen \seteq \abs{\skfe.\ct} -\abs{x}$ and $\skectlen \seteq \abs{\ske.\ct}$, where $\abs{x}$ is the input length of $\PKFESKL$, $\abs{\skfe.\ct}$ is the ciphertext length of $\SKFE$, and $\abs{\ske.\ct}$ is the ciphertext length of $\SKE$.
\begin{description}
\item[$\Setup(1^\secp)$:] $ $
\begin{itemize}
    \item Generate $(\fe.\pk,\fe.\msk) \gets \FE.\Setup(1^\secp)$.
    \item Output $(\pk,\msk)\seteq (\fe.\pk,\fe.\msk)$.
\end{itemize}
\item[$\qKG(\msk,f)$:] $ $
\begin{itemize}
	\item Generate $(\skl.\ek,\skl.\qsk,\skl.\vk)\gets \SKL.\qKG(1^\secp)$.
	\item Choose $\ske.\ct \chosen \zo{\skectlen}$.
	\item Construct a circuit $W[f,\skl.\ek,\ske.\ct]$, which is described in~\cref{fig:func_wrap_W}.
	\item Generate $\fe.\sk_W \gets \FE.\KG(\fe.\msk,W[f,\skl.\ek,\ske.\ct])$.
	\item Output $\qfsk \seteq (\fe.\sk_W,\skl.\qsk)$ and $\vk \seteq \skl.\vk$.
\end{itemize}
\item[$\Enc(\pk,x)$:] $ $
\begin{itemize}
\item Choose $\prfkey \chosen \prfgen(1^\secp)$.
\item Compute $\fe.\ct \gets \FE.\Enc(\fe.\pk,(x\concat 0^{\xpadlen},\bot,\prfkey))$.
\item Output $\ct \seteq \fe.\ct$.
\end{itemize}
\item[$\qDec(\qfsk,\ct)$:] $ $
\begin{itemize}
\item Parse $\qfsk = (\fe.\sk,\skl.\qsk)$ and $\ct =\fe.\ct$.
\item Compute $\skl.\ct \gets \FE.\Dec(\fe.\sk,\fe.\ct)$.
\item Compute and output $y \gets \SKL.\qDec(\skl.\qsk,\skl.\ct)$.
\end{itemize}
\item[$\qVrfy(\vk,\qfsk^\prime)$:] $ $
\begin{itemize}
\item Parse $\vk = \skl.\vk$ and $\qfsk^\prime=(\fe.\sk^\prime,\skl.\qsk^\prime)$.
\item Compute and output $\SKL.\qVrfy(\skl.\vk,\skl.\qsk^\prime)$.
\end{itemize}
\end{description}

\protocol
{Function $W[f,\skl.\ek,\ske.\ct](x^\prime,\ske.\sk,\prfkey)$}
{The description of $W[f,\skl.\ek,\ske.\ct]$}
{fig:func_wrap_W}
{
\begin{description}
\setlength{\parskip}{0.3mm} 
\setlength{\itemsep}{0.3mm} 
\item[Constants:] Function $f$, PKE-SKL encryption key $\skl.\ek$, SKE ciphertext $\ske.\ct$.
\item[Input:] Plaintext $x^\prime$, SKE key $\ske.\sk$, PRF key $\prfkey$.
\end{description}
\begin{enumerate}
	\item If $\ske.\sk = \bot$, do the following:
	\begin{itemize}
		\item Parse $x^\prime = x \concat \barx$ such that $\abs{\barx}=\xpadlen$.
	\item Compute and output $\skl.\ct \seteq \SKL.\Enc(\skl.\ek,f(x);\prf_{\prfkey}(\skl.\ek))$.
	\end{itemize}
	\item If $\ske.\sk \ne \bot$, do the following:
	\begin{itemize}
	\item Compute $\skfe.\sk \gets \SKE.\Dec(\ske.\sk,\ske.\ct)$.
	\item Compute and output $z \seteq \SKFE.\Dec(\skfe.\sk,x^\prime)$.
	\end{itemize}
\end{enumerate}
}

\paragraph{Correctness.}
The decryption correctness of $\PKFESKL$ follows from the correctness of $\FE$ and the decryption correctness of $\SKL$.
The verification correcntess of $\PKFESKL$ follows from the verification correcntess of $\SKL$.

\subsection{Security Proofs}\label{sec:PKFE_SKL_proofs}
We prove the security of $\PKFESKL$.
\begin{theorem}\label{thm:PKFESKL_from_PKFE}
If $\PKFE$ is adaptively secure, $\SKFE$ is adaptively single-ciphertext function private, $\PRF$ is a secure punctured PRF, and $\SKE$ has the ciphertext pseudorandomness, then $\PKFESKL$ above is Ada-IND-KLA.
\end{theorem}

\begin{theorem}\label{thm:bounded_PKFESKL_from_bounded_PKFE}
If $\PKFE$ is $q$-bounded adaptively secure, $\SKFE$ is adaptively single-ciphertext function private, $\PRF$ is a secure punctured PRF, and $\SKE$ has the ciphertext pseudorandomness, then $\PKFESKL$ above is $q$-bounded Ada-IND-KLA.
\end{theorem}

The proof of~\cref{thm:bounded_PKFESKL_from_bounded_PKFE} is almost the same as that of~\cref{thm:PKFESKL_from_PKFE}. Hence, we focus on the proof of~\cref{thm:PKFESKL_from_PKFE}.
We can also consider a simulation-based security for $q$-bounded security as Kitagawa and Nishimaki~\cite{AC:KitNis22} and believe that we can achieve it using a similar technique. However, it is out of scope of this work.
\ryo{If we consider the bounded collusion-resistant setting, we can conisder a simulation-based definition as~\cite{AC:KitNis22}, but that may be too much.}\takashi{I agree}

\begin{proof}[Proof of~\cref{thm:PKFESKL_from_PKFE}]
In the proof, we embed an SKFE ciphertext $\skfe.\ct \gets \SKFE.\Enc(\skfe.\msk,(x,\bot,\prfkey,0,\bot))$ into the challenge ciphertext.
More specifically, we generate $\fe.\ct \gets \PKFE.\Enc(\fe.\pk,(\skfe.\ct,\ske.\sk,\bot))$ and $\ske.\ct \gets \SKE.\Enc(\ske.\sk,\SKFE.\KG(\skfe.\msk,T[f,\skl.\ek]))$, where $T[f,\skl.\ek]$ is described in~\cref{fig:func_wrap_SKFE}.
By using this embedding, we can use the function privacy of $\SKFE$ and can alter both plaintexts and functions in the proof.

\protocol
{Function $T[f,\skl.\ek](x,x_1,\prfkey,j,\skl.\ct^\ast)$}
{The description of $T[f,\skl.\ek]$}
{fig:func_wrap_SKFE}
{
\begin{description}
\setlength{\parskip}{0.3mm} 
\setlength{\itemsep}{0.3mm} 
\item[Constants:] Function $f$, encryption key of PKE-SKL $\skl.\ek$.
\item[Input:] Plaintext $x$ and $x_1$, PRF key $\prfkey$, index $j$, a PKE-SKL ciphertext $\skl.\ct^\ast$.\footnote{Although inputs $(x_1,j,\skl.\ct^\ast)$ are not used in this function, we need them in hybrid games later.}
\end{description}
\begin{enumerate}
	\item Compute and output $\skl.\ct \seteq \SKL.\Enc(\skl.\ek,f(x);\prf_{\prfkey}(\skl.\ek))$. 
\end{enumerate}
}
Let $q$ be the total number of key queries to $\Oracle{\qKG}$. In the collusion-resistant setting, $q$ is an unbounded polynomial. Note that even if $q$ is an unbounded polynomial, we need only $\poly(\secp)$ bits to describe $q$ as an integer.
We assume that the adversary does not send the same $f$ to $\Oracle{\qKG}$ more than once without loss of generality.
We define a sequence of hybrid games.
\begin{description}
\item[$\hybi{0}$:] This is the same as $\expc{\PKFESKL,\qA}{ada}{ind}{kla}(1^\secp,0)$. More specifically, it is as follows.
        \begin{enumerate}
            \item The challenger generates $(\fe.\pk,\fe.\msk)\gets\FE.\Setup(1^\secp)$ and sends $\pk\seteq \fe.\pk$ to $\qA$.
            $\qA$ can access the following oracles.
            \begin{description}
            \item[$\Oracle{\qKG}(f_i)$:] 
            Given $f_i$, 
            it generates $(\skl.\ek_i,\skl.\qsk_i,\skl.\vk_i)\la\SKL.\qKG(1^\secp)$, $\ske.\ct_i \chosen \zo{\skectlen}$, and $\fe.\sk_{W,i} \gets \FE.\KG(\fe.\msk,W[f_i,\skl.\ek_i,\ske.\ct_i])$, sends $\qfsk_{i} \seteq (\fe.\sk_{W,i},\skl.\qsk_i)$ to $\qA$, and adds $(f_i,\vk_i,\bot)$ to $\List{\qKG}$.
            \item[$\Oracle{\qVrfy}(f_i,\qfsk^\prime_i)$:] Given $(f_i,\qfsk^\prime_i)$, it finds an entry $(f_i,\vk_i,V_i)$ from $\List{\qKG}$ and parse $\qfsk^\prime_i = (\fe.\sk^\prime_i,\skl.\qsk^\prime_i)$. (If there is no such entry, it returns $\bot$.) It returns $d \seteq \SKL.\qVrfy(\skl.\vk,\skl.\qsk^\prime)$. If $V_i = \top$, it does not update the entry. Otherwise, it updates the entry by setting $V_i \seteq d$.
            \end{description}
            \item When $\qA$ sends $(x_0^*,x_1^*)$ to the challenger, the challenger checks if for any entry $(f,\vk,V)$ in $\List{\qKG}$ such that $f(x_0^*)\ne f(x_1^*)$, it holds that $V=\top$. If so, the challenger generates $\prfkey \gets \PRF.\Gen(1^\secp)$ and $\fe.\ct^\ast \gets \FE.\Enc(\fe.\pk,(x_\coin^\ast\concat 0^{\xpadlen},\bot,\prfkey))$ and sends $\ct^* \seteq \fe.\ct^\ast$ to $\qA$. Otherwise, the challenger outputs $0$. Hereafter, $\qA$ is not allowed to send a function $f$ such that $f(x_0^*)\ne f(x_1^*)$ to $\Oracle{\qKG}$.
            \item $\qA$ outputs a guess $\coin^\prime$ for $\coin$. The challenger outputs $\coin'$ as the final output of the experiment.
        \end{enumerate}
\item[$\hybi{1}$:] This is the same as $\hybi{0}$ except that for all $i\in[q]$, we generate $\ske.\ct_i \gets \SKE.\Enc(\ske.\sk,\skfe.\sk_i)$, where $\skfe.\sk_i \gets \SKFE.\KG(\skfe.\msk,T[f_i,\skl.\ek_i])$ and $(\skl.\ek_i,\skl.\qsk_i,\skl.\vk_i) \gets \SKL.\KeyGen(1^\secp)$. Note that the SKE secret key $\ske.\sk$ never appears in the view of $\qA$. Hence, we obtain $\abs{\Pr[\hybi{0}=1] - \Pr[\hybi{1} =1]} = \negl(\secp)$ by the security of $\SKE$.\ryo{I omit the proof since it is trivial. If you think we should write down the reduction, please let me know.}\takashi{I don't think that's needed.}
\item[$\hybi{2}$:] This is the same as $\hybi{1}$ except that we generate $\fe.\ct^\ast \gets \PKFE.\Enc(\fe.\pk,(\skfe.\ct^\ast,\ske.\sk,\bot))$, where $\skfe.\ct^\ast \gets \SKFE.\Enc(\skfe.\msk,(x_0^\ast,\bot,\prfkey,0,\bot))$. By the definition of $W$ described in~\cref{fig:func_wrap_W}, if we decrypt $\fe.\ct^\ast$ by $\fe.\sk_i$, we obtain
\begin{itemize}
\item $\SKL.\Enc(\skl.\ek_i,f(x_0^\ast);\prf_{\prfkey}(\skl.\ek_i))$ in $\hybi{1}$ since the plaintext in $\fe.\ct^\ast$ is $(x_0^\ast \concat 0^{\xpadlen},\bot,\prfkey)$,
\item $z_i = \SKFE.\Dec(\skfe.\sk_i,\skfe.\ct^\ast)$ in $\hybi{2}$ since $\ske.\ct$ is a ciphertext of $\skfe.\sk_i$ and the plaintext in $\fe.\ct^\ast$ is $(\skfe.\ct^\ast,\ske.\sk,\bot)$, where $\skfe.\ct^\ast =\SKFE.\Enc(\skfe.\msk,(x_0^\ast,\bot,\prfkey,0,\bot))$. By the correctness of $\SKFE$, $z_i = \SKL.\Enc(\skl.\ek_i,f(x_0^\ast);\prf_{\prfkey}(\skl.\ek_i))$.
\end{itemize}
That is, for all $i\in [q]$, it holds that $W[f_i,\skl.\ek_i,\ske.\ct_i](x_0^\ast \concat 0^{\xpadlen},\bot,\prfkey) = W[f_i,\skl.\ek_i,\ske.\ct_i](\skfe.\ct^\ast,\ske.\sk,\bot)$. Hence, we can use the security of $\PKFE$ and obtain $\abs{\Pr[\hybi{1}=1] - \Pr[\hybi{2} =1]} = \negl(\secp)$. See~\cref{lem:pkfe_hyb_one_two} for the detail.

After this game, we can focus on $\SKFE$.
\item[$\hybi{3}$:] This is the same as $\hybi{2}$ except that we generate $\skfe.\ct^\ast \gets \SKFE.\Enc(\skfe.\msk,(x_0^\ast,x_1^\ast,\prfkey,0,\bot))$ and $\skfe.\sk_i \gets \SKFE.\KG(\skfe.\msk,\Vhyb[f_i,\skl.\ek_i,i])$, where $\Vhyb[f_i,\skl.\ek_i,i]$ is described in~\cref{fig:func_wrap_SKFE_hybrid}. Since $i \in [q]$, it holds that $\Vhyb[f_i,\skl.\ek_i,i](x_0^\ast,x_1^\ast,\prfkey,0,\bot) = T[f_i,\skl.\ek_i](x_0^\ast,\bot,\prfkey,0,\bot)$ for all $i\in [q]$.
Hence, by the adaptively single-ciphertext function privacy of $\SKFE$, we obtain $\abs{\Pr[\hybi{2}=1] - \Pr[\hybi{3} =1]} = \negl(\secp)$. See~\cref{lem:pkfe_hyb_two_three} for the detail.
\item[$\hybij{3}{j}$:] This is the same as $\hybi{3}$ except that we generate $\skfe.\ct^\ast \gets \SKFE.\Enc(\skfe.\msk,(x_0^\ast,x_1^\ast,\prfkey,j,\bot))$. Apparently, $\hybij{3}{0}$ is the same as $\hybi{3}$. We show it holds that $\abs{\Pr[\hybij{3}{j-1}=1] - \Pr[\hybij{3}{j} =1]} = \negl(\secp)$ for $j \in [q]$ in~\cref{lem:embed_and_forget}.
\item[$\hybi{4}$:] This is the same as $\hybij{3}{q}$ except that we generate $\skfe.\sk_i \gets \SKFE.\KG(\skfe.\msk,T[f_i,\skl.\ek_i])$ and $\skfe.\ct^\ast \gets \SKFE.\Enc(\skfe.\msk,(x_1^\ast,\bot,\prfkey,0,\bot))$.
Recall that in $\hybij{3}{q}$, we use $\skfe.\sk_i \gets \SKFE.\KG(\skfe.\msk,\Vhyb[f_i,\skl.\ek_i,i])$ and $\skfe.\ct^\ast \gets \SKFE.\Enc(\skfe.\msk,(x_0^\ast,x_1^\ast,\prfkey,q,\bot))$. By the definition of $\Vhyb$ and $T$, it holds that for all $i\in [q]$,
\begin{align}
\Vhyb[f_i,\skl.\ek_i,i](x_0^\ast,x_1^\ast,\prfkey,q,\bot) & = \SKL.\Enc(\skl.\ek_i,f_i(x_1^\ast);\prf_{\prfkey}(\skl.\ek_i)) \\
&= T[f_i,\skl.\ek_i](x_1^\ast,\bot,\prfkey,0,\bot).
\end{align}
Hence, we can use the adaptively single-ciphertext function privacy of $\SKFE$ and obtain $\abs{\Pr[\hybij{3}{q}=1] - \Pr[\hybi{4} =1]} = \negl(\secp)$. See~\cref{lem:pkfe_hyb_three_four} for the detail.

Now, we use $x_1^\ast$ instead of $x_0^\ast$ and erased $x_0^\ast$ in the challenge ciphertext. Hence, we focus on $\PKFE$ again and undo the changes from $\hybi{1}$ to $\hybi{2}$ and from $\hybi{0}$ to $\hybi{1}$.
\item[$\hybi{5}$:] This is the same as $\hybi{4}$ except that we generate $\fe.\ct^\ast \gets \PKFE.\Enc(\fe.\pk,(x_1^\ast\concat 0^{\xpadlen},\bot,\prfkey))$. This is the reverse transition from $\hybi{1}$ to $\hybi{2}$, so we obtain $\abs{\Pr[\hybi{4}=1] - \Pr[\hybi{5} =1]} = \negl(\secp)$ by the security of PKFE as the proof of~\cref{lem:pkfe_hyb_one_two}.
\item[$\hybi{6}$:] This is the same as $\hybi{5}$ except that we generate $\ske.\ct_i \chosen \zo{\ell}$. As the transition from $\hybi{0}$ to $\hybi{1}$, we obtain $\abs{\Pr[\hybi{5}=1] - \Pr[\hybi{6} =1]} = \negl(\secp)$ by the ciphertext pseudorandomness of SKE.
It is easy to see that $\hybi{6}$ is the same as $\expc{\PKFESKL,\qA}{ada}{ind}{kla}(1^\secp,1)$.\ryo{As the transition from $\hybi{0}$ to $\hybi{1}$, I omit the proof.}
\end{description}
\protocol
{Function $\Vhyb[f_i,\skl.\ek_i,i](x_0,x_1,\prfkey,j,\skl.\ct^\ast)$}
{The description of $\Vhyb[f_i,\skl.\ek_i,i]$}
{fig:func_wrap_SKFE_hybrid}
{
\begin{description}
\setlength{\parskip}{0.3mm} 
\setlength{\itemsep}{0.3mm} 
\item[Constants:] Function $f_i$, PKE-SKL encryption key $\skl.\ek_i$, index $i$.
\item[Input:] Plaintext $x_0$ and $x_1$, PRF key $\prfkey$, index $j$, a PKE-SKL ciphertext $\skl.\ct^\ast$. 
\end{description}
\begin{enumerate}
	\item If $i \le j$, compute and output $\skl.\ct_i \seteq \SKL.\Enc(\skl.\ek_i,f_i(x_1);\prf_{\prfkey}(\skl.\ek_i))$.
	\item If $i > j$, compute and output $\skl.\ct_i \seteq \SKL.\Enc(\skl.\ek_i,f_i(x_0);\prf_{\prfkey}(\skl.\ek_i))$.
\end{enumerate}
}

If we prove~\cref{lem:embed_and_forget}, we complete the proof of~\cref{thm:PKFESKL_from_PKFE}.
\end{proof}

\begin{lemma}\label{lem:embed_and_forget}
For all $j \in [q]$, it holds that $\abs{\Pr[\hybij{3}{j-1}=1] - \Pr[\hybij{3}{j} =1]} = \negl(\secp)$ if $\SKFE$ is fully function private, $\SKL$ is IND-KLA, and $\PRF$ is a puncturable PRF.
\end{lemma}
\begin{proof}
We define a sequence of hybrid games.
\begin{description}
\item[$\sfG_0$:] This is the same as $\hybij{3}{j-1}$. That is, $\skfe.\ct^\ast \gets \SKFE.\Enc(\skfe.\msk,(x_0^\ast,x_1^\ast,\prfkey,j-1,\bot))$ and $\skfe.\sk_i \gets \SKFE.\KG(\skfe.\msk,\Vhyb[f_i,\skl.\ek_i,i])$.
\item[$\sfG_1$:] This is the same as $\sfG_0$ except that we generate $\skfe.\ct^\ast \gets \SKFE.\Enc(\skfe.\msk,(x_0^\ast,x_1^\ast,\prfkey,j,\skl.\ct^\ast))$, where $\skl.\ct^\ast \gets \SKL.\Enc(\skl.\ek_j,f_j(x_0^\ast);\prf_{\prfkey}(\skl.\ek_j))$ and $\skfe.\sk_i \gets \SKFE.\KG(\skfe.\msk,\Vembed[f_i,\skl.\ek_i,i])$, where $\Vembed[f_i,\skl.\ek_i,i]$ is described in~\cref{fig:func_wrap_SKFE_embed}.
By the definitions of $\Vhyb$ and $\Vembed$, it holds that
\[\Vhyb[f_i,\skl.\ek_i,i](x_0^\ast,x_1^\ast,\prfkey,j-1,\bot)=\Vembed[f_i,\skl.\ek_i,i](x_0^\ast,x_1^\ast,\prfkey,j,\skl.\ct^\ast)\]
for all $i\in [q]$ since $\skl.\ct^\ast$ is an encryption of $f_j(x_0^\ast)$.
Hence, by the adaptively single-ciphertext function privacy of $\SKFE$, we obtain $\abs{\Pr[\sfG_0 =1] - \Pr[\sfG_1 =1]} = \negl(\secp)$. See~\cref{lem:pkfe_Ghyb_zero_one} for the detail.
\item[$\sfG_2$:] This is the same as $\sfG_1$ except that we use a punctured PRF key $\prfkey_{\ne  \skl.\ek_j} = \Puncture(\prfkey, \skl.\ek_j)$. By the functionality of punctured PRF keys, it holds that
\[\Vembed[f_i,\skl.\ek_i,i](x_0^\ast,x_1^\ast,\prfkey,j,\skl.\ct^\ast)=\Vembed[f_i,\skl.\ek_i,i](x_0^\ast,x_1^\ast,\prfkey_{\ne  \skl.\ek_j},j,\skl.\ct^\ast)\]
for all $i\in [q]$. Note that $\Vembed$ directly uses $\skl.\ct^\ast$ instead of computing $\SKL.\Enc(\skl.\ek_j,f_j(x_0^\ast);\prf_{\prfkey}(\skl.\ek_j))$, so $\prfkey_{\ne \skl.\ek_j}$ is sufficient for the functional equivalence. The only difference between the two games is whether the PRF key is $\prfkey$ or $\prfkey_{\ne \skl.\ek_j}$. Hence, we can use the adaptively single-ciphertext function privacy of $\SKFE$ and obtain $\abs{\Pr[\sfG_1 =1] - \Pr[\sfG_2 =1]} = \negl(\secp)$. We omit the proof since it is easy.\ryo{I omit the proof since it is easy. If you think we should write down it, please let me know.}
\item[$\sfG_3$:] This is the same as $\sfG_2$ except that we generate $\skl.\ct^\ast \gets \SKL.\Enc(\skl.\ek_j,f_j(x_0^\ast))$. That is, we use uniform randomness for generating $\skl.\ct^\ast$. By the punctured pseudorandomness of $\PRF$, we obtain $\abs{\Pr[\sfG_2 =1] - \Pr[\sfG_3 =1]} = \negl(\secp)$. We omit the proof since it is easy.\ryo{I omit the proof since it is easy. If you think we should write down it, please let me know.}
\item[$\sfG_4$:] This is the same as $\sfG_3$ except that we generate $\skl.\ct^\ast \gets \SKL.\Enc(\skl.\ek_j,f_j(x_1^\ast))$.
We consider two cases.
\begin{itemize}
 \item If $(f_j,\vk_j,\bot)$ is recorded in $\List{\qKG}$, that is, valid $\qfsk_{j}$ is not returned, it must hold that $f_j(x_0^\ast)=f_j(x_1^\ast)$ by the requirement of Ada-IND-KLA security. In this case, the distribution of $\skl.\ct^\ast \gets \SKL.\Enc(\skl.\ek_j,f_j(x_0^\ast))$ is trivially the same as that of $\skl.\ct^\ast \gets \SKL.\Enc(\skl.\ek_j,f_j(x_1^\ast))$. Hence, we obtain $\Pr[\sfG_3 =1] = \Pr[\sfG_4 =1]$.
\item If $(f_j,\vk_j,\top)$ is recorded in $\List{\qKG}$, that is, it is certified that the adversary returned valid $\qfsk_{j}$, it could hold that $f_j(x_0^\ast)\ne f_j(x_1^\ast)$ by the requirement of Ada-IND-KLA security. In this case, we use IND-KLA security of $\SKL$ since $\skl.\qsk_j$ was returned. We have that $\skl.\ct^\ast \gets \SKL.\Enc(\skl.\ek_j,f_j(x_0^\ast))$ is computationally indistinguishable from $\skl.\ct^\ast \gets \SKL.\Enc(\skl.\ek_j,f_j(x_1^\ast))$. Hence, we obtain $\abs{\Pr[\sfG_3 =1] - \Pr[\sfG_4 =1]} = \negl(\secp)$ in this case. See~\cref{lem:pkfe_Ghyb_three_four} for the detail.
\shota{Reference here may be incorrect.}\ryo{Fixed.}
\end{itemize}
Hence, we obtain $\abs{\Pr[\sfG_3 =1] - \Pr[\sfG_4 =1]} = \negl(\secp)$ in either cases.
\item[$\sfG_5$:] This is the same as $\sfG_4$ except that we undo the change in $\sfG_3$. That is, we use $\prf_{\prfkey}( \skl.\ek_j)$ for the randomness of $\skl.\ct^\ast$. We obtain $\abs{\Pr[\sfG_4 =1] - \Pr[\sfG_5 =1]} = \negl(\secp)$ by the punctured pseudorandomness of $\PRF$. We omit the proof since it is easy.\ryo{I omit the proof since it is easy. If you think we should write down it, please let me know.}
\item[$\sfG_6$:] This is the same as $\sfG_5$ except that we undo the change in $\sfG_2$. That is, we use a unpunctured PRF key $\prfkey$. We obtain $\abs{\Pr[\sfG_1 =1] - \Pr[\sfG_2 =1]} = \negl(\secp)$ by the adaptively single-ciphertext function privacy as the transition from $\sfG_1$ to $\sfG_2$. So, we omit the proof.\ryo{I omit the proof since it is easy. If you think we should write down it, please let me know.}
\item[$\sfG_7$:] This is the same as $\sfG_6$ except that we undo the change in $\sfG_1$, but the index is still $j$. That is, we use $\skfe.\ct^\ast \gets \SKFE.\Enc(\skfe.\msk,(x_0^\ast,x_1^\ast,\prfkey,j,\bot))$ and $\skfe.\sk_i \gets \SKFE.\KG(\skfe.\msk,\Vhyb[f_i,\skl.\ek_i,i])$. We obtain $\abs{\Pr[\sfG_6 =1] - \Pr[\sfG_7 =1]} = \negl(\secp)$ by the adaptively single-ciphertext function privacy of $\SKFE$. The proof is similar to that of~\cref{lem:pkfe_Ghyb_zero_one}. So, we omit the proof.
\end{description}

\protocol
{Function $\Vembed[f_i,\skl.\ek_i,i](x_0,x_1,\prfkey,j,\skl.\ct^\ast)$}
{The description of $\Vembed[f_i,\skl.\ek_i,i]$}
{fig:func_wrap_SKFE_embed}
{
\begin{description}
\setlength{\parskip}{0.3mm} 
\setlength{\itemsep}{0.3mm} 
\item[Constants:] Function $f_i$, encryption key of PKE-SKL $\skl.\ek_i$, index $i$.
\item[Input:] Plaintext $x_0$, $x_1$, PRF key $\prfkey$, index $j$, an SKL ciphertext $\skl.\ct^\ast$.
\end{description}
\begin{enumerate}
	\item If $i=j$, output $\skl.\ct^\ast$.
	\item If $i < j$, compute and output $\skl.\ct_i \seteq \SKL.\Enc(\skl.\ek_i,f_i(x_1);\prf_{\prfkey}(\skl.\ek_i))$.
	\item If $i > j$, compute and output $\skl.\ct_i \seteq \SKL.\Enc(\skl.\ek_i,f_i(x_0);\prf_{\prfkey}(\skl.\ek_i))$.
\end{enumerate}
}
It is easy to see that $\sfG_7$ is the same as $\hybij{3}{j}$. Therefore, we complete the proof.
\end{proof}

By~\cref{thm:bounded_PKFESKL_from_bounded_PKFE,thm:1ct_adaptive_function_private_SKFE,thm:pprf-owf,thm:pseudorandom_ske,thm:bounded_PKFE_from_PKE,thm:PKE-SKL_from_PKE}, we obtain the following corollary.
\begin{corollary}
If there exists IND-CPA secure PKE, there exists $q$-bounded Ada-IND-KLA PKFE-SKL for $\Ppoly$.
\end{corollary}

By~\cref{thm:PKFESKL_from_PKFE,thm:1ct_adaptive_function_private_SKFE,thm:pprf-owf,thm:pseudorandom_ske,thm:PKE-SKL_from_PKE} and known theorems about PKFE~\cite{TCC:GarSri16,TCC:LiMic16,C:KNTY19}, we obtain the following corollary.
\begin{corollary}
If there exists single-key selective-message-function secure\footnote{The adversary must select the target plaintext pair and function at the beginning of the game. This is the same as weakly selective security by Garg and Srinivasan~\cite{TCC:GarSri16}.} and weakly compact PKFE for $\Ppoly$, there exists Ada-IND-KLA PKFE for $\Ppoly$.
\end{corollary}

\else 

\input{submission_preliminaries}
\input{submission_PKE-SKL}
\input{submission_CoIC}
\input{submission_construction_PKE-SKL}
\input{submission_ABE-SKL}
\input{submission_FE-SKL}
\fi

\ifnum\anonymous=1
\else
 \section*{Acknowledgement}
We thank Jiayu Zhang for pointing out a technical similarity to \cite{STOC:Zhang21,FOCS:Zhang22}, Prabhanjan Ananth for discussions on the relationship between our work and their concurrent work \cite{APV23}, and anonymous reviewers of QIP 2023 and Eurocrypt 2023 for their valuable comments. 
This work was supported in part by the DST “Swarnajayanti” fellowship, Cybersecurity Center of Excellence, IIT Madras, National Blockchain Project and the Algorand Centres of Excellence programme managed by Algorand Foundation. Any opinions, findings, and conclusions or recommendations expressed in this material are those of the author(s) and do not necessarily reflect the views of sponsors.
The fourth author was partially supported by JST AIP Acceleration Research JPMJCR22U5 and JSPS KAKENHI Grant Number 19H01109, Japan. 
\fi

	\ifnum\llncs=1
\bibliographystyle{splncs04}
\bibliography{bib/abbrev3,bib/crypto,bib/siamcomp_jacm,bib/other}
	\else
\bibliographystyle{alpha} 
\bibliography{bib/abbrev3,bib/crypto,bib/siamcomp_jacm,bib/other}
	\fi

\ifnum\cameraready=0
	\ifnum\llncs=0
	\appendix

\section{SDE Implies PKE-SKL}\label{sec:SDE_and_PKE-SKL}
In this section, we discuss the relationship between SDE and PKE-SKL. 
There are many incomparable security definitions for SDE in the literature. 
Coladangelo et al. \cite{C:CLLZ21} defined two incomparable security definitions called \emph{CPA-style anti-piracy} and \emph{random challenge anti-piracy}.\footnote{They actually also defined stronger variants of them called strong anti-piracy security and strong anti-piracy against random plaintexts. See \cite[Definition 6.11 and D.4 in the full version]{C:CLLZ21} for the detail.\ryo{I slightly modified these sentences.}}  All constructions of SDE in \cite{C:CLLZ21} are shown to satisfy both CPA-style anti-piracy and random challenge anti-piracy. 
Georgiou and Zhandry \cite{EPRINT:GeoZha20} defined yet another security definition, which is similar to but slightly different from CPA-style anti-piracy of \cite{C:CLLZ21}.\footnote{We note that \cite{EPRINT:GeoZha20}  appeared before \cite{C:CLLZ21}.} Though we do not see any relationships between security notions in \cite{EPRINT:GeoZha20} and \cite{C:CLLZ21}, it seems possible to prove that the construction given in \cite{EPRINT:GeoZha20} satisfies both CPA-style anti-piracy and random challenge anti-piracy of \cite{C:CLLZ21} because it is very similar to one of the schemes given in \cite{C:CLLZ21}.\footnote{Here, we are referring to the construction of SDE based on one-shot signatures and extractable witness encryption in  \cite[Section 5]{EPRINT:GeoZha20}.
For proving that the scheme satisfies the security notions of~\cite{C:CLLZ21}, we will need to go trough the ``strong'' variants of them similarly to \cite{C:CLLZ21}. 
}
In the following, we show that SDE with random challenge anti-piracy implies IND-KLA secure PKE-SKL. This means that all known constructions of SDE can be used to construct PKE-SKL. 

The definitions of SDE and its random challenge anti-piracy are given below. 
The syntax of SDE is identical to that of PKE except that the key generation and decryption algorithms are quantum and the decryption key is quantum. 
\begin{definition}[Single-Decryptor Encryption]
A single-decryptor encryption (SDE) scheme $\SDE$ is a tuple of three algorithms $(\qKG, \Enc, \qDec)$. 
Below, let $\cX$  be the message space of $\SDE$. 
\begin{description}
\item[$\qKG(1^\secp)\ra(\ek,\qdk)$:] The key generation algorithm takes a security parameter $1^\lambda$, and outputs an encryption key $\ek$ and a decryption key $\qdk$. 

\item[$\Enc(\ek,\msg)\ra\ct$:] The encryption algorithm takes an encryption key $\ek$ and a message $\msg \in \cX$, and outputs a ciphertext $\ct$.

\item[$\qDec(\qdk,\ct)\ra\tilde{\msg}$:] The decryption algorithm takes a decryption key $\qdk$ and a ciphertext $\ct$, and outputs a value $\tilde{\msg}$.

\item[Correctness:]For every $\msg \in \cX$, we have
\begin{align}
\Pr\left[
\qDec(\qdk, \ct) \allowbreak = \msg
\ \middle |
\begin{array}{ll}
(\ek,\qdk)\gets\qKG(1^\secp)\\
\ct\gets\Enc(\ek,\msg)
\end{array}
\right] 
=1-\negl(\secp).
\end{align}
\end{description}
\end{definition}

In the following definition of random challenge anti-piracy, we use the notion of quantum programs with classical inputs and outputs as defined in  \Cref{def:Q_program_C_IO}. 
\begin{definition}[Random Challenge Anti-Piracy]
We say that an SDE scheme $\SDE$  with the message space $\Xs$ satisfies random-challenge anti-piracy, if it satisfies the following requirement, formalized from the experiment $\expb{\SDE,\qA}{rand}{chal}(1^\secp)$ between an adversary $\qA$ and a challenger $\qC$:
        \begin{enumerate}
            \item  $\qC$ 
            runs $(\ek,\qdk)\gets\qKG(1^\secp)$ and sends $\ek$ and $\qdk$ to $\qA$. 
            \item 
          $\qA$ sends two (possibly entangled) quantum programs $(\qD_0,\qD_1)$ with classical inputs and outputs to $\qC$. 
           \item 
           For $b\in \bit$, 
           $\qC$ 
            chooses $\msg^*_b\gets \Xs$, generates $\ct^*_b\la\Enc(\ek,\msg^*_b)$, and runs $\qD_b$ on input $\ct^*_b$ to obtain an output $\msg_b$.  
         $\qC$ outputs $1$ if $\msg_b=\msg^*_b$ for $b\in \bit$ and otherwise outputs $0$ as the final output of the experiment.
        \end{enumerate}
        For any QPT $\qA$, it holds that
\begin{align}
\advb{\SDE,\qA}{rand}{chal}(\secp) \seteq \Pr[\expb{\SDE,\qA}{rand}{chal} (1^\secp) \ra 1]\leq \negl(\secp).
\end{align}
\end{definition}

We prove the following theorem.
\begin{theorem}
If there exists an SDE scheme that satisfies random challenge anti-piracy, there exists an IND-KLA secure PKE-SKL scheme. 
\end{theorem}
\begin{proof}
Let $\SDE=(\SDE.\qKG, \SDE.\Enc, \SDE.\qDec)$ be an SDE scheme that satisfies random challenge anti-piracy. 
By \cref{thm:one-query-ow_to_full}, it suffices to construct a one-query OW-KLA secure PKE-SKL scheme.
We construct a one-query OW-KLA secure PKE-SKL scheme $\SKL=(\SKL.\qKG, \SKL.\Enc, \SKL.\qDec, \SKL.\qVrfy)$ as follows. 
\begin{description}
\item[$\SKL.\qKG(1^\secp)$:]~
Run $(\sde.\ek,\sde.\qdk)\gets \SDE.\qKG(1^\secp)$
    and output $\skl.\ek\seteq \sde.\ek$, $\skl.\qdk\seteq \sde.\qdk$, and $\skl.\vk\seteq \sde.\ek$. 

\item[$\SKL.\Enc(\skl.\ek,\msg)$:] This is identical to $\SDE.\Enc$.

\item[$\SKL.\qDec(\qdk,\ct)$:] This is identical to $\SDE.\qDec$. 

\item[$\SKL.\qVrfy(\skl.\vk,\skl.\widetilde{\qdk})$:] 
Parse $\skl.\vk=\sde.\ek$, 
choose $\msg^*\gets \Xs$, 
run $\ct^* \gets \SDE.\Enc(\sde.\ek,\msg^*)$ and $\msg\gets \SDE.\qDec(\skl.\widetilde{\qdk},\ct^*)$,  
and output $\top$ if and only if $\msg=\msg^*$. 
\end{description}
Suppose that $\SKL$ is not one-query OW-KLA secure. Let $\qA$ be a QPT adversary that breaks the one-query OW-KLA security of $\SKL$. We construct a QPT adversary $\qB$ that breaks the random challenge anti-piracy of $\SDE$ as follows.
\begin{description}
\item[$\qB(\sde.\ek,\sde.\qdk)$:]
Set $\skl.\ek\seteq \sde.\ek$, $\skl.\qdk\seteq \sde.\qdk$, and $\skl.\vk\seteq \sde.\ek$ and sends $(\skl.\ek,\skl.\qdk,\skl.\vk)$ to $\qA$.
When $\qA$ makes a verification query $\skl.\widetilde{\qdk}$, $\qB$ returns $1$ to $\qA$ as the response from the oracle. Let $\qD_0$ be the quantum program with classical inputs and outputs  that takes $\ct$ as input and outputs $\msg\gets \SDE.\qDec(\skl.\widetilde{\qdk},\ct)$. 
When $\qA$ sends $\requestchallenge$, let $\qD_1$ be the quantum program with classical inputs and outputs, in which $\qA$'s internal state is hardwired, that takes $\ct$ as input, runs the rest of $\qA$ on the challenge ciphertext $\ct$, and outputs $\qA$'s output $\msg$. 
Output $(\qD_0,\qD_1)$. 
\end{description}
By the construction of $\qB$ and the deferred measurement principle, it is immediate to see that 
$\advb{\SDE,\qB}{rand}{chal}(\secp)=\advb{\SKL,\qA}{ow}{kla}(\secp)$. Thus, $\qB$ breaks the random challenge anti-piracy of $\SDE$, which is contradiction. Therefore, $\SKL$ is one-query OW-KLA secure.  
\end{proof}
\begin{remark}[On CPA-Style Anti-Piracy]
We do not know if SDE with CPA-style anti-piracy implies PKE-SKL. On the other hand, it seems possible to show that SDE with the ``strong'' variant of CPA-style anti-piracy (called strong anti piracy~\cite[Definition 6.11 in the full version]{C:CLLZ21}) implies PKE-SKL. 
In the single-bit encryption setting, the security roughly means that the adversary given one decryption key cannot generate two ``good'' distinguishers that distinguish encryptions of $0$ and $1$. 
Then our idea is to construct a PKE-SKL scheme whose verification algorithm accepts if a returned decryption key gives a ``good'' distinguisher. Then the strong anti piracy ensures that if the adversary passes the verification, then it cannot keep a ``good'' distinguisher, which in particular means that it cannot distinguish encryptions of $0$ and $1$. Thus, the PKE-SKL scheme is one-query IND-KLA secure. \takashi{Does this explanation make sense? If no, we can remove the technical explanation.}
\end{remark}
    
\section{OW-CPA from CoIC-KLA}\label{sec:OW_from_CoIC}
We show the following lemma.
\begin{lemma}
If a PKE scheme with a super-polynomial-size message space is CoIC-KLA secure, then it is OW-CPA secure. \takashi{By the way, can we prove that CoIC-KLA implies IND-CPA? If so, proving CoIC-KLA -> IND-CPA -> OW-CPA might be clearer. 
But I couldn't see it.}\fuyuki{Me neither. I think the problem is that the random coin $a$ is not given to the adversary in CoIC-KLA game. If the adversary is given $a$, the adversary can break CoIC-KLA security using an IND-CPA breaker by giving $\Enc(\ek,\msg_{a\oplus b})$ to the IND-CPA breaker. Is it correct? (In fact, I think there is no problem to give $a$ to the adversary for CoIC-KLA security. We should have defined CoIC-KLA security in this way for clear implications.)}
\end{lemma}
\begin{proof}
Let $\PKE=(\KG, \Enc, \Dec)$ be a CoIC-KLA secure PKE scheme with the message space $\Xs$ such that $|\Xs|$ is super-polynomial in $\secp$. Toward contradiction, suppose that it is not OW-CPA secure. Let $\qA$ be an adversary that breaks OW-CPA security of $\PKE$. Then we construct $\qB$ that breaks CoIC-KLA security of $\PKE$ as follows.
\begin{description}
\item[$\qB(\ek_0,\ek_1,\qdk)$:] 
Measure $\qdk$ to get $(\beta,\dk_\beta)$ for $\beta\in \bit$. 
Choose $(\msg_0^*,\msg_1^*)\gets \Xs^2$ and send $(\msg_0^*,\msg_1^*)$ to the challenger (without making any oracle query). Upon receiving $(\ct^*_0,\ct^*_1)$ from the challenger, 
run $\msg_\beta'\gets \Dec(\dk_\beta,\ct^*_\beta)$ and $\msg_{\beta\oplus 1}'\gets \qA(\ek_{\beta \oplus 1},\ct^*_{\beta \oplus 1})$ and output $0$ if $\msg_0'=\msg_1'$ and $1$ otherwise.  
\end{description}
Note that the challenger implicitly chooses $a,b\gets \bit$ and generates 
$\ct^*_0\gets \Enc(\ek_0,\msg^*_a)$ and $\ct^*_1\gets \Enc(\ek_1,\msg^*_{a\oplus b})$. 
$\qB$'s goal is to guess $b$. 

If $b=0$, 
by the correctness of $\PKE$, we have $\Pr[\msg'_\beta=\msg^*_{a}]=1-\negl(\secp)$.
By the assumption that $\qA$ breaks OW-CPA security, $\Pr[\msg'_{\beta\oplus 1}=\msg^*_{a}]$ is non-negligible. 
In particular, $\Pr[\qB(\ek_0,\ek_1,\qdk)\rightarrow 0 |b=0]$ is non-negligible.   
If $b=1$, 
by the correctness of $\PKE$, we have $\Pr[\msg'_\beta=\msg^*_{a+\beta}]=1-\negl(\secp)$.
On the other hand,  $\ct^*_{\beta \oplus 1}$ contains no information of $\msg^*_{a\oplus \beta}$. Therefore, $\Pr[\msg'_{\beta \oplus 1}=\msg^*_{a\oplus \beta}]\le 1/|\Xs|=\negl(\secp)$. 
Thus, $\Pr[\qB(\ek_0,\ek_1,\qdk)\rightarrow 0 |b=1]=\negl(\secp)$. 
Thus, $
|2\Pr[\qB(\ek_0,\ek_1,\qdk)\rightarrow b]-1|
=|\Pr[\qB(\ek_0,\ek_1,\qdk)\rightarrow 0 |b=0]-\Pr[\qB(\ek_0,\ek_1,\qdk)\rightarrow 0 |b=1]|$ is non-negligible. This contradicts the assumed  CoIC-KLA security. Thus, $\PKE$ is OW-CPA secure. 
\end{proof}

\section{Deferred Proofs for PKFE-SKL}\label{sec:proofs_PKFE_SKL}
In this section, we present the deferred proofs in~\cref{sec:PKFE_SKL}.
\begin{lemma}\label{lem:pkfe_hyb_one_two}
If $\PKFE$ is adaptively secure, it holds that $\abs{\Pr[\hybi{1}=1] - \Pr[\hybi{2} =1]} =\negl(\secp)$.
\end{lemma}
\begin{proof}
We construct an adversary $\qB$ for $\PKFE$ by using the distinguisher $\qD$ for these two games.
\begin{enumerate}
\item $\qB$ is given $\fe.\pk$ and sends $\pk\seteq \fe.\pk$ to $\qD$. $\qB$ also generates $\ske.\sk \chosen \zo{\secp}$ and $\skfe.\msk \gets \SKFE.\Setup(1^\secp)$.
\item When $\qD$ sends $f_i$ to $\Oracle{\qKG}$, 
$\qB$ generates $(\skl.\ek_i,\skl.\qsk_i,\skl.\vk_i) \gets \SKL.\qKG(1^\secp)$, $\skfe.\sk_i \gets \SKFE.\KG(\skfe.\msk,V[f_i,\skl.\ek_i])$, and $\ske.\ct_i \gets \SKE.\Enc(\ske.\sk, \skfe.\sk_i)$. Then, $\qB$ sends $W[f_i,\skl.\ek_i,\ske.\ct_i]$ to its challenger and receives $\fe.\sk_{W,i} \gets \FE.\KG(\fe.\msk,W[f_i,\skl.\ek_i,\ske.\ct_i])$. $\qB$ returns $\qfsk_{i} \seteq (\fe.\sk_{W,i},\skl.\qsk_i)$ to $\qD$ and adds $(f_i,\skl.\vk_i,\bot)$ to $\List{\qKG}$.
\item When $\qD$ sends $(f_i,\qfsk_i^\prime)$ to $\Oracle{\qVrfy}$, $\qB$ finds an entry $(f_i,\skl.\vk_i,V_i)$ from $\List{\qKG}$ and parses $\qfsk_i^\prime = (\fe.\sk_i^\prime,\skl.\qsk_i^\prime)$. $\qB$ returns $d \seteq \SKL.\qVrfy(\skl.\vk_i,\skl.\qsk_i^\prime)$. If $V_i = \top$, $\qB$ does not update the entry. Otherwise, $\qB$ updates the entry by setting $V_i \seteq d$.
\item When $\qD$ sends $(x_0^\ast,x_1^\ast)$, $\qB$ generates $\prfkey\gets \PRF.\Gen(1^\secp)$, $\skfe.\ct^\ast \gets \SKFE.\Enc(\skfe.\msk,(x_0^\ast,\bot,\prfkey,0,\bot))$. $\qB$ sets $X_0^\ast \seteq (x_0^\ast \concat 0^{\xpadlen},\bot,\prfkey)$ and $X_1^\ast \seteq (\skfe.\ct^\ast,\ske.\sk,\bot)$, sends $(X_0^\ast,X_1^\ast)$ to its challenger, and receives $\fe.\ct^\ast$. $\qB$ passes $\ct^\ast \seteq \fe.\ct^\ast$ to $\qD$.
\item $\qB$ outputs what $\qD$ outputs.
\end{enumerate}
 By the definition of $W$ described in~\cref{fig:func_wrap_W}, if we decrypt $\fe.\ct^\ast$ by $\fe.\sk_{W,i}$, we obtain
\begin{itemize}
\item $\SKL.\Enc(\skl.\ek_i,f(x_0);\prf_{\prfkey}(\skl.\ek_i))$ if $\fe.\ct^\ast$ is generated from $X_0^\ast$,
\item $z_i = \SKFE.\Dec(\skfe.\sk_i,\skfe.\ct^\ast)$ if $\fe.\ct^\ast$ is generated from $X_1^\ast$ since $\ske.\ct_i$ is a ciphertext of $\skfe.\sk_i$, where $\skfe.\ct^\ast=\SKFE.\Enc(\skfe.\msk,(x_0^\ast,\bot,\prfkey,0,\bot))$. By the correctness of $\SKFE$ and the definition of $T[f_i,\skl.\ek_i]$, it holds that $z_i = \SKL.\Enc(\skl.\ek_i,f(x_0^\ast);\prf_{\prfkey}( \skl.\ek_i))$.
\end{itemize}
That is, for all $i\in [q]$, it holds that $W[f_i,\skl.\ek_i,\ske.\ct_i](X_0^\ast) = W[f_i,\skl.\ek_i,\ske.\ct_i](X_1^\ast)$, and $\qB$ is a valid adversary of PKFE.

It is easy to see that if $\fe.\ct^\ast$ is an encryption of $X_0^\ast$ and $X_1^\ast$, $\qB$ perfectly simulates $\hybi{1}$ and $\hybi{2}$, respectively. This completes the proof.
\end{proof}

\begin{lemma}\label{lem:pkfe_hyb_two_three}
If $\SKFE$ is adaptively single-ciphertext function private, it holds that $\abs{\Pr[\hybi{2}=1] - \Pr[\hybi{3} =1]} = \negl(\secp)$.
\end{lemma}
\begin{proof}
We construct an adversary $\qB$ for $\SKFE$ by using the distinguisher $\qD$ for these two games.
\begin{enumerate}
\item $\qB$ generates $(\fe.\pk,\fe.\msk)\gets \FE.\Setup(1^\secp)$ and $\ske.\sk \chosen \zo{\secp}$, and sends $\pk \seteq \fe.\pk$ to $\qD$.
\item When $\qD$ sends $f_i$ to $\Oracle{\qKG}$, 
$\qB$ generates $(\skl.\ek_i,\skl.\qsk_i,\skl.\vk_i) \gets \SKL.\qKG(1^\secp)$, sends a key query $(F_{0,i},F_{1,i})\seteq (T[f_i,\skl.\ek_i],\Vhyb[f_i,\skl.\ek_i,i])$ to its challenger, and receives $\skfe.\sk_i$. $\qB$ also generates $\ske.\ct_i \gets \SKE.\Enc(\ske.\sk, \skfe.\sk_i)$ and $\fe.\sk_{W,i} \gets \FE.\KG(\fe.\msk,W[f_i,\skl.\ek_i,\ske.\ct_i])$. $\qB$ returns $\qfsk_{i} \seteq (\fe.\sk_{W,i},\skl.\qsk_i)$ to $\qD$ and adds $(f_i,\skl.\vk_i,\bot)$ to $\List{\qKG}$.
\item When $\qD$ sends $(f_i,\qfsk_i^\prime)$ to $\Oracle{\qVrfy}$, $\qB$ finds an entry $(f_i,\skl.\vk_i,V_i)$ from $\List{\qKG}$ and parses $\qfsk_i^\prime = (\fe.\sk_i^\prime,\skl.\qsk_i^\prime)$. $\qB$ returns $d \seteq \SKL.\qVrfy(\skl.\vk_i,\skl.\qsk_i^\prime)$. If $V_i = \top$, $\qB$ does not update the entry. Otherwise, $\qB$ updates the entry by setting $V_i \seteq d$.
\item When $\qD$ sends $(x_0^\ast,x_1^\ast)$, $\qB$ generates $\prfkey\gets \PRF.\Gen(1^\secp)$, sets $X_0^\ast \seteq (x_0^\ast,\bot,\prfkey,0,\bot)$ and $X_1^\ast \seteq (x_0^\ast,x_1^\ast,\prfkey,0,\bot)$, sends an encryption query $(X_0^\ast,X_1^\ast)$ to its challenger, and receives $\skfe.\ct^\ast$. $\qB$ also generates $\fe.\ct^\ast \gets \FE.\Enc(\fe.\pk,(\skfe.\ct^\ast,\ske.\sk,\bot))$ and passes $\ct^\ast \seteq \fe.\ct^\ast$ to $\qD$.
\item $\qB$ outputs what $\qD$ outputs.
\end{enumerate}
Since $i \in [q]$, it holds that $\Vhyb[f_i,\skl.\ek_i,i](x_0^\ast,x_1^\ast,\prfkey,0,\bot) = T[f_i,\skl.\ek_i](x_0^\ast,\bot,\prfkey,0,\bot)$ for all $i\in [q]$. That is, $F_{0,i}(X_0^\ast)=F_{1,i}(X_1^\ast)$ for all $i \in [q]$ and $\qB$ is an valid adversary for SKFE.

If $\skfe.\ct^\ast$ is an encryption of $X_0^\ast$ and $\skfe.\sk_i$ is a functional decryption key for $F_{0,i}$, $\qB$ perfectly simulate $\hybi{2}$. If $\skfe.\ct^\ast$ is an encryption of $X_1^\ast$ and $\skfe.\sk_i$ is a functional decryption key for $F_{1,i}$, $\qB$ perfectly simulate $\hybi{3}$. This completes the proof.
\end{proof}

\begin{lemma}\label{lem:pkfe_hyb_three_four}
If $\SKFE$ is adaptively single-ciphertext function private, it holds that $\abs{\Pr[\hybij{3}{q}=1] - \Pr[\hybi{4} =1]} = \negl(\secp)$.
\end{lemma}
\begin{proof}
We construct an adversary $\qB$ for $\SKFE$ by using the distinguisher $\qD$ for these two games.
\begin{enumerate}
\item $\qB$ generates $(\fe.\pk,\fe.\msk)\gets \FE.\Setup(1^\secp)$ and $\ske.\sk \chosen \zo{\secp}$, and sends $\pk \seteq \fe.\pk$ to $\qD$.
\item When $\qD$ sends $f_i$ to $\Oracle{\qKG}$, 
$\qB$ generates $(\skl.\ek_i,\skl.\qsk_i,\skl.\vk_i) \gets \SKL.\qKG(1^\secp)$, sends a key query $(F_{0,i},F_{1,i})\seteq (\Vhyb[f_i,\skl.\ek_i,i],T[f_i,\skl.\ek_i])$ to its challenger, and receives $\skfe.\sk_i$. $\qB$ also generates $\ske.\ct_i \gets \SKE.\Enc(\ske.\sk, \skfe.\sk_i)$ and $\fe.\sk_{W,i} \gets \FE.\KG(\fe.\msk,W[f_i,\skl.\ek_i,\ske.\ct_i])$. $\qB$ returns $\qfsk_{i} \seteq (\fe.\sk_{W,i},\skl.\qsk_i)$ to $\qD$ and adds $(f_i,\skl.\vk_i,\bot)$ to $\List{\qKG}$.
\item When $\qD$ sends $(f_i,\qfsk_i^\prime)$ to $\Oracle{\qVrfy}$, $\qB$ finds an entry $(f_i,\skl.\vk_i,V_i)$ from $\List{\qKG}$ and parses $\qfsk_i^\prime = (\fe.\sk_i^\prime,\skl.\qsk_i^\prime)$. $\qB$ returns $d \seteq \SKL.\qVrfy(\skl.\vk_i,\skl.\qsk_i^\prime)$. If $V_i = \top$, $\qB$ does not update the entry. Otherwise, $\qB$ updates the entry by setting $V_i \seteq d$.
\item When $\qD$ sends $(x_0^\ast,x_1^\ast)$, $\qB$ generates $\prfkey\gets \PRF.\Gen(1^\secp)$, sets $X_0^\ast \seteq (x_0^\ast,x_1^\ast,\prfkey,q,\bot)$ and $X_1^\ast \seteq (x_1^\ast,\bot,\prfkey,0,\bot)$, sends an encryption query $(X_0^\ast,X_1^\ast)$ to its challenger, and receives $\skfe.\ct^\ast$. $\qB$ also generates $\fe.\ct^\ast \gets \FE.\Enc(\fe.\pk,(\skfe.\ct^\ast,\ske.\sk,\bot))$ and passes $\ct^\ast \seteq \fe.\ct^\ast$ to $\qD$.
\item $\qB$ outputs what $\qD$ outputs.
\end{enumerate}
By the definition of $\Vhyb$ and $T$, it holds that for all $i\in [q]$,
\begin{align}
\Vhyb[f_i,\skl.\ek_i,i](x_0^\ast,x_1^\ast,\prfkey,q,\bot) & = \SKL.\Enc(\skl.\ek_i,f_i(x_1^\ast);\prf_{\prfkey}(\skl.\ek_i)) \\
&= T[f_i,\skl.\ek_i](x_1^\ast,\bot,\prfkey,0,\bot).
\end{align}
That is, $F_{0,i}(X_0^\ast)=F_{1,i}(X_1^\ast)$ for all $i \in [q]$ and $\qB$ is an valid adversary for SKFE.

If $\skfe.\ct^\ast$ is an encryption of $X_0^\ast$ and $\skfe.\sk_i$ is a functional decryption key for $F_{0,i}$, $\qB$ perfectly simulate $\hybij{3}{q}$. If $\skfe.\ct^\ast$ is an encryption of $X_1^\ast$ and $\skfe.\sk_i$ is a functional decryption key for $F_{1,i}$, $\qB$ perfectly simulate $\hybi{4}$. This completes the proof.
\end{proof}

\begin{lemma}\label{lem:pkfe_Ghyb_zero_one}
If $\SKFE$ is adaptively single-ciphertext function private, it holds that $\abs{\Pr[\sfG_0 =1] - \Pr[\sfG_1 =1]} = \negl(\secp)$.
\end{lemma}
\begin{proof}
We construct an adversary $\qB$ for $\SKFE$ by using the distinguisher $\qD$ for these two games.
\begin{enumerate}
\item $\qB$ generates $(\fe.\pk,\fe.\msk)\gets \FE.\Setup(1^\secp)$ and $\ske.\sk \chosen \zo{\secp}$, and sends $\pk \seteq \fe.\pk$ to $\qD$.
\item When $\qD$ sends $f_i$ to $\Oracle{\qKG}$, 
$\qB$ generates $(\skl.\ek_i,\skl.\qsk_i,\skl.\vk_i) \gets \SKL.\qKG(1^\secp)$, sends a key query $(F_{0,i},F_{1,i})\seteq (\Vhyb[f_i,\skl.\ek_i,i],\Vembed[f_i,\skl.\ek_i,i])$ to its challenger, and receives $\skfe.\sk_i$. $\qB$ also generates $\ske.\ct_i \gets \SKE.\Enc(\ske.\sk, \skfe.\sk_i)$ and $\fe.\sk_{W,i} \gets \FE.\KG(\fe.\msk,W[f_i,\skl.\ek_i,\ske.\ct_i])$. $\qB$ returns $\qfsk_{i} \seteq (\fe.\sk_{W,i},\skl.\qsk_i)$ to $\qD$ and adds $(f_i,\skl.\vk_i,\bot)$ to $\List{\qKG}$.
\item When $\qD$ sends $(f_i,\qfsk_i^\prime)$ to $\Oracle{\qVrfy}$, $\qB$ finds an entry $(f_i,\skl.\vk_i,V_i)$ from $\List{\qKG}$ and parses $\qfsk_i^\prime = (\fe.\sk_i^\prime,\skl.\qsk_i^\prime)$. $\qB$ returns $d \seteq \SKL.\qVrfy(\skl.\vk_i,\skl.\qsk_i^\prime)$. If $V_i = \top$, $\qB$ does not update the entry. Otherwise, $\qB$ updates the entry by setting $V_i \seteq d$.
\item When $\qD$ sends $(x_0^\ast,x_1^\ast)$, $\qB$ generates $\prfkey\gets \PRF.\Gen(1^\secp)$ and $\skl.\ct^\ast \gets \SKL.\Enc(\skl.\ek_j,f_j(x_0^\ast);\prf_{\prfkey}(\skl.\ek_j))$, sets $X_0^\ast \seteq (x_0^\ast,x_1^\ast,\prfkey,j-1,\bot)$ and $X_1^\ast \seteq (x_0^\ast,x_1^\ast,\prfkey,j,\skl.\ct^\ast)$, sends an encryption query $(X_0^\ast,X_1^\ast)$ to its challenger, and receives $\skfe.\ct^\ast$. $\qB$ generates $\fe.\ct^\ast \gets \FE.\Enc(\fe.\pk,(\skfe.\ct^\ast,\ske.\sk,\bot))$ and passes $\ct^\ast \seteq \fe.\ct^\ast$ to $\qD$.
\item $\qB$ outputs what $\qD$ outputs.
\end{enumerate}
By the definitions of $\Vhyb$ and $\Vembed$, it holds that
\begin{align}
\Vhyb[f_i,\skl.\ek_i,i](x_0^\ast,x_1^\ast,\prfkey,j-1,\bot) & = \SKL.\Enc(\skl.\ek_i,f_i(x_0^\ast);\prf_{\prfkey}(\skl.\ek_i))\\
&= \Vembed[f_i,\skl.\ek_i,i](x_0^\ast,x_1^\ast,\prfkey,j,\skl.\ct^\ast)
\end{align} for all $i\in [j,q]$ since $\skl.\ct^\ast$ is an encryption of $f_j(x_0)$.
it also holds that
\begin{align}
\Vhyb[f_i,\skl.\ek_i,i](x_0^\ast,x_1^\ast,\prfkey,j-1,\bot) & = \SKL.\Enc(\skl.\ek_i,f_i(x_1^\ast);\prf_{\prfkey}(\skl.\ek_i))\\
&= \Vembed[f_i,\skl.\ek_i,i](x_0^\ast,x_1^\ast,\prfkey,j,\skl.\ct^\ast)
\end{align} for all $i\in [1,j-1]$. Hence, for all $i\in [q]$, it holds that $F_{0,i}(X_0^\ast)=F_{1,i}(X_1\ast)$ and $\qB$ is an valid adversary for SKFE. If $\skfe.\ct^\ast$ is an encryption of $X_0^\ast$ and $\skfe.\sk_i$ is a functional decryption key for $F_{0,i}$, $\qB$ perfectly simulate $\sfG_{0}$. If $\skfe.\ct^\ast$ is an encryption of $X_1^\ast$ and $\skfe.\sk_i$ is a functional decryption key for $F_{1,i}$, $\qB$ perfectly simulate $\sfG_{1}$. This completes the proof.
\end{proof}

\begin{lemma}\label{lem:pkfe_Ghyb_three_four}
If $\SKL$ IND-KLA, it holds that $\abs{\Pr[\sfG_3 =1] - \Pr[\sfG_4 =1]} = \negl(\secp)$.
\end{lemma}
\begin{proof}
We focus on the case where the adversary returns a valid $\qfsk_j = (\fe.\sk_{W,j},\skl.\qsk_j)$, which is the answer to the $j$-th key query, since $f_j(x_0^\ast)= f_j(x_1^\ast)$ must hold if $\qfsk_j$ is not returned.
Hence $f_j(x_0^\ast) \ne f_j(x_1^\ast)$ is allowed in this case.

We construct an adversary $\qB$ for $\SKL$ by using the distinguisher $\qD$ for these two games.
\begin{enumerate}
\item $\qB$ is given $(\skl.\ek^\ast,\skl.\qsk^\ast)$ and sets $(\skl.\ek_j,\skl.\ek_j) \seteq (\skl.\ek^\ast,\skl.\qsk^\ast)$.
\item $\qB$ generates $(\fe.\pk,\fe.\msk)\gets \FE.\Setup(1^\secp)$, $\skfe.\msk \gets \SKFE.\Setup(1^\secp)$, and $\ske.\sk \chosen \zo{\secp}$, and sends $\pk \seteq \fe.\pk$ to $\qD$.
\item When $\qD$ sends the $i$-th query $f_i$ to $\Oracle{\qKG}$, 
if $i\ne j$, $\qB$ generates $(\skl.\ek_i,\skl.\qsk_i,\skl.\vk_i) \gets \SKL.\qKG(1^\secp)$. For all $i \in [q]$, $\qB$ generates $\skfe.\sk_i \gets \SKFE.\KG(\skfe.\msk,\Vembed[f_i,\skl.\ek_i,i])$, $\ske.\ct_i \gets \SKE.\Enc(\ske.\sk, \skfe.\sk_i)$, and $\fe.\sk_{W,i} \gets \FE.\KG(\fe.\msk,\allowbreak W[f_i,\skl.\ek_i,\ske.\ct_i])$, and returns $\qfsk_{i} \seteq (\fe.\sk_{W,i},\skl.\qsk_i)$ to $\qD$. Note that $\skl.\qsk_j =\skl.\qsk$ is given from the challenger. If $i\ne j$, $\qB$ adds $(f_i,\skl.\vk_i,\bot)$ to $\List{\qKG}$.
If $i=j$, $\qB$ adds $(f_j,\bot,\bot)$ to $\List{\qKG}$.
\item When $\qD$ sends $(f_i,\qfsk_i^\prime)$ to $\Oracle{\qVrfy}$, $\qB$ finds an entry $(f_i,\skl.\vk_i,V_i)$ from $\List{\qKG}$ and parses $\qfsk_i^\prime = (\fe.\sk_i^\prime,\skl.\qsk_i^\prime)$.
\begin{itemize}
\item If $f_i \ne f_j$, $\qB$ returns $d \seteq \SKL.\qVrfy(\skl.\vk_i,\skl.\qsk_i^\prime)$ since $\skl.\vk_i \ne \bot$. If $V_i =\top$ $\qB$ does not update the entry. Otherwise, $\qB$ updates the entry by setting $V_i \seteq d$.
\item Else if $f_i = f_j$, $\qB$ sends $\skl.\qsk_j^\prime$ to its challenger ($\Oracle{\SKL.\qVrfy}$ of IND-KLA), receives the result $d_j$, and passes $d_j$ to $\qD$. If $V_j =\top$, $\qB$ does not update the entry. Otherwise, $\qB$ updates the entry by setting $V_j \seteq d_j$.
\end{itemize}
\item When $\qD$ sends $(x_0^\ast,x_1^\ast)$, $\qB$ generates $\prfkey\gets \PRF.\Gen(1^\secp)$ and $\prfkey_{\ne \skl.\ek_j} = \Puncture(\prfkey,\skl.\ek_j)$, sends $(x_0^\ast,x_1^\ast)$ to its challenger, and receives $\skl.\ct^\ast \gets \SKL.\Enc(\skl.\ek_j,f_j(x_\coin^\ast))$. $\qB$ generates $\skfe.\ct^\ast \gets  \SKFE.\Enc(\skfe.\msk,\allowbreak (x_0^\ast,x_1^\ast,\prfkey_{\ne \skl.\ek_j},j,\skl.\ct^\ast))$ and $\fe.\ct^\ast \gets \FE.\Enc(\fe.\pk,(\skfe.\ct^\ast,\ske.\sk,\bot))$ and passes $\ct^\ast \seteq \fe.\ct^\ast$ to $\qD$.
\item $\qB$ outputs what $\qD$ outputs.
\end{enumerate}
It is easy to see that $\qB$ perfectly simulates $\sfG_3$ and $\sfG_4$ if $\coin=0$ and $\coin=1$, respectively.
This completes the proof.
\end{proof}

\else

	\setcounter{tocdepth}{2}
	\tableofcontents

	\fi
	\else
\fi

\end{document}